	\documentclass[
		letterpaper,american,USenglish,authorcolumns,
		numberwithinsect,
		anonymous=false%
	]{lipics-v2021}
	\def\docclass{lipics}
	\def\version{arxiv}

	\def\draftmode{false} 

\usepackage[utf8]{inputenc}
\makeatletter
\usepackage{xifthen}

\newcommand\iflipics[2]{\ifthenelse{\equal{\docclass}{lipics}}{#1}{#2}}
\newcommand\ifkoma[2]{\ifthenelse{\equal{\docclass}{koma}}{#1}{#2}}
\newcommand\ifieee[2]{\ifthenelse{\equal{\docclass}{ieee}}{#1}{#2}}
\newcommand\ifsiam[2]{\ifthenelse{\equal{\docclass}{siam}}{#1}{#2}}
\newcommand\ifmysiam[2]{\ifthenelse{\equal{\docclass}{my-siam}}{#1}{#2}}
\newcommand\ifacm[2]{\ifthenelse{\equal{\docclass}{acm}}{#1}{#2}}
\newcommand\ifdcc[2]{\ifthenelse{\equal{\docclass}{dcc}}{#1}{#2}}
\ifthenelse{ \equal{\docclass}{lipics} \OR \equal{\docclass}{koma} \OR \equal{\docclass}{ieee} \OR \equal{\docclass}{siam} \OR \equal{\docclass}{my-siam} \OR \equal{\docclass}{acm} \OR \equal{\docclass}{dcc} }{
	\PackageInfo{paper}{Building paper with docclass = \docclass} 
}{
	\PackageWarning{paper}{docclass = "\docclass", but must be one of "lipics", "koma", "ieee", "siam", "my-siam", "acm", "dcc"}
}

\newcommand\ifmanuscript[2]{\ifthenelse{\equal{\version}{manuscript}}{#1}{#2}}
\newcommand\ifarxiv[2]{\ifthenelse{\equal{\version}{arxiv}}{#1}{#2}}
\newcommand\ifsubmission[2]{\ifthenelse{\equal{\version}{submission}}{#1}{#2}}
\newcommand\ifproceedings[2]{\ifthenelse{\equal{\version}{proceedings}}{#1}{#2}}
\ifthenelse{ 
	\equal{\version}{manuscript} 
	\OR \equal{\version}{arxiv} 
	\OR \equal{\version}{submission} 
	\OR \equal{\version}{proceedings} 
}{
	\PackageInfo{paper}{Building paper version = \version} 
}{
	\PackageWarning{paper}{version = "\version", but must be one of "manuscript", "arxiv", "submission", "proceedings"}
}

\newcommand\ifdraft[2]{\ifthenelse{\equal{\draftmode}{true}}{#1}{#2}}
\ifthenelse{ \equal{\draftmode}{true} \OR \equal{\draftmode}{false} }{
	\PackageInfo{paper}{Building paper with draftmode = \draftmode} 
}{
	\PackageWarning{paper}{draftmode = "\draftmode", but must be "true" or "false"}
}

\usepackage[T1]{fontenc}
\ifsiam{
	\usepackage{lmodern}
	\usepackage{slantsc}
}{}
\ifmysiam{
	\usepackage{lmodern}
	\usepackage{slantsc}
}{}
\ifkoma{
	\usepackage{lmodern}
	\usepackage{slantsc}
}{}
\iflipics{
	\usepackage{lmodern}
	\usepackage{slantsc}
}{}
\ifdcc{
	\usepackage{lmodern}
	\usepackage{slantsc}
}{}

\usepackage{babel}
\input{ushyphex.tex} 

\usepackage{array,multicol}
\ifieee{
	\usepackage[cmex10]{amsmath,mathtools}
	\usepackage{amsfonts,amssymb}
}{}
\ifkoma{
	\usepackage{amsmath,amsfonts,amssymb,mathtools}
}{}
\iflipics{
	\usepackage{amsmath,amsfonts,amssymb,mathtools}
}{}
\ifsiam{
	\usepackage{amsmath,amsfonts,amssymb,mathtools}
}{}
\ifmysiam{
	\usepackage{amsmath,amsfonts,amssymb,mathtools}
}{}
\ifacm{
	\usepackage{mathtools}
}{}
\ifdcc{
	\usepackage{amsmath,amsfonts,amssymb,mathtools}
}{}

\usepackage{mleftright}\mleftright 
\usepackage{relsize,xspace,booktabs,adjustbox,needspace,pbox,relsize}
\ifieee{
	
	\usepackage{enumitem}
}{
	\usepackage{enumitem}
}
\usepackage{graphicx}

\ifacm{}{
	\usepackage{colonequals}
}

\ifdraft{
}{%
}
\usepackage{wref}
\iflipics{
	\usepackage[bibtex]{url-doi-arxiv}
}{
	\usepackage[bibtex-alpha]{url-doi-arxiv}
}

\ifsiam{
	\usepackage{ltexpprt}
}{}

\newdimen\makeboxdimen

\newcommand\plaincenter[1]{%
	\mbox{}\hfill#1\hfill\mbox{}%
}

\ifthenelse{\equal{\docclass}{koma} \OR \equal{\docclass}{my-siam}}{
	\setlength\parindent{1.5em}
	\usepackage[headsepline]{scrlayer-scrpage}
	\pagestyle{scrheadings}
	\clearscrheadfoot
	\AtBeginDocument{%
		\automark[section]{}%
	}
	\ohead{\pagemark}
	\rehead{\mytitle}
	\lohead{\headmark}
	\addtokomafont{caption}{\sffamily\small}
	\addtokomafont{captionlabel}{\sffamily\textbf}
	\setcapmargin{2em}
}{}
\ifthenelse{\equal{\docclass}{my-siam}}{
	\setcapmargin{1em}
	\setcapindent{0em}
}{}
\ifmysiam{
	\setlength\parskip{0pt}
	\RedeclareSectionCommand[
		beforeskip=-1.25\baselineskip,
		afterskip=0.75\baselineskip,
	]{section}
	\RedeclareSectionCommand[
		beforeskip=-1\baselineskip,
		afterskip=-1.5em,
	]{subsection}
	\RedeclareSectionCommand[
		beforeskip=-1\baselineskip,
		afterskip=-1.5em,
	]{subsubsection}
	\RedeclareSectionCommand[
		beforeskip=-.25\baselineskip,
		indent=1.5em,
		afterskip=-1em,
	]{paragraph}
}{}
\ifdcc{
	\ifproceedings{}{
		\pagestyle{plain}
		\setlength{\footskip}{5ex}
	}
}{}

\AtBeginDocument{%
	\let\mytitle\@title%
}

\let\oldthebibliography\thebibliography
\renewcommand\thebibliography[1]{%
	\oldthebibliography{#1}%
	\pdfbookmark[1]{References}{}%
}

\usepackage{lscape} 

\ifkoma{
	\usepackage{float}
	\floatstyle{plain}
	\usepackage{newfloat}
	\DeclareFloatingEnvironment[%
			name=Algorithm,%
			placement=thb,%
		]{algorithm}
}{}
\iflipics{
	\usepackage{newfloat}
	\DeclareFloatingEnvironment[%
			name=Algorithm,%
			placement=thb,%
		]{algorithm}
}{}
\ifmysiam{

	\usepackage{newfloat}
	\DeclareFloatingEnvironment[%
			name=Algorithm,%
			placement=thb,%

		]{algorithm}
}

\ifmysiam{
	\setcounter{topnumber}{3}
	\setcounter{bottomnumber}{3}
	\setcounter{totalnumber}{3}     
	\setcounter{dbltopnumber}{3}    

}{}
\ifsiam{

	\setcounter{topnumber}{3}
	\setcounter{bottomnumber}{3}
	\setcounter{totalnumber}{3}     
	\setcounter{dbltopnumber}{3}    

}{}

\usepackage{dcolumn}

\usepackage{wclrscode}

\usepackage{textcomp} 
\usepackage{listings}

\usepackage{tikz}

\usetikzlibrary{positioning,arrows.meta,fit}
\usetikzlibrary{backgrounds,calc,trees,graphs}
\usetikzlibrary{shapes.geometric,shapes.misc}
\usetikzlibrary{tikzmark}

\pgfdeclarelayer{background}
\pgfsetlayers{background,main}

\usepackage[outline]{contour}
\contourlength{.05em}

\iflipics{
	\newtheorem{fact}[theorem]{Fact}
	\newtheorem{openproblem}[theorem]{Open Problem}
	
	\newenvironment{proofof}[1]{%
		\begin{proof}[{{Proof of #1{}}}]%
	}{%
		\end{proof}%
	}
}{}
\ifacm{
	\AtEndPreamble{
		\theoremstyle{acmdefinition}
		\newtheorem{remark}[theorem]{Remark}
		\newtheorem{fact}[theorem]{Fact}
	}
	
	\newenvironment{proofof}[1]{%
		\begin{proof}[{{Proof of #1{}}}]%
	}{%
		\end{proof}%
	}
}{}
\ifsiam{
	\newtheorem{remark}{Remark}
	\newenvironment{proofof}[1]{%
			\begin{proof}[{{#1{}}}]%
		}{%
			\end{proof}%
		}
}{}
\ifthenelse{\equal{\docclass}{lipics} \OR \equal{\docclass}{siam} \OR \equal{\docclass}{acm}}{}{
	\usepackage[amsmath,hyperref,thmmarks]{ntheorem}
	
	\theorembodyfont{\slshape}
	\theoremseparator{:}
	\newtheoremstyle{proofstyle}%
	  {\item[\theorem@headerfont\hskip\labelsep ##1\theorem@separator]}%
	  {\item[\theorem@headerfont\hskip\labelsep ##3\theorem@separator]}
	
	\theorempreskip{\topsep} 

	\theoremsymbol{\adjustbox{scale=.8}{$\triangleleft\mkern-1mu$}}
	
	\newtheorem{theorem}{Theorem}[section]
	
	\theoremstyle{plain}
	\theorempreskip{\topsep}
	
	\newtheorem{proposition}[theorem]{Proposition}
	\newtheorem{lemma}[theorem]{Lemma}
	\newtheorem{conjecture}[theorem]{Conjecture}
	\newtheorem{openproblem}[theorem]{Open Problem}
	\newtheorem{corollary}[theorem]{Corollary}
	\newtheorem{definition}[theorem]{Definition}
	
	\theoremstyle{plain}
	\theorembodyfont{\upshape}

	\newtheorem{remark}[theorem]{Remark}

	\theoremsymbol{\raisebox{-.25ex}{$\Box$}}
	\qedsymbol{\raisebox{-.25ex}{$\Box$}}
	
	\theoremstyle{proofstyle}
	\newtheorem{proof}{Proof}
	\newenvironment{proofof}[1]{%
		\begin{proof}[{{Proof of #1{}}}]%
	}{%
		\end{proof}%
	}

}

\iflipics{
	\newenvironment{thmenumerate}[2][]{%
		\begin{enumerate}[
			label={\textsf{\textbf{\color{darkgray}{\makebox[\widthof{(a)}][c]{\textup{(\alph*)}}}}}},
			ref={\ref{#2}\kern.1em--\kern.1em(\alph*)},
			itemsep=0pt,
			topsep=.5ex,
			leftmargin=1.75em,
			#1
		]%
	}{%
		\end{enumerate}%
	}
}{
	\newenvironment{thmenumerate}[2][]{%
		\begin{enumerate}[
			label={\makebox[\widthof{(a)}][c]{\textup{(\alph*)}}},
			ref={\ref{#2}\kern.1em--\kern.1em(\alph*)},
			itemsep=0pt,
			#1
		]%
	}{%
		\end{enumerate}%
	}
}

\newcommand*\ie{\mbox{i.\hspace{.2ex}e.}}
\newcommand*\eg{\mbox{e.\hspace{.2ex}g.}}

\newcommand\R{\mathbb R}
\newcommand\N{\mathbb N}

\usepackage{fixmath}

\newcommand{\ESymbol}{\mathbb{E}}

\newcommand{\ProbSymbol}{\ensuremath{\mathbb{P}}}

\DeclarePairedDelimiterXPP\Prob[1]{\ProbSymbol}[]{}{%
	#1%
}
\DeclarePairedDelimiterXPP\E[1]{\ESymbol}[]{}{%
	#1%
}
\DeclarePairedDelimiterXPP\Eover[2]{\ESymbol_{#1}}[]{}{%
	#2%
}
\DeclarePairedDelimiterXPP\ProbIn[2]{\ProbSymbol_{#1}}[]{}{%
	#2%
}
\providecommand{\Prob}{} 
\providecommand{\ProbIn}{} 
\providecommand{\E}{} 
\providecommand{\Eover}{} 

\newcommand{\surroundedmath}[3]{
	\mathchoice{
		#1{#2{#3}#2}%
	}{
		#1{#3}%
	}{
		#1{#3}%
	}{
		#1{#3}%
	}%
}
\newcommand\rel[1]{\surroundedmath{\mathrel}{\:}{#1}}
\newcommand\wrel[1]{\surroundedmath{\mathrel}{\;}{#1}}
\newcommand\wwrel[1]{\surroundedmath{\mathrel}{\;\;}{#1}}

\iflipics{}{
	\makeatletter
	\let\oldalign\align
	\let\endoldalign\endalign
	\renewenvironment{align}{%
		\begingroup%
		\let\oldhalign\halign
		\def\halign{%
			\let\oldbreak\\%
			\def\nonnumberbreak{\nonumber\oldbreak*}%
			\def\\{%
				\@ifstar{\nonnumberbreak}{\oldbreak}%
			}%
			\oldhalign%
		}
		\oldalign%
	}{%
		\endoldalign%
		\endgroup%
	}
}
\newcommand*\numberthis[1][]{\stepcounter{equation}\tag{\theequation}}

\allowdisplaybreaks[3]

\newcommand\splitaftercomma[1]{%
  \begingroup
  \begingroup\lccode`~=`, \lowercase{\endgroup
    \edef~{\mathchar\the\mathcode`, \penalty0 \noexpand\hspace{0pt plus .25em}}%
  }\mathcode`,="8000 #1%
  \endgroup
}

\def\mydots{\xleaders\hbox to.5em{\hfill.\hfill}\hfill}
\newlength\tmpLenNotations

\ifdraft{%
	\iflipics{
		\usepackage{lineno} 
	}{
		\usepackage[switch]{lineno} 
	}
	\linenumbers
	\overfullrule=6mm
	
	\usepackage[color,notref,notcite]{showkeys}
	\definecolor{refkey}{gray}{.99}
	\colorlet{labelkey}{green!60!black!60}
	
	\usepackage[inline,nolabel]{showlabels}

	\showlabels{cite}
	\showlabels{citealt}
	\showlabels{citealp}
	\showlabels{citet}
	\showlabels{Citet}
	\showlabels{citep}
	\showlabels{citeauthor}
	\showlabels{Citeauthor}
	\showlabels{citefullauthor}
	\showlabels{citeyear}
	\showlabels{citeyearpar}
	\showlabels{wref}
	\showlabels{wpref}
	\showlabels{wtpref}
	\showlabels{wildref}
	\showlabels{wildpageref}
	\showlabels{wildtpageref}
}{}

\iflipics{
	\ifmanuscript{\hideLIPIcs}{}
	\ifarxiv{\hideLIPIcs}{}
	\ifsubmission{}{\nolinenumbers}
}{}

\ifdraft{}{%
	\usepackage{microtype}
}

\hypersetup{
	final,
	unicode=true, 
	bookmarks=true,
	bookmarksnumbered=true,
	bookmarksdepth=3,
	bookmarksopen=true,
	breaklinks=true,
	hidelinks,
}

\newsavebox\tmpbox

\iflipics{
	\let\oldparagraph\paragraph
	\renewcommand\paragraph[1]{%
		\oldparagraph*{#1}
	}
}{
	\let\oldparagraph\paragraph
	\renewcommand\paragraph[1]{%
		\oldparagraph{#1.}
	}
}

\ifmysiam{
	\let\oldsubsection\subsection
	\renewcommand\subsection[1]{%
		\oldsubsection{#1.}%
	}
	\let\oldsubsubsection\subsubsection
	\renewcommand\subsubsection[1]{%
		\oldsubsubsection{#1.}%
	}
}{}
\ifsiam{
	\let\oldsubsection\subsection
	\renewcommand\subsection[1]{%
		\oldsubsection{#1.}%
	}
	\let\oldsubsubsection\subsubsection
	\renewcommand\subsubsection[1]{%
		\oldsubsubsection{#1.}%
	}
}{}

\let\epsilon\varepsilon

\def\myacknowledgements{}
\ifkoma{
	\newcommand\acknowledgements[1]{\def\myacknowledgements{\paragraph{Acknowledgements}#1}}
}{}
\ifieee{
	\newcommand\acknowledgements[1]{\def\myacknowledgements{\section*{Acknowledgement}#1}}
}{}
\ifsiam{
	\newcommand\acknowledgements[1]{\def\myacknowledgements{\section*{Acknowledgement}#1}}
}{}
\ifmysiam{
	\newcommand\acknowledgements[1]{\def\myacknowledgements{\section*{Acknowledgement}#1}}
}{}
\ifdcc{
	\newcommand\acknowledgements[1]{\def\myacknowledgements{\section*{Acknowledgement}#1}}
}{}
\ifacm{
	\newcommand\acknowledgements[1]{\def\myacknowledgements{
		\section*{Acknowledgement}#1%
	}}
}{}

\newcommand\coloredunderline[3][1pt]{
	\begin{tikzpicture}[baseline=(slot.base)]
		\node[inner sep=0pt] (slot) {#3} ;
		\path (slot.south west) + (0,-#1) coordinate (ll) ;
		\path (slot.south east) + (0,-#1) coordinate (lr) ;
		\draw[line width=#1,#2] (ll) -- (lr) ;
	\end{tikzpicture}%
	\xspace%
}

\newcommand\mathcoloredunderline[3][1pt]{
	\mathchoice{%
		\coloredunderline[#1]{#2}{\ensuremath{\displaystyle #3}}%
	}{%
		\coloredunderline[#1]{#2}{\ensuremath{\textstyle #3}}%
	}{%
		\coloredunderline[#1]{#2}{\ensuremath{\scriptstyle #3}}%
	}{%
		\coloredunderline[#1]{#2}{\ensuremath{\scriptscriptstyle #3}}%
	}%
}

\definecolor{applegreen}{rgb}{0.55, 0.71, 0.0}

\tikzset{graph node/.style = {circle, thick, draw}}
\tikzset{edge descriptor/.style = {text=black, fill=white,circle}}

\tikzset{C1/.style = {applegreen}}
\tikzset{C2/.style = {blue}}
\tikzset{C3/.style = {red}}

\tikzset{std/.style = {line width=3pt, draw}}
\tikzset{2 colour 1/.style = {line width=3pt, dash pattern= on 9pt off 11pt, draw}}
\tikzset{2 colour 2/.style = {line width=3pt, dash pattern= on 9pt off 11pt,dash phase=10pt, draw}}

\usepackage{fontawesome5}

\makeatother

\usepackage{hyperref}

\ifacm{
	\title{Polyamorous Scheduling}
}{
	\title{Polyamorous Scheduling}
}

\iflipics{
	\author{Leszek G\k asieniec}{University of Liverpool, UK}{l.a.gasieniec@liverpool.ac.uk}{https://orcid.org/0000-0003-1809-9814}{}
	\author{Benjamin Smith}{University of Liverpool, UK}{b.m.smith@liverpool.ac.uk}{https://orcid.org/0000-0003-2306-3461}{}
	\author{Sebastian Wild}{University of Liverpool, UK}{wild@liverpool.ac.uk}{https://orcid.org/0000-0002-6061-9177}{}
	\authorrunning{L. G\k asieniec, B. Smith, and S. Wild}
	\Copyright{Leszek G\k asieniec, Benjamin Smith, and Sebastian Wild}
	
		\ccsdesc[500]{Theory of computation~Scheduling algorithms}%
		\ccsdesc[500]{Theory of computation~Problems, reductions and completeness}%
	\keywords{periodic scheduling, Pinwheel Scheduling, edge-coloring, chromatic index, approximation algorithms, hardness of approximation}
}{}

\ifacm{
}{}
\ifthenelse{\NOT \equal{\docclass}{lipics} \AND \NOT \equal{\docclass}{acm}}{
	\newcommand\email[1]{\texttt{#1}}
	\author{%
				Leszek G\k asieniec%
					\footnote{U. of Liverpool, UK,
           \email{l.a.gasieniec\,@\,liverpool.ac.uk}}
			\and 
				Benjamin Smith%
					\footnote{U. of Liverpool, UK, 
					\email{b.m.smith\,@\,liverpool.ac.uk}}
			\and
				Sebastian Wild%
					\footnote{U. of Liverpool, UK, 
					\email{sebastian.wild\,@\,liverpool.ac.uk}}
	}

	\date{\small\today}
}{}
\ifsiam{
}{}

\acknowledgements{%
	We thank Casper Moldrup Rysgaard and Justin Dallant for many helpful discussions, and Casper especially, for being our rubber duck for a brutally unpolished version of the hardness-of-approximation reduction.
	We also thank Viktor Zamaraev for fruitful initial discussions and for setting us on the right track with the chromatic-index problem.
	\mbox{}\hfill\smash{\raisebox{-.3ex}{\includegraphics[height=2ex]{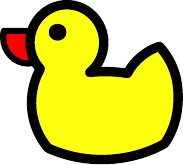}}}%
}

\definecolor{applegreen}{rgb}{0.55, 0.71, 0.0}
\definecolor{amethyst}{rgb}{0.6, 0.4, 0.8}
\definecolor{blue-violet}{rgb}{0.54, 0.17, 0.89}

\tikzset{graph node/.style = {circle, thick, draw}}
\tikzset{invisible node/.style = {}}
\tikzset{edge descriptor/.style = {text=black, fill=white,circle, inner sep=0cm}}
\tikzset{shorthand/.style = {rectangle, thick, fill=black, text=white, inner sep=3.5pt, draw=white, font=\bfseries\boldmath}}
\tikzset{SWAP/.style = {shorthand, minimum width = 3.975cm}}

\tikzset{ourBlack/.style = {black}}
\tikzset{ourRed/.style = {red}}
\tikzset{ourBlue/.style = {blue}}
\tikzset{ourGreen/.style = {applegreen}}
\tikzset{ourPurple/.style = {blue-violet}}

\tikzset{std/.style = {line width=1.5pt, draw}}
\tikzset{2 colour 1/.style = {line width=1.5pt, dash pattern= on 9pt off 11pt, draw}}
\tikzset{2 colour 2/.style = {line width=1.5pt, dash pattern= on 9pt off 11pt,dash phase=10pt, draw}}

\newcommand\SWAP{\ensuremath{\mathit{SWAP}}\xspace}
\newcommand\OR{\ensuremath{\mathit{OR}}\xspace}
\newcommand\SLOT{\mathcoloredunderline[.8pt]{}{\phantom{3_1}}}
\newcommand\mathunderline[2]{\mathcoloredunderline[.8pt]{#1}{#2}}

\begin{document}

\ifacm{}{\maketitle} 

\begin{abstract}
Finding schedules for pairwise meetings between the members of a complex social group without creating interpersonal conflict is challenging, especially when different relationships have different needs.
We formally define and study the underlying optimisation problem: Polyamorous Scheduling. 

In Polyamorous Scheduling, we are given an edge-weighted graph and 
try to find a periodic schedule of matchings in this graph 
such that the maximal weighted waiting time between consecutive occurrences of the same edge is 
minimised.
We show that the problem is NP-hard and that there is no efficient approximation algorithm with a better ratio than $4/3$ unless P = NP.
On the positive side, we obtain an $O(\log n)$-approximation algorithm; indeed, a $O(\log \Delta)$-approximation for $\Delta$ the maximum degree, \ie, the largest number of relationships of any individual.
We also define a generalisation of density from the Pinwheel Scheduling Problem, ``poly density'',
and ask whether there exists a poly-density threshold similar to the $5/6$-density threshold for Pinwheel Scheduling [Kawamura, STOC 2024].
Polyamorous Scheduling is a natural generalisation of Pinwheel Scheduling with respect to its optimisation variant, Bamboo Garden Trimming.

Our work contributes the first nontrivial hardness-of-approximation reduction for any periodic scheduling problem,
and opens up numerous avenues for further study of Polyamorous Scheduling.
\end{abstract}

\ifacm{%
	\maketitle%
}{}

\ifarxiv{\bigskip}{}

\section{Introduction}

We study a natural periodic scheduling problem faced by groups of regularity-loving polyamorous people:
Consider a set of persons and a set of pairwise relationships between them, each with a value representing its neediness, importance, or emotional weight. Find a periodic schedule of pairwise meetings between couples that minimizes the maximal weighted waiting time between such meetings, given that each person can meet with at most one of their partners on any particular day. 

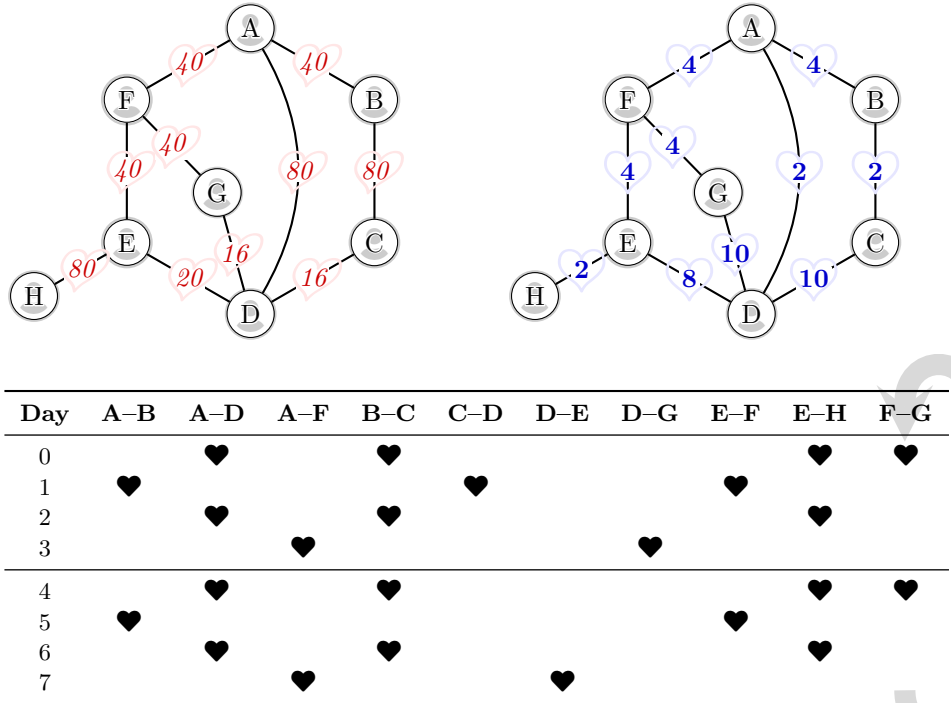
\begin{figure}[bht]
	
	\centering 
	
	\plaincenter{\begin{tikzpicture}[
			person/.style={circle,draw,inner sep=2.7pt},
			relation/.style={thick},
			growth label/.style={inner sep=1pt,fill=white,text=red!80!black,font=\itshape},
			scale=1.9,
		]
		\foreach \n/\p in {A/0,B/1,C/2,D/3,E/4,F/5} {
			\node[person] (\n) at ({-\p*60+90}:1) {\contour{white}{\n}} ;
		}
		\node[person] (G) at ($(F)!.5!(D) + (.2,.1)$) {\contour{white}{G}} ;
		\node[person] (H) at ({-4*60+90}:1.75) {\contour{white}{H}} ;
		\begin{pgfonlayer}{background}
			\foreach \p in {A,...,H} {
				\node[circle,inner sep=-1pt,fill=black!20,text=white,scale=1.8] at (\p) {\faUserCircle};
			}
		\end{pgfonlayer}
		\foreach \a/\b/\w/\s in {
				A/B/40/,%
				B/C/80/,%
				C/D/16/,%
				D/E/20/,%
				E/F/40/,%
				F/A/40/,%
				F/G/40/,%
				A/D/80/bend left,%
				D/G/16/,%
				E/H/80/%
		} {
			\draw[relation] (\a) to[\s] 
				node[growth label] (f\a\b) {\w} 
				(\b) ;
			\node[red!10,yscale=2.2,xscale=2.7,yshift=-.2ex,xslant=0.3] at (f\a\b) {$\heartsuit$};
		}
	\end{tikzpicture}\hfill
	\begin{tikzpicture}[
			person/.style={circle,draw,inner sep=2.7pt},
			relation/.style={thick},
			freq label/.style={inner sep=1pt,fill=white,text=blue!80!black,font=\bfseries},
			scale=1.9,
		]
		\foreach \n/\p in {A/0,B/1,C/2,D/3,E/4,F/5} {
			\node[person] (\n) at ({-\p*60+90}:1) {\contour{white}{\n}} ;
		}
		\node[person] (G) at ($(F)!.5!(D) + (.2,.1)$) {\contour{white}{G}} ;
		\node[person] (H) at ({-4*60+90}:1.75) {\contour{white}{H}} ;
		\begin{pgfonlayer}{background}
			\foreach \p in {A,...,H} {
				\node[circle,inner sep=-1pt,fill=black!20,text=white,scale=1.8] at (\p) {\faUserCircle};
			}
		\end{pgfonlayer}
		\foreach \a/\b/\w/\s in {
				A/B/4/,%
				B/C/2/,%
				C/D/10/,%
				D/E/8/,%
				E/F/4/,%
				F/A/4/,%
				F/G/4/,%
				A/D/2/bend left,%
				D/G/10/,%
				E/H/2/%
		} {
			\draw[relation] (\a) to[\s] node[freq label] (g\a\b) {\w} (\b) ;
			\node[blue!10,yscale=2.2,xscale=2.7,yshift=-.2ex] at (g\a\b) {$\heartsuit$};
		}
	\end{tikzpicture}}

	\bigskip
	\medskip
	
	\begin{tikzpicture}[remember picture,overlay] 
		\coordinate (top) at ($(pic cs:FG-top)+(-.5em,2ex)$);
		\coordinate (bot) at ($(pic cs:FG-bot)+(-.5em,0ex)$);
			\draw[black!15,-stealth,looseness=2,line width=.8em] 
				(bot) to[out=-90,in=-90] ++(.9,0) -- 
				($(top) + (.9,.5)$) to[out=90,in=90] ($(top) + (0,.5)$) -- (top) ;
	\end{tikzpicture}
	\def\m{\faHeart}
	\def\b{\textbf}
	\small
	\begin{tabular}{c*{10}{c}}
	\toprule
		\b{Day} & \b{A--B} & \b{A--D} & \b{A--F} & \b{B--C} & \b{C--D} & \b{D--E} & \b{D--G} & \b{E--F} & \b{E--H} &           \b{F--G}            \\
	\midrule
		   0    &          &    \m    &          &    \m    &          &          &          &          &    \m    &      \m\tikzmark{FG-top}      \\
		   1    &    \m    &          &          &          &    \m    &          &          &    \m    &          &                               \\
		   2    &          &    \m    &          &    \m    &          &          &          &          &    \m    &                               \\
		   3    &          &          &    \m    &          &          &          &    \m    &          &          &                               \\
	\midrule
		   4    &          &    \m    &          &    \m    &          &          &          &          &    \m    &              \m               \\
		   5    &    \m    &          &          &          &          &          &          &    \m    &          &                               \\
		   6    &          &    \m    &          &    \m    &          &          &          &          &    \m    &                               \\
		   7    &          &          &    \m    &          &          &    \m    &          &          &          & \phantom{\m}\tikzmark{FG-bot} \\
	\bottomrule
	\end{tabular}
	\smallskip
	\caption{%
		An example Optimisation Polyamorous Scheduling instance with 8 persons: \underbar Alex, \underbar Brady, \underbar Charlie, \underbar Daisy, \underbar Eli, \underbar Frankie, \underbar Grace, and \underbar Holly. 
		\textbf{Top left:} Graph representation with \textcolor{red!80!black}{\textit{edge labels}} showing the weight (desire growth rates) of each pairwise relationship.
		\textbf{Bottom:} An optimal schedule for the instance.  On each day, a set of meetings is scheduled as indicated by \faHeart{}s. The schedule has a period of 8 days: after day $7$, we start from day $0$ again.
		\textbf{Top right:} A decision version of the instance obtained for heat~160. The \textcolor{blue!80!black}{\textbf{edge labels}} here are the frequencies with which edges have to be scheduled stay below heat 160. 
	}
	\label{fig:intro-example}
\end{figure}

Before formally defining the Polyamorous Scheduling Problem (Poly Scheduling for short), 
we illustrate some features of the problem on an example.
\wref{fig:intro-example} shows an instance using the natural graph-based representation: We have vertices for people and weighted (undirected) edges for relationships.
It is easy to check that the schedule given at the bottom of \wref{fig:intro-example} never schedules more than one daily meeting for any of the 8 persons in the group; in the graph representation, the set of meetings for each day must form a matching.
Each day the mutual desire for a meeting experienced by each couple grows by the weight or \emph{desire growth rate} of that relationship\footnote{
	``Remember, absence makes the heart grow fonder''~\cite{Robin_Hood_1973}.\\ (\url{https://getyarn.io/yarn-clip/ae628721-c1d1-49d1-bd7c-78cbffceabf0})
	\hfill\mbox{}\smash{\includegraphics[height=3ex]{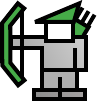}}%
} -- that is, until a meeting occurs and their desire is reset to zero.
We will refer to the highest desire ever felt by any pair when following a given schedule as the \emph{heat} of the schedule.
The heat of the schedule in \wref{fig:intro-example} is 160: as the reader can verify, no pair ever feels a desire greater than 160 before meeting and resetting their desire to zero. Desire 160 is also attained; \eg, Alex and Daisy are scheduled to see each other every other day, but over the period of 2 days between subsequent meetings, their desire grows to $2\cdot 80 = 160$.

For the instance in \wref{fig:intro-example}, it is easy to show that no schedule with heat $<160$ exists.
For that, we first convert from desire growth rates to required \emph{frequencies:} Under a heat-160 schedule, a pair with desire growth rate $g$ must meet at least every $\lfloor 160/g\rfloor $ days.  The top-right part of \wref{fig:intro-example} shows the result.
It is easy to check that the given schedule indeed achieves these frequencies. However, any further reduction of the desired heat to $160-\varepsilon$ would leave, \eg, Alex hopelessly overcommitted: the relation with Daisy would get frequency $\lfloor 80/(160-\varepsilon)\rfloor = 1$, forcing them to meet \emph{every} day; but then Brady and Frankie, each with frequency $\lfloor 40/(160-\varepsilon)\rfloor \le 3$ cannot be scheduled at all.

\ifarxiv{\needspace{3\baselineskip}}{}

While local arguments suffice for our small example, in general, Poly Scheduling is NP-hard (as shown below).
We therefore focus this paper on approximation algorithms and inapproximability results.

\subsection{Formal Problem Statement}

We begin by defining a decision version of Poly Scheduling.
In the \emph{Decision Polyamorous Scheduling Problem}, we are given a set of people and pairwise relationships
with ``attendance frequencies'' $f_{i,j}$, and we are trying to find a daily schedule of two-person meetings such that each couple $\{i,j\}$ meets at least every $f_{i,j}$ days.
The only constraint on the number of meetings that can occur on any given day is that each person can only participate in (at most) one of them.
A Decision Polyamorous Scheduling instance can naturally be modelled as a graph of people with the edges representing their relationships. Because each person can participate in at most one meeting per day, the edges scheduled on any given day must form a matching in this graph.

\begin{definition}[Decision Polyamorous Scheduling (DPS)]
\label{def:DPS}
	A DPS instance $\mathcal P_d = (P, R, f)$ (a ``(decision) \emph{polycule}'') consists of
	an undirected graph $(P,R)$ where the 
	vertices $P = \{p_1, \ldots,p_n\}$ are $n$ \emph{persons} and the 
	edges $R$ are pairwise relationships, with integer \emph{frequencies} $f : R \to \N$ for each relationship.
	
	The goal is to find an infinite schedule $S : \N_0 \to 2^R$, such that 
	\begin{thmenumerate}{def:DPS}
	\item \label{def:DPS-no-conflicts} 
		(no conflicts) for all days $t\in \N_0$, $S(t)$ is a matching in $\mathcal P_d$, and
	\item \label{def:DPS-frequencies} 
		(frequencies) for all $e \in R$ and $t\in \N_0$, we have 
		$e \in S(t)\cup S(t+1)\cup \cdots \cup S(t+f(e)-1)$;
	\end{thmenumerate}
	or to report that no such schedule exists.
	In the latter case, $\mathcal P_d$ is called \emph{infeasible}.
\end{definition}

We write $f_{i,j}$ and $f_e$ as shorthands for $f(\{p_i,p_j\})$ resp.\ $f(e)$.
An infinite schedule exists if and only if a \emph{periodic} schedule exists,
\ie, a schedule where there is a $T\in\N$ such that for all $t$, we have $S(t) = S(t+T)$:
any feasible schedule corresponds to an infinite walk in the finite configuration graph of the problem (see \wref{sec:computational-complexity}), implying the existence of a finite cycle.
A periodic schedule can be finitely described by listing $S(0),S(1),\ldots,S(T-1)$.

By relaxing the hard maximum frequencies of couple meetings to ``desire growth rates'', we obtain the Optimisation Polyamorous Scheduling (OPS) Problem. 
Our objective is to find a schedule that minimizes the \emph{``heat''}, \ie, the \emph{worst pain of separation} ever felt in the polycule by any couple.

\begin{definition}[Optimisation Polyamorous Scheduling]
\label{def:OPS}
	An OPS instance (or ``optimisation polycule'') $\mathcal P_o=(P, R, g)$ consists of 
	an undirected graph $(P,R)$
	along with a \emph{desire growth rate} $g:R\to \R_{>0}$ for each relationship in $R$.
	An infinite schedule $S : \N_0 \to 2^R$ is valid if, for all days, $t\in \N_0$, $S(t)$ is a matching in $\mathcal P_o$.
	
	The goal is to find a valid schedule that minimizes the \emph{heat} $h = h(S)$ of the schedule 
	where
	$h(S) = \max_{e\in R} h_e(S)$ and
	\[
		h_e(S) \wrel= \sup_{d\in \N} \,
		\begin{dcases*}
			(d+1)\cdot g(e) & $\exists t \in \N_0 : e \notin S(t)\cup S(t+1)\cup \cdots \cup S(t+d-1)$; \\
			g(e) & otherwise.
		\end{dcases*}
	\]
\end{definition}

As for DPS, $S$ can be assumed to be periodic without loss of generality, meaning that $S$ is finitely representable.

\ifarxiv{\needspace{5\baselineskip}}{}

\subsection{Related Work}
\label{sec:related-work}

Polyamorous Scheduling itself has not been studied to our knowledge.
Other variants of periodic scheduling have attracted considerable interest recently~\cite{Kawamura2023,HohneVanStee2023,AfshaniDeBergBuchinGaoLofflerNayyeriRaichelSarkarWangYang2022}, including FUN~\cite{BiloGualaLeucciProietiScornavacca2021}.

The arguably simplest periodic scheduling problem is \emph{Pinwheel Scheduling}.
In Pinwheel Scheduling~\cite{PinwheelIntro} we are given $k$ positive integer \emph{frequencies} 
$f_1\le f_2\le\cdots\le f_k$, and the goal is to find a Pinwheel schedule,
\ie, an infinite schedule of tasks $1,\ldots,k$
such that any contiguous time window of length $f_i$ contains 
at least one occurrence of $i$, for $i=1,\ldots,k$,
(or to report the non-existence of such a schedule).

Pinwheel Scheduling is NP-hard~\cite{pinwheelNPhard}, but unknown to be in NP~\cite{Kawamura2023},
(see \cite{GasieniecSmithWild2022} for more discussion).
Poly Scheduling inherits these properties.

The \emph{density} of a Pinwheel Scheduling instance is given by $d = \sum_{i=1}^k 1/f_i$.
It is easy to see that $d\le 1$ is a necessary condition for $A$ to be schedulable,
but this is not sufficient, as the infeasible instance $(2,3,M)$ with $d=\frac56+1/M$, for any $M\in\N$ shows.
However, there is a threshold $d^*$ so that $d\le d^*$ implies schedulability:
Whenever $d\le \frac12$, we can replace each frequency $f_i$ by $2^{\lceil \lg(f_i) \rceil}$ without increasing $d$ above $1$; then a periodic Pinwheel schedule always exists using the largest frequency as period length.
A long sequence of works~\cite{PinwheelIntro,largeClassPinwheel,pinwheel2numbers,3Task5over6,pinwheelSchedulable2over3,achievableDensities,5over6upto5tasks,GasieniecSmithWild2022} 
successively improved bounds on $d^*$,
culminating very recently in Kawamura's proof~\cite{Kawamura2023} that it is indeed a the sharp threshold, $d^*=\frac56$, 
confirming the corresponding conjecture of Chan and Chin from 1993~\cite{pinwheelSchedulable2Tasks}.
Generalizations of Pinwheel Scheduling have also been studied,
\eg, with jobs of different lengths~\cite{HanLin1992,generalisedPinwheelScheduling}.

Pinwheel Scheduling is a special case of DPS, where the underlying graph $(P,R)$ is a \emph{star}, \ie,
a centre connected to $k$ pendant vertices with edges of frequencies~$f_1,\ldots,f_k$.
Note that it is not generally possible to obtain a polyamorous schedule by combining the local schedule of each person\footnote{The current state-of-the-art approach in practice, usually via Google Calendar.}; see for example a triangle with edge frequencies~$2$: 
In the DPS instance $(\{A,B,C\}, \{A{-}B,B{-}C,A{-}C\}, f)$ with $f(e)=2$ for all edges, the local problem for each person is feasible by alternating between their two partners, but the global DPS instance has no solution.
This example also shows that the simple strategy of replacing $f_i$ by $2^{\lceil \lg(f_i) \rceil}$ is not sufficient to guarantee the existence of a schedule for Poly Scheduling.
Indeed, it is unclear whether any such constant-factor scaling of frequencies exists which applies to all Poly Scheduling instances.

There are two natural optimisation variants of Pinwheel Scheduling.
In \emph{Windows Scheduling}~\cite{BarNoyNaorSchieber2003} 
tasks with frequencies are given and the goal is to find a perpetual scheduling that
minimizes the \emph{number} of tasks that need to be done \emph{simultaneously}  
while respecting all frequencies (\ie, the number of channels or servers needed to schedule all tasks).
Efficient constant-factor approximation algorithms are known that use the connection to Bin Packing~\cite{windowsSchedulingBinPacking} (where we bin tasks by used channels),
even when the sets of tasks to schedule changes over time~\cite{ChanWong2004}.

The \emph{Bamboo Garden Trimming (BGT) Problem}~\cite{BGTIntro,GasieniecJurdzinskiKlasingLevcopoulosLingasMinRadzik2024}
retains the restriction of one task per day, but converts the frequencies into \emph{growth rates} $g_1\le \cdots\le g_k$ (of $k$ bamboo plants $1,\ldots, k$) and 
asks to find a perpetual schedule that minimizes the \emph{height} ever reached by any plant.
BGT also allows efficient constant-factor approximations 
whose approximation factor has seen a lively race of successively improvements over last few years:
from 2~\cite{BGTIntro} over
$\frac{12}7 \approx 1.71$~\cite{vanEe2021},
$1.6$~\cite{GasieniecJurdzinskiKlasingLevcopoulosLingasMinRadzik2024}, and
$1.4$~\cite{HohneVanStee2023},
down to the current record, $\frac43\approx 1.33$, again by Kawamura~\cite{Kawamura2023}.
As for the Windows Scheduling problem, no hardness of approximation results are known.
It remains open whether it is possible to obtain a PTAS for the Bamboo Garden Trimming Problem~\cite{GasieniecJurdzinskiKlasingLevcopoulosLingasMinRadzik2024} or the Windows Scheduling Problem.
We show that the same is not true for Poly Scheduling (see \wref{thm:sat-inapprox} below).

As for Pinwheel Scheduling and DPS, Bamboo Garden Trimming is the special case of OPS on star graphs.
Although BGT can be approximated well, since it is in general not possible to combine local schedules into a global schedule for a polycule (as noted above),
it is not clear whether Poly Scheduling allows an efficient constant-factor approximation.

All mentioned problems above have simple \emph{fractional} counterparts that are much easier to solve and hence provide necessary conditions.
Indeed, this is the motivation for density in Pinwheel Scheduling: 
if we allow a schedule to spend arbitrary fractions of the day on different tasks,
we obtain a schedule if and only if the density is at most $1$. 
(Spending a $1/f_i$ fraction on task $i$ each day is best possible).
For Windows Scheduling, any valid schedule must partition the tasks into bins (channels/servers), so that 
each bin admits a Pinwheel schedule. Relaxing the latter constraint to ``density at most $1$'' yields a standard bin packing problem, to which we can apply existing techniques;
(packing bins only up to density $5/6$ guarantees a Pinwheel schedule, at the expense of a $6/5$ factor increase in channels).
For Bamboo Garden Trimming, the optimal fractional schedule spends a $G/g_i$ fraction of each day with task $i$, where $G$ is the sum of all growth rates,
thus achieving height exactly $G$.
For Poly Scheduling, we can similarly define a fractional problem, but its structure is much richer (see \wref{sec:fractional}).

There are further periodic scheduling problems with less direct connections to Poly Scheduling that received attention in the literature.
Patrolling problems typically involve periodic schedules: for example,
\cite{AfshaniDeBergBuchinGaoLofflerNayyeriRaichelSarkarWangYang2022} finds schedules for a fleet of $k$ identical robots to patrol (unweighted) points in a metric space,
whereas the ``Continuous BGT Problem''~\cite{GasieniecJurdzinskiKlasingLevcopoulosLingasMinRadzik2024} sends
a single robot to points with different frequencies requirements;
\cite{KawamuraSoejima2020} tasks $k$ robots with patrolling a line or a circle.
The underlying geometry in these problems requires different techniques from our work.
The \emph{Point Patrolling Problem} studied in~\cite{KawamuraSoejima2020} can be seen as a ``covering version'' of Pinwheel Scheduling: each day, we have to assign one of $n$ workers to a single, daily recurring task, where worker $i$ requires a break of $a_i$ days before they can be made to work again.
Yet another twist on a patrolling problem is the \emph{Replenishment Problems with Fixed Turnover Times} given in~\cite{BosmanVanEeJiaoMarchettiSpaccamelaRaviStougie2022}, where vertices in a graph have to be visited with given frequencies, but instead of restricting the number of vertices that can be visited per day, the \emph{length of a tour} to visit them (starting at a depot node) shall be minimized.

In the \emph{Fair Hitting Sequence} Problem~\cite{CiceroneDiStefanoGasieniecJurdzinskiNavarraRadzikStachowiak2019}, we are given a collection of sets $\mathcal S = \{S_1,\ldots,S_m\}$, each consisting of a subset of the set of elements $\mathcal V = \{v_1,\ldots,v_n\}$. Each set $S_j$ has an urgency factor $g_j$, which is comparable to the growth rates in BGT instances with one key difference: 
A set $S_j$ is hit whenever any $v_i\in S_j$ is scheduled.
The goal is again similar to BGT; to find a perpetual schedule of elements $v_i\in \mathcal V$ that minimizes the time between visits to each set $S_j$, weighted by $g_j$.
There is also a decision variant, similar to Pinwheel Scheduling in that growth rates are replaced by frequencies.
We use a similar layering technique in our approximation algorithm (\wref{sec:approximations}) as the $O(\log^2 n)$-approximation from~\cite{CiceroneDiStefanoGasieniecJurdzinskiNavarraRadzikStachowiak2019}, but we obtain a better approximation ratio for Poly Scheduling.
Their $O(\log n)$-approximation based on randomized rounding does not extend to Poly Scheduling since the used  linear program has exponentially many variables for Poly Scheduling (\wref{sec:fractional}).

\subsection{Our Results}

Despite the recent flurry of results on periodic scheduling, 
Polyamorous Scheduling seems not to have been studied before.
Apart from its immediate practical applications, 
some quirks make Polyamorous Scheduling an interesting combinatorial optimization problem in its own~right.
\emph{One of these quirks seems to be that love makes blind \dots{} Shortly after publishing the first version of this manuscript, we found a substantial simplification and strengthening of our hardness-of-approximation results. We keep all original results in this revised manuscript for the record and for future improvement, but point out where they have been modified or superseded in this version.}

Our originally strongest hardness-of-approximation result, which 
rules out the existence of a PTAS (polynomial-time approximation scheme) for Poly Scheduling,
is based on a direct reduction from 3SAT.

\begin{theorem}[SAT Hardness of approximation]
\label{thm:sat-inapprox}
	Unless P = NP, there is no polynomial-time $(1+\delta)$-approximation algorithm 
	for the Optimisation Poly Scheduling problem for any $\delta < \frac1{12}$.
\end{theorem}

Love may not always reduce to logic, but propositional logic provably reduces to \emph{scheduling} love.

We also obtain a much simpler, and indeed stronger, inapproximability result, \wref{thm:color-inapprox}, from
containing the 3-Regular Chromatic Index Problem as a special case.

\begin{theorem}[Hardness of approximation]
\label{thm:color-inapprox}
	Unless P = NP, there is no polynomial-time $(1+\delta)$-approximation algorithm 
	for the Optimisation Poly Scheduling problem for any $\delta < \frac1{3}$.
\end{theorem}

While in its current form, \wref{thm:sat-inapprox} follows from \wref{thm:color-inapprox},
the direct 3SAT reduction is significantly more versatile and we hope to improve the lower bound on the approximation ratio in future work.
The core idea of the reduction in \wref{thm:sat-inapprox} is to force any valid schedule to have a periodic structure with a 3-day period, where edges scheduled on days $t$ with $t\equiv 0 \pmod 3$ represent the value \textit{True} and edges scheduled on days with $t\equiv 1 \pmod 3$ represent \textit{False}; the remaining slots, $t\equiv 2 \pmod 3$, are required to enforce correct propagation along logic gadgets.
Indeed, the actual construction uses a 6-day period throughout, where slots 2 and 5 are further distinguished.
A variable is represented as a person with 4 relationships with frequencies $[3,3,6,6]$.
The two frequency-6 edges are used to connect all variables, 
which forces the frequency-3 edges to choose slot $0$ or $1 \pmod 3$. 
The choice of which edge is scheduled in slot 0 encodes the variable assignment.
Complications arise in the construction because some gadgets, \eg, the \emph{logical or} ($\vee$), can only guarantee a weaker form of truth value encoding, which requires further gadgets (such as the ``tensioning gadget'') to preserve the reduction.
\wref{sec:inapproximability} gives the detailed gadget constructions and proofs;
\wpref{fig:worked-eg:Pdsz} shows the resulting DPS instance corresponding to an example 3-CNF formula.

\wref{thm:sat-inapprox} of course implies the NP-hardness of Poly Scheduling;
we overall have \emph{3} independent reductions establishing the NP-hardness for Poly Scheduling (\wref{sec:computational-complexity});
the best known upper bound for the complexity is PSPACE.

We could thus call Poly Scheduling \emph{very} NP-hard; 
yet efficient approximation algorithms are possible.
A simple round-robin schedule using an edge colouring yields a good approximation if \emph{both}
maximum degree and ratio between smallest and largest desire growth rates are small (\wref{thm:color-approx}).

\begin{theorem}[Colouring approximation]
\label{thm:color-approx}
	For an Optimisation Poly Scheduling instance $\mathcal P_o = (P,R,g)$ set 
	$g_{\min} = \min_{e\in R} g(e)$,
	$g_{\max} = \max_{e\in R} g(e)$, and let
	$\Delta$ be the maximum degree in $(P,R)$ and $h^*$ be the heat of an optimal schedule.
	There is an algorithm that computes in polynomial time a schedule $S$ of heat $h$ with 
	$\frac h{h^*} \le \min\bigl\{  \frac{\Delta+1}{\Delta}\cdot \frac {g_{\max}}{g_{\min}} ,\; \Delta+1 \bigr\}$.
\end{theorem}

A fully general approximation seems only possible with much weaker ratios; 
we provide a $O(\log \Delta)$-approximation by applying \wref{thm:color-approx} to groups with similar weight
and interleaving the resulting schedules.

\begin{theorem}[Layering approximation]
\label{thm:layer-approx}
	For an Optimisation Poly Scheduling instance $\mathcal P_o = (P,R,g)$, let
	$\Delta$ be the maximum degree in $(P,R)$ and $h^*$ be the heat of an optimal schedule.
	There is an algorithm that computes in polynomial time a schedule $S$ of heat $h$ with 
	$\frac h{h^*} \le 3\lg(\Delta+1) = O(\log n)$, where $n=|P|$.
\end{theorem}

Finally, we generalize the notion of density to Poly Scheduling.
As discussed above, density has proven instrumental in understanding the structure of Pinwheel Scheduling 
and in devising better approximation algorithms, by providing a simple, instance-specific lower bound.
For Poly Scheduling, the fractional problem is much richer, and indeed remains nontrivial to solve.
We devise a generalization of density for Poly Scheduling from the dual of the LP corresponding to a fractional variant of Poly Scheduling with gives the following instance-specific lower bound.

\begin{theorem}[Fractional lower bound]
\label{thm:dual-lp-bound}
	Given an OPS instance $\mathcal P_o = (P,R,g)$ with optimal heat~$h^*$.
	For any values $z_e\in[0,1]$, $e\in R$, with $\sum_{e\in R} z_e = 1$, we have
	\[
		h^* \wwrel\ge \bar h(z) \wrel= \frac1{\displaystyle \max_{M\in \mathcal M} \sum_{e\in M} \frac{z_e}{g(e)} }.
	\]
	with the maximum ranging over all maximal matchings in $(P,R)$.
	The largest value $\bar h^*$ of $h(z)$ over all feasible $z$, is the \emph{poly density} of $\mathcal P_o$.
\end{theorem}

The bound implies (and formally establishes) simple ad-hoc bounds such as the following, 
which corresponds to the lower bound of $G$ on the height in Bamboo Garden Trimming.

\begin{corollary}[Total growth bound]
\label{cor:total-growth-bound}
	Given an OPS instance $\mathcal P_o = (P,R,g)$ with optimal heat $h^*$.
	Let $G = \sum_{e\in R} g(e)$ and $m$ be the size of a maximum matching in $(P,R)$.
	Then $h^* \ge G / m$.
\end{corollary}

More importantly though, \wref{thm:dual-lp-bound} allows us to define a \emph{poly density} similarly to the Pinwheel Scheduling Problem, and allows us to formulate the most interesting open problem about Poly Scheduling.
For a DPS instance $\mathcal P_d = (P,R,f)$, define the poly density of $\mathcal P_d$, $\bar h^*(\mathcal P_d)$, as the poly density of the OPS instance $\mathcal P_o = (P,R,1/f)$ (see also \wref{lem:ops-to-dps}).

\begin{openproblem}[Poly Density Threshold?]
    \label{open:poly-density}
	Is there a constant $c$ such that every Decision Poly Scheduling instance $\mathcal P_d = (P,R,f)$
	with poly density $\bar h^*(\mathcal P_d) \le c$ admits a valid schedule?
\end{openproblem}

\section{Preliminaries}
\label{sec:preliminaries}

In this section, we introduce some general notation and collect a few simple facts about Poly Scheduling used later.

We write $[n..m]$ for $\{n,n+1,\ldots,m\}$ and $[n]$ for $[1..n]$.
For a set $A$, we denote its powerset by $2^A$.
All graphs in this paper are simple and undirected. We denote by $\mathcal M = \mathcal M(V,E)$ the set of \emph{inclusion-maximal matchings} in graph $(V,E)$,
where matching has the usual meaning of an edge set with no two edges incident to the same vertex.
By $\Delta=\Delta(V,E)$, we denote the \emph{maximum degree} in $(V,E)$.
A \emph{pendant vertex} is a vertex with degree 1.
The \emph{chromatic index} $\chi_1 = \chi_1(V,E)$ is the smallest number $C$ of ``colours'' in a proper edge colouring of $(V,E)$ (\ie, the number of disjoint matchings required to cover $E$);
by Vizing's Theorem~\cite{graphEdgeColourVizing}, we have $\Delta \le \chi_1\le \Delta+1$ for every graph.
Misra and Gries provide a polynomial-time algorithm for edge colouring any graph using at most $\Delta + 1$ colours~\cite{deltaPlusOneEdgeColouringAlg}.

Given a schedule $S:\N_0\to 2^R$ and an edge $e\in R$, we define the \emph{(maximal) recurrence time} $r(e) = r_S(e)$ of $e$ in $S$ as the maximal time between consecutive occurrences of $e$ in $S$, formally 
\[
	r_S(e) \wrel= \sup_{d\in \N} \,
	\begin{dcases*}
		d+1 & $\exists t \in \N_0 : e \notin S(t)\cup S(t+1)\cup \cdots \cup S(t+d-1)$; \\
		0 & otherwise.
	\end{dcases*}
\]
Using recurrence time, the heat $h = h(S)$ of a schedule $S$ in an OPS instance $(P,R,g)$ is 
$h(S) = \max_{e\in R} g(e)\cdot r(e)$.
Clearly, for any schedule $S:\N_0\to 2^R$, we can obtain $S':\N_0\to \mathcal M$ by adding edges to $S(t)$ until we have a maximal matching $S'(t) \supseteq S(t)$; then $r_{S'}(e) \le r_S(e)$ for all $e\in R$ and hence $S'$ is a valid schedule for any DPS instance for which $S$ is valid, and if $S$ schedules an OPS instance with heat $h(S)$ then $S'$ does too, with $h(S') \le h(S)$.

We use \wref{lem:ops-to-dps} to reduce OPS to DPS, and \wref{lem:dps-to-ops} to formalize how DPS solves OPS:

\begin{lemma}[OPS to DPS]
	\label{lem:ops-to-dps}
	For every combination of OPS instance $\mathcal{P}_{o}=(P,R,g)$ and heat value $h$, there exists a DPS instance $\mathcal{P}_{d}=(P,R,f)$ such that 
	\begin{thmenumerate}{lem:ops-to-dps}
		\item any feasible schedule $S:\N_0\to 2^R$ for $\mathcal{P}_{d}$ is a schedule for $\mathcal{P}_{o}$ with heat $\le h$, and 
		\item any schedule $S'$ for $\mathcal{P}_{o}$ with heat $h' > h$ is not feasible for $\mathcal{P}_{d}$.
	\end{thmenumerate}
\end{lemma}
\begin{proof}
Consider an OPS polycule $\mathcal{P}_{o} = (P,R,g)$; we set $\mathcal{P}_{d} = (P,R,f)$ where $f(e) = \bigl\lfloor \frac{h}{g(e)} \bigr\rfloor$ for all $e\in R$. Schedules satisfying $\mathcal{P}_{d}$ when applied to $\mathcal{P}_{o}$ will allow heat of at most 
$\max_{e\in R} \, g(e)\cdot f(e) = \max_{e\in R} \, g(e) \lfloor\frac{h}{g(e)}\rfloor \leq h$.

Now consider a schedule $S'$ for $\mathcal{P}_{o}$ with heat $h' > h$. By definition, $h' = \max_{e\in R} r_{S'}(e)\cdot g(e)$, where $r(e)=r_{S'}(e)$ is the recurrence time of $e$ in $S'$.
Assume towards a contradiction that $r(e)\leq f(e)$ for all $e\in R$. This implies that 
$h'
= \max_{e\in R} r(e)\cdot g(e) 
\leq \max_{e\in R} \lfloor \frac{h}{g(e)} \rfloor\cdot g(e)
\le h
$, a contradiction to the assumption.
\end{proof}

\begin{lemma}[DPS to OPS]
	\label{lem:dps-to-ops}
	Let $\mathcal P_d = (P,R,f)$ be a DPS instance. Set $F=\max_{e\in R} f(e)$. 
	There is an OPS instance $\mathcal P_o = (P,R,g)$ such that the following holds.
	\begin{thmenumerate}{lem:dps-to-ops}
		\item If $\mathcal P_d$ is feasible, then $\mathcal P_o$ admits a schedule of height $h\le 1$.
		\item If $\mathcal P_d$ is infeasible, then the optimal heat $h^*$ of $\mathcal P_o$ satisfies $h^* \ge (F+1)/F$.
	\end{thmenumerate}
\end{lemma}
\begin{proof}
Let $\mathcal P_d = (P,R,f)$ be given; we set $\mathcal P_o = (P,R,g)$ with $g(e) = 1/f(e)$ for $e\in R$.
Any feasible schedule $S$ for $\mathcal P_d$ has recurrence time $r(e) \le f(e)$ by definition,
so its heat in $\mathcal P_o$ is $h(S) = \max r(e) g(e) = \max r(e) \frac1{f(e)} \le 1$.
If conversely $\mathcal P_d$ is infeasible, for every $S:\N_0\to 2^R$ there exists an edge $e\in R$ where $r_S(e) > f(e)$, \ie, $r_S(e) \ge f(e)+1$. For the heat of $S$ in $\mathcal P_o$, this means
$h(S) = \max_{e\in R} r_S(e) g(e) \ge \max_{e\in R} (f(e)+1) \frac1{f(e)} \ge (F+1)/F $.
\end{proof}

We will often use the \emph{Normal Form} of OPS instances in proofs; this can be assumed without loss of generality but is not generally useful for algorithms unless $h^*$ is known:

\begin{lemma}[Normal Form OPS]
	\label{lem:normal-form-ops}
	For every OPS instance $\mathcal P_o = (P,R,g)$, there is an equivalent OPS instance $\mathcal P'_o = (P,R,g')$ with optimal heat $1$ where $g':R\to \mathcal U$ for $\mathcal U = \{1/m : m \in \N_{\ge1}\}$, \ie, the set of unit fractions. 
	More precisely, for every schedule $S:\N_0\to 2^R$ holds: $S$ has optimal heat $h^*$ in $\mathcal P_o$ if and only if $S$ has heat $1$ in $\mathcal P'_o$.
\end{lemma}
\begin{proof}
Let $(P,R,g)$ be an arbitrary OPS instance with optimal heat $h^*$.
Setting $\hat g(e) = g(e)/h^*$ yields OPS instance $(P,R,\hat g)$ with optimal heat $1$.
We now start by setting $g'(e) = \hat g(e)$ for all $e\in R$.
Fix an optimal schedule $S$.
Suppose that for some relation $e\in R$, we have $g'(e) \notin \mathcal U$.
In $S$, there is a maximal separation $r(e) = q \in \N$ between consecutive occurrences of $e$ with $q\cdot g(e) \le h^*$. But then, increasing $g(e)$ to $h^*/q$ would not affect the heat of $S$.
We can thus set $g'(e) = 1/q$.
By induction, we thus obtain $g': R\to \mathcal U$ without affecting the heat of $S$.
\end{proof}

\section{Inapproximability}
\label{sec:inapproximability}

In this section, we prove \wref{thm:sat-inapprox}.  
We show via a reduction from MAX-3SAT that OPS does not allow efficient $(1+\delta)$-approximation algorithms for $\delta<\frac{1}{12}$ unless P $=$ NP. 

\subsection{Overview of Proof}
\label{sec:inapproximability-overview}

The proof of \wref{thm:sat-inapprox} has two steps:
The first step is a reduction from the decision version of MAX-3SAT (D-MAX-3SAT) to DPS.
D-MAX-3SAT is the following problem: given a 3-CNF formula $\varphi = c_1\land c_2\land \cdots \land c_m$ with clauses $C=\{c_1,c_2,\ldots, c_m\}$ over variables $X=\{x_1, x_2,\ldots, x_{n'}\}$ and integer $k$, decide whether there is an assignment of Boolean values to the variables in $X$ such that at least $k$ clauses in $C$ are satisfied. 

\begin{lemma}[D-MAX-3SAT $\le_p$ DPS]
    \label{lem:dps=3sat}
    For any 3-CNF formula $\varphi$ with $m$ clauses and integer $k\leq m$, we can construct in polynomial time a decision polycule $\mathcal P_{d\varphi k}$ which has a valid schedule if and only if at least $k$ clauses of $\varphi$ can be simultaneously satisfied.
\end{lemma}

\wref{lem:dps=3sat} is our key technical contribution and its proof will be given over the course of the remainder of this section.

The second step in the proof of \wref{thm:sat-inapprox} is to convert the decision polycule $\mathcal P_{d\varphi k}$ from \wref{lem:dps=3sat} to an optimisation polycule $\mathcal P_{o\varphi k}$ using \wref{lem:dps-to-ops}.
It will be immediate from the construction that the largest frequency in $\mathcal P_{d\varphi k}$ is $F=12$.
So by \wref{lem:dps-to-ops}, $\mathcal P_{o\varphi k}$ has either $h^* = 1$, namely if $\mathcal P_{d\varphi k}$ is feasible, or $h^*\ge \frac{13}{12}$, otherwise.

\begin{proofof}{\wref{thm:sat-inapprox}}
Assume that there is a polynomial-time approximation algorithm $A$ for OPS with approximation ratio $\alpha<\frac{13}{12}$. 
If $\mathcal P_{o\varphi k}$ has optimal heat $h^* = 1$, $A$ produces a schedule with heat $1\leq h\leq \alpha < \frac{13}{12}$, whereas if $h^*\ge \frac{13}{12}$, $A$ must produce a schedule of heat $\frac{13}{12}\leq h\leq \alpha\cdot\frac{13}{12}$.
So, by running $A$, we are able to distinguish between $h^*\le 1$ and $h^*\ge \frac{13}{12}$ for $\mathcal P_{o\varphi k}$, hence between feasibility or infeasibility of $\mathcal P_{d\varphi k}$ and, therefore between Yes and No instances of D-MAX-3SAT via the polynomial-time reduction from \wref{lem:dps=3sat}.
As a generalisation of 3SAT, D-MAX-3SAT is NP-hard, hence P $=$ NP follows.
\end{proofof}

Let us denote by $\alpha^*$ the approximability threshold for Optimisation Poly Scheduling,
that is:
efficient polynomial-time approximation algorithms with approximation ratio $\alpha$ exist 
if and only if $\alpha \ge \alpha^*$ (assuming P $\ne$ NP).
\wref{thm:sat-inapprox} shows that $\alpha^*\geq\frac{13}{12}$ and
\wref{thm:layer-approx} shows that $\alpha^* = O(\log n)$,
leaving a substantial gap. 
We conjecture that the constant $\frac{13}{12}$ can be improved by careful analysis of our construction, but we leave this to future work. 

\begin{conjecture}
    \label{conj:4/3}
    $\alpha^* \ge \frac43$.
\end{conjecture}

\begin{remark}[Towards stronger inapproximability results]
    Our reduction includes additional degrees of freedom not currently used towards the proof of \wref{lem:dps=3sat}.
    In particular, for the current statement, a reduction for standard 3SAT would have sufficed, removing the sorting network from the construction.
    However, to find better constant or even superconstant lower bounds for $\alpha^*$ it seems likely that starting with a gapped MAX-3SAT problem can provide stronger gaps for the outcome.
    For that, we need that our construction allows us to specify freely how many clauses have to be satisfiable for $\mathcal P_{d\varphi k}$ to be feasible.
    While we leave this to future work, we include here the required features in the construction which may facilitate these results.
\end{remark}

\subsection{Reduction Overview}
\label{sec:DMAX3SAT-to-OPS}

We now give the proof of \wref{lem:dps=3sat}.
For the remainder of this section, we assume a 3-CNF formula $\varphi = c_1\land \cdots\land c_m$ over variables $X = \{x_1,\ldots, x_{n'}\}$ and integer $k$ are given.
We will describe how to construct the DPS instances $\mathcal P_{d\varphi k}$ that admits a schedule iff
there is a variable assignment $v:X \to \{\mathrm{True},\mathrm{False}\}$ that satisfies at least $k$ clauses in $C = \{c_1,\ldots,c_m\}$. 
The construction is based on building components of Boolean formulas via DPS ``gadgets'': 
\begin{itemize}
\item \emph{variables} (\wref{sec:variables}),
\item clauses (\emph{OR gadgets}) (\wref{sec:OR}),
\item a \emph{sorting network}, comprised of \emph{\SWAP gadgets} (\wref{sec:SWAP}) to group satisfied outputs together, and 
\item a check for $\ge k$ true clauses, the \emph{tension gadget} (\wref{sec:Tension}).
\end{itemize}
To make those gadgets work, we require further auxiliary gadgets:
\begin{itemize}
\item a \emph{``True Clock''} to break ties between symmetric choices for schedules (\wref{sec:TrueClock}),
\item slot duplication gadgets: \emph{$D_3$ duplicators} (\wref{sec:D3}), \emph{$D_6$ duplicators} (\wref{sec:D6}), $D_{12}$ duplicators (\wref{sec:D12}), and
\item \emph{slot splitting gadgets} $S_{B6}$, $S_{B12}$, $S_{G12}$ (\wref{sec:making-12s}).
\end{itemize}

The overall conversion algorithm is stated in \wref{def:gadgets-assemble} below;
a worked example is shown in \wref{fig:worked-eg:Pdsz}.

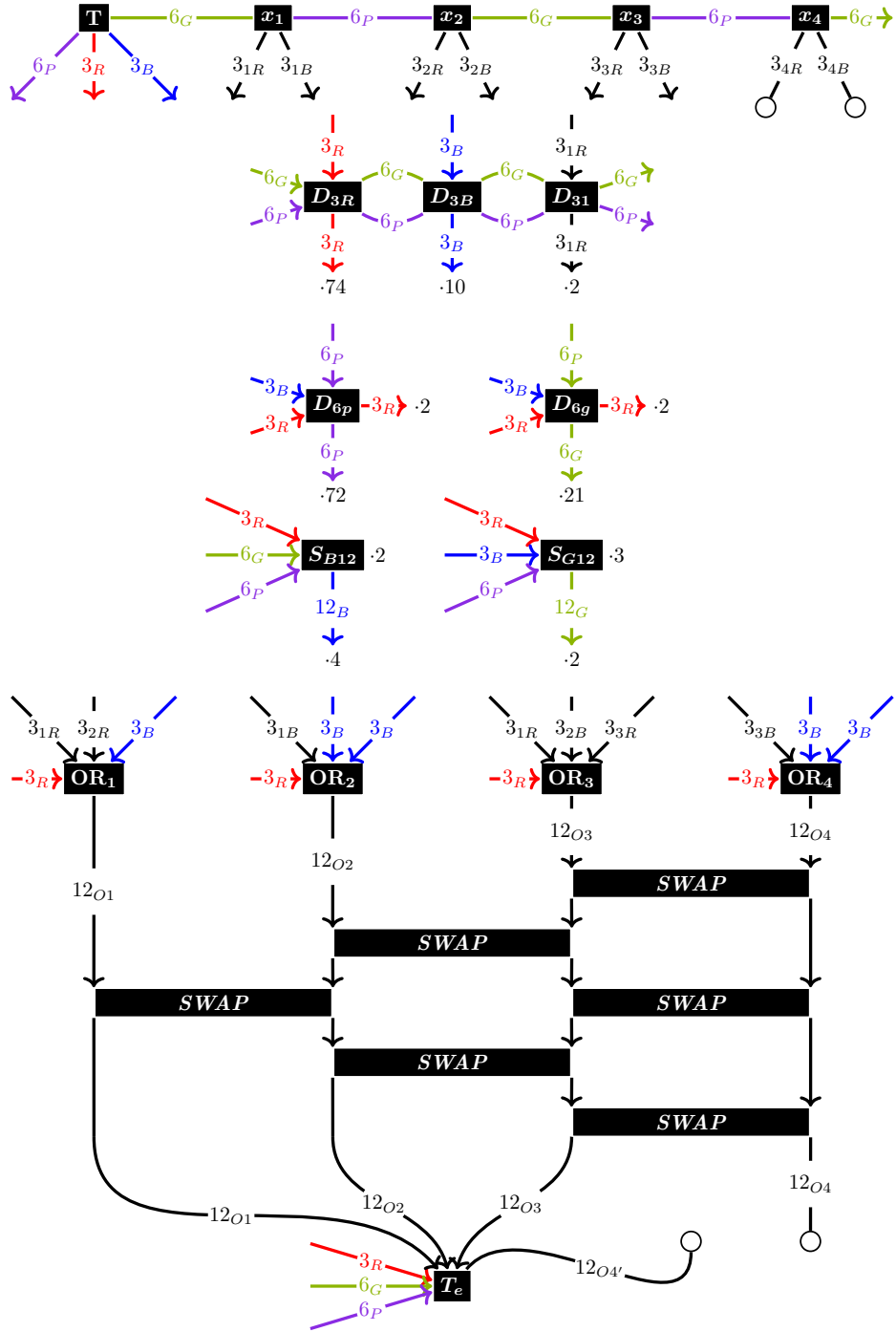
\begin{figure*}[p]
    \centering
    \resizebox{!}{0.84\textheight}{
    
    \begin{tikzpicture}[scale=2]
            
        \node[shorthand] (T) at (-1, 0) {T};
        \node[shorthand] (x1) at (0.5, 0) {$x_1$};
        \node[shorthand] (x2) at (2, 0) {$x_2$};
        \node[shorthand] (x3) at (3.5, 0) {$x_3$};
        \node[shorthand] (x4) at (5, 0) {$x_4$};
        \node[invisible node] (spare out) at (5.75, 0) {};
    
        \path[-] (T) edge[std, ourGreen] node[edge descriptor] {\color{applegreen}$6_G$} (x1);
        \path[-] (x1) edge[std, ourPurple] node[edge descriptor] {\color{blue-violet}$6_P$} (x2);
        \path[-] (x2) edge[std, ourGreen] node[edge descriptor] {\color{applegreen}$6_G$} (x3);
        \path[-] (x3) edge[std, ourPurple] node[edge descriptor] {\color{blue-violet}$6_P$} (x4);
        \path[->] (x4) edge[std, ourGreen] node[edge descriptor] {\color{applegreen}$6_G$} (spare out);
    
        \node[invisible node] (T out 1) at (-1.75, -0.75) {};    
        \node[invisible node] (T out 2) at (-1, -0.75) {};
        \node[invisible node] (T out 3) at (-0.25, -0.75) {};
        \node[invisible node] (x1 out 1) at (0.125, -0.75) {};
        \node[invisible node] (x1 out 2) at (0.875, -0.75) {};
        \node[invisible node] (x2 out 1) at (1.625, -0.75) {};
        \node[invisible node] (x2 out 2) at (2.375, -0.75) {};
        \node[invisible node] (x3 out 1) at (3.125, -0.75) {};
        \node[invisible node] (x3 out 2) at (3.875, -0.75) {};
        \node[graph node] (x4 out 1) at (4.625, -0.75) {};
        \node[graph node] (x4 out 2) at (5.375, -0.75) {};
    
        \path[->] (T) edge[std, ourPurple] node[edge descriptor] {\color{blue-violet}$6_P$} (T out 1);
        \path[->] (T) edge[std, ourRed] node[edge descriptor] {\color{red}$3_R$} (T out 2);
        \path[->] (T) edge[std, ourBlue] node[edge descriptor] {\color{blue}$3_B$} (T out 3);
        \path[->] (x1) edge[std, ourBlack] node[edge descriptor] {$3_{1R}$} (x1 out 1);
        \path[->] (x1) edge[std, ourBlack] node[edge descriptor] {$3_{1B}$} (x1 out 2);
        \path[->] (x2) edge[std, ourBlack] node[edge descriptor] {$3_{2R}$} (x2 out 1);
        \path[->] (x2) edge[std, ourBlack] node[edge descriptor] {$3_{2B}$} (x2 out 2);
        \path[->] (x3) edge[std, ourBlack] node[edge descriptor] {$3_{3R}$} (x3 out 1);
        \path[->] (x3) edge[std, ourBlack] node[edge descriptor] {$3_{3B}$} (x3 out 2);
        \path[-] (x4) edge[std, ourBlack] node[edge descriptor] {$3_{4R}$} (x4 out 1);
        \path[-] (x4) edge[std, ourBlack] node[edge descriptor] {$3_{4B}$} (x4 out 2);

        \def\xadj{0}
       
        \node[invisible node] (2inR)    at (1+\xadj, -0.75) {};    
        \node[invisible node] (2inB)    at (2+\xadj, -0.75) {};
        \node[invisible node] (2in1)    at (3+\xadj, -0.75) {};

        \node[shorthand] (D3R)          at (1+\xadj, -1.5) {$D_{3R}$};
        \node[shorthand] (D3B)          at (2+\xadj, -1.5) {$D_{3B}$};
        \node[shorthand] (D31)          at (3+\xadj, -1.5) {$D_{31}$};

        \path[->] (2inR) edge[std, ourRed] node[edge descriptor] {\color{red}$3_R$} (D3R);
        \path[->] (2inB) edge[std, ourBlue] node[edge descriptor] {\color{blue}$3_B$} (D3B);
        \path[->] (2in1) edge[std, ourBlack] node[edge descriptor] {$3_{1R}$} (D31);

        \node[invisible node] (2inG)    at (0.25+\xadj, -1.75+0.5) {};    
        \node[invisible node] (2inP)    at (0.25+\xadj, -2.25+0.5) {};
        \node[invisible node] (2outG)   at (3.75+\xadj, -1.75+0.5) {};
        \node[invisible node] (2outP)   at (3.75+\xadj, -2.25+0.5) {};

        \path[->] (2inG) edge[std, ourGreen] node[edge descriptor] {\color{applegreen}$6_G$} (D3R);
        \path[->] (2inP) edge[std, ourPurple] node[edge descriptor] {\color{blue-violet}$6_P$} (D3R);
        \path[->] (D31) edge[std, ourGreen] node[edge descriptor] {\color{applegreen}$6_G$} (2outG);
        \path[->] (D31) edge[std, ourPurple] node[edge descriptor] {\color{blue-violet}$6_P$} (2outP);

        \path[-] (D3R) edge[std, ourGreen, bend left] node[edge descriptor] {\color{applegreen}$6_G$} (D3B);
        \path[-] (D3R) edge[std, ourPurple, bend right] node[edge descriptor] {\color{blue-violet}$6_P$} (D3B);
        \path[-] (D3B) edge[std, ourGreen, bend left] node[edge descriptor] {\color{applegreen}$6_G$} (D31);
        \path[-] (D3B) edge[std, ourPurple, bend right] node[edge descriptor] {\color{blue-violet}$6_P$} (D31);

        \node[invisible node] (2outR)   at (1+\xadj, -2.75+0.5) {$\cdot 74$};    
        \node[invisible node] (2outB)   at (2+\xadj, -2.75+0.5) {$\cdot 10$};
        \node[invisible node] (2out1)   at (3+\xadj, -2.75+0.5) {$\cdot 2$};

        \path[->] (D3R) edge[std, ourRed] node[edge descriptor] {\color{red}$3_R$} (2outR);
        \path[->] (D3B) edge[std, ourBlue] node[edge descriptor] {\color{blue}$3_B$} (2outB);
        \path[->] (D31) edge[std, ourBlack] node[edge descriptor] {$3_{1R}$} (2out1);

        \def\yadj{0.75}
        \def\xadj{}
        \node[invisible node](D6P in B)     at (\xadj+0.25, -3.75+\yadj){};
        \node[invisible node](D6P in R)     at (\xadj+0.25, -4.25+\yadj){};
        \node[invisible node](D6P in P)     at (\xadj+1, -3.25+\yadj){};
        \node[shorthand] (D6P)              at (\xadj+1, -4+\yadj) {$D_{6p}$};
        \node[invisible node](D6P out R)    at (\xadj+1.75, -4+\yadj){$\cdot 2$};

        \path[->] (D6P in P) edge[std, ourPurple] node[edge descriptor] {\color{blue-violet}$6_P$} (D6P);
        \path[->] (D6P in B) edge[std, ourBlue] node[edge descriptor] {\color{blue}$3_B$} (D6P);
        \path[->] (D6P in R) edge[std, ourRed] node[edge descriptor] {\color{red}$3_R$} (D6P);
        \path[->] (D6P) edge[std, ourRed] node[edge descriptor] {\color{red}$3_R$} (D6P out R);

        \node[invisible node](D6P out P)    at (\xadj+1, -4.75+\yadj){$\cdot 72$};

        \path[->] (D6P) edge[std, ourPurple] node[edge descriptor] {\color{blue-violet}$6_P$} (D6P out P);

        \def\xadj{}
    
        \node[invisible node](D6G in B)     at (\xadj+2.25, -3.75+\yadj){};
        \node[invisible node](D6G in R)     at (\xadj+2.25, -4.25+\yadj){};
        \node[invisible node](D6G in G)     at (\xadj+3, -3.25+\yadj){};
        \node[shorthand] (D6G)              at (\xadj+3, -4+\yadj) {$D_{6g}$};
        \node[invisible node](D6G out R)    at (\xadj+3.75, -4+\yadj){$\cdot 2$};

        \path[->] (D6G in G) edge[std, ourGreen] node[edge descriptor] {\color{applegreen}$6_P$} (D6G);
        \path[->] (D6G in B) edge[std, ourBlue] node[edge descriptor] {\color{blue}$3_B$} (D6G);
        \path[->] (D6G in R) edge[std, ourRed] node[edge descriptor] {\color{red}$3_R$} (D6G);
        \path[->] (D6G) edge[std, ourRed] node[edge descriptor] {\color{red}$3_R$} (D6G out R);

        \node[invisible node](D6G out G)    at (\xadj+3, -4.75+\yadj){$\cdot 21$};

        \path[->] (D6G) edge[std, ourGreen] node[edge descriptor] {\color{applegreen}$6_G$} (D6G out G);

        \def\xadj{-0.5}
        \def\yadj{-2.5+0.75}
        
        \node[shorthand]      (short split 12)  at (1.5+\xadj, -2.75+\yadj){$S_{B12}$};
        \node[invisible node] (short red 12)    at (0.375+\xadj, -2.25+\yadj) {};
        \node[invisible node] (short green 12)  at (0.375+\xadj, -2.75+\yadj) {};
        \node[invisible node] (short purple 12) at (0.375+\xadj, -3.25+\yadj) {};
        \node[invisible node] (short out 12)  at (1.5+\xadj, -3.625+\yadj) {$\cdot 4$};
        \node[invisible node]      (many B12)  at (1.875+\xadj, -2.75+\yadj){$\cdot 2$};

        \path[->] (short red 12) edge[std, ourRed] node[edge descriptor] {$\color{red}3_R$} (short split 12);
        \path[->] (short green 12) edge[std, ourGreen] node[edge descriptor] {$\color{applegreen}6_G$} (short split 12);
        \path[->] (short purple 12) edge[std, ourPurple] node[edge descriptor] {$\color{blue-violet}6_P$} (short split 12);
        \path[->] (short split 12) edge[std, ourBlue] node[edge descriptor] {$\color{blue}12_B$} (short out 12);

        \def\xadj{1.5}
        \def\yadj{-4.5+0.75}
        
        \node[shorthand]      (short split 12)  at (\xadj+1.5, \yadj-0.75){$S_{G12}$};
        \node[invisible node] (short red 12)    at (\xadj+0.375, \yadj-0.25) {};
        \node[invisible node] (short green 12)  at (\xadj+0.375, \yadj-0.75) {};
        \node[invisible node] (short purple 12) at (\xadj+0.375, \yadj-1.25) {};
        \node[invisible node] (short out 12 1)  at (\xadj+1.5, \yadj-1.625) {$\cdot 2$};
        \node[invisible node]      (many G12)  at (1.875+\xadj, -0.75+\yadj){$\cdot 3$};

        \path[->] (short red 12) edge[std, ourRed] node[edge descriptor] {$\color{red}3_R$} (short split 12);
        \path[->] (short green 12) edge[std, ourBlue] node[edge descriptor] {$\color{blue}3_B$} (short split 12);
        \path[->] (short purple 12) edge[std, ourPurple] node[edge descriptor] {$\color{blue-violet}6_P$} (short split 12);
        \path[->] (short split 12) edge[std, ourGreen] node[edge descriptor] {$\color{applegreen}12_G$} (short out 12 1);

        \def\yadj{-1.375}
        \node[invisible node] (OR1 in top 1)    at (-1.75, -5.25+1+\yadj) {};
        \node[invisible node] (OR1 in top 2)    at (-1, -5.25+1+\yadj) {};
        \node[invisible node] (OR1 in top 3)    at (-0.25, -5.25+1+\yadj) {};
        \node[shorthand] (OR1)                  at (-1, -6+1+\yadj) {OR$_1$};
        \node[invisible node] (OR1 in red)      at (-1.75, -6+1+\yadj){};
        
        \path[->] (OR1 in top 1) edge[std, ourBlack] node[edge descriptor] {$3_{1R}$} (OR1);
        \path[->] (OR1 in top 2) edge[std, ourBlack] node[edge descriptor] {$3_{2R}$} (OR1);
        \path[->] (OR1 in top 3) edge[std, ourBlue] node[edge descriptor] {$\color{blue}3_B$} (OR1);
        \path[->] (OR1 in red) edge[std, ourRed] node[edge descriptor] {\color{red}$3_R$} (OR1);

        \node[invisible node] (OR2 in top 1)    at (0.25, -5.25+1+\yadj) {};
        \node[invisible node] (OR2 in top 2)    at (1, -5.25+1+\yadj) {};
        \node[invisible node] (OR2 in top 3)    at (1.75, -5.25+1+\yadj) {};
        \node[shorthand] (OR2)                  at (1, -6+1+\yadj) {OR$_2$};
        \node[invisible node] (OR2 in red)      at (0.25, -6+1+\yadj){};
        
        \path[->] (OR2 in top 1) edge[std, ourBlack] node[edge descriptor] {$3_{1B}$} (OR2);
        \path[->] (OR2 in top 2) edge[std, ourBlue] node[edge descriptor] {$\color{blue}3_B$} (OR2);
        \path[->] (OR2 in top 3) edge[std, ourBlue] node[edge descriptor] {$\color{blue}3_B$} (OR2);
        \path[->] (OR2 in red) edge[std, ourRed] node[edge descriptor] {\color{red}$3_R$} (OR2);

        \node[invisible node] (OR3 in top 1)    at (2.25, -5.25+1+\yadj) {};
        \node[invisible node] (OR3 in top 2)    at (3, -5.25+1+\yadj) {};
        \node[invisible node] (OR3 in top 3)    at (3.75, -5.25+1+\yadj) {};
        \node[shorthand] (OR3)                  at (3, -6+1+\yadj) {OR$_3$};
        \node[invisible node] (OR3 in red)      at (2.25, -6+1+\yadj){};
        
        \path[->] (OR3 in top 1) edge[std, ourBlack] node[edge descriptor] {$3_{1R}$} (OR3);
        \path[->] (OR3 in top 2) edge[std, ourBlack] node[edge descriptor] {$3_{2B}$} (OR3);
        \path[->] (OR3 in top 3) edge[std, ourBlack] node[edge descriptor] {$3_{3R}$} (OR3);
        \path[->] (OR3 in red) edge[std, ourRed] node[edge descriptor] {\color{red}$3_R$} (OR3);

        \node[invisible node] (OR4 in top 1)    at (4.25, -5.25+1+\yadj) {};
        \node[invisible node] (OR4 in top 2)    at (5, -5.25+1+\yadj) {};
        \node[invisible node] (OR4 in top 3)    at (5.75, -5.25+1+\yadj) {};
        \node[shorthand] (OR4)                  at (5, -6+1+\yadj) {OR$_4$};
        \node[invisible node] (OR4 in red)      at (4.25, -6+1+\yadj){};
        
        \path[->] (OR4 in top 1) edge[std, ourBlack] node[edge descriptor] {$3_{3B}$} (OR4);
        \path[->] (OR4 in top 2) edge[std, ourBlue] node[edge descriptor] {$\color{blue}3_B$} (OR4);
        \path[->] (OR4 in top 3) edge[std, ourBlue] node[edge descriptor] {$\color{blue}3_B$} (OR4);
        \path[->] (OR4 in red) edge[std, ourRed] node[edge descriptor] {\color{red}$3_R$} (OR4);

        \def\yadj{-4.125-4+0.875}
        \node[SWAP] (SWAP1)                 at (4, 0+\yadj) {\SWAP};

        \path[->] (OR3) edge[std, ourBlack] node[edge descriptor] {$12_{O3}$} (SWAP1.north west);
        \path[->] (OR4) edge[std, ourBlack] node[edge descriptor] {$12_{O4}$} (SWAP1.north east);

        \node[SWAP] (SWAP2)                 at (2, -0.5+\yadj) {\SWAP};

        \path[->] (OR2) edge[std, ourBlack] node[edge descriptor] {$12_{O2}$} (SWAP2.north west);
        \path[->] (SWAP1.south west) edge[std, ourBlack] (SWAP2.north east);

        \node[SWAP] (SWAP3)                 at (0, -1+\yadj) {\SWAP};

        \path[->] (OR1) edge[std, ourBlack] node[edge descriptor] {$12_{O1}$} (SWAP3.north west);
        \path[->] (SWAP2.south west) edge[std, ourBlack] (SWAP3.north east);

        \node[SWAP] (SWAP4)                 at (4, -1+\yadj) {\SWAP};

        \path[->] (SWAP2.south east) edge[std, ourBlack] (SWAP4.north west);
        \path[->] (SWAP1.south east) edge[std, ourBlack] (SWAP4.north east);

        \node[SWAP] (SWAP5)                 at (2, -1.5+\yadj) {\SWAP};

        \path[->] (SWAP3.south east) edge[std, ourBlack] (SWAP5.north west);
        \path[->] (SWAP4.south west) edge[std, ourBlack] (SWAP5.north east);

        \node[SWAP] (SWAP6)                 at (4, -2+\yadj) {\SWAP};
        
        \path[->] (SWAP5.south east) edge[std, ourBlack] (SWAP6.north west);
        \path[->] (SWAP4.south east) edge[std, ourBlack] (SWAP6.north east);

        \coordinate (in 1)        at (-1, -2.125+\yadj);
        \coordinate (in 2)        at (1, -2.125+\yadj);
        \coordinate (in 3)        at (3, -2.125+\yadj);
        \node[graph node] (in 4)        at (5, -3+\yadj){};

        \path[-] (SWAP3.south west) edge[std, ourBlack] (in 1);
        \path[-] (SWAP5.south west) edge[std, ourBlack] (in 2);
        \path[-] (SWAP6.south west) edge[std, ourBlack] (in 3);
        \path[-] (SWAP6.south east) edge[std, ourBlack] node[edge descriptor] {$12_{O4}$} (in 4);

        \def\yadjnew{5-3.5+0.375}
        \node[shorthand]      (tension)     at (2, -12.5+\yadjnew){$T_e$};
        \node[invisible node] (red)         at (0.75, -0.375-11.75+\yadjnew) {};
        \node[invisible node] (green)       at (0.75, -0.75-11.75+\yadjnew) {};
        \node[invisible node] (purple)      at (0.75, -1.125-11.75+\yadjnew) {};

        \path[->] (red) edge[std, ourRed] node[edge descriptor] {$\color{red}3_R$} (tension);
        \path[->] (green) edge[std, ourGreen] node[edge descriptor] {$\color{applegreen}6_G$} (tension);
        \path[->] (purple) edge[std, ourPurple] node[edge descriptor] {$\color{blue-violet}6_P$} (tension);
        \path[<-] (tension) edge[in = -90, out = 130, std, ourBlack] node[edge descriptor] {$12_{O1}$} (in 1);
        \path[<-] (tension) edge[in = -90, out = 105, std, ourBlack] node[edge descriptor] {$12_{O2}$} (in 2);
        \path[<-] (tension) edge[in = -90, out = 75, std, ourBlack] node[edge descriptor] {$12_{O3}$} (in 3);

        \node[graph node] (out 4)        at (4, -3+\yadj){};
        
        \path[-] (tension) edge[in = -90, out = 50, std, ourBlack] node[edge descriptor] {$12_{O4'}$} (out 4);
        
    \end{tikzpicture}
    }
    
    \caption{
        A DPS polycule which is schedulable iff there is some assignment that simultaneously satisfies at least 3 of $(x_1\lor x_2), (\overline{x_1}), (x_1 \lor \overline{x_2} \lor x_3),$ and $ (\overline{x_3})$ (possible in this case). 
        Note that re-unifying the bottom-most $12_{O4}$ edge with the $12_{O4'}$ edge will make a polycule which is schedulable iff all 4 clauses can be satisfied (and hence has no valid schedule here). Similarly, breaking the $12_{O3}$ edge in the middle and connecting the two new ends to pendant nodes will create a polycule which is schedulable iff 2 clauses can be satisfied.        
        Connections between layers are omitted for clarity but flow from top to bottom, starting with the variable layer, then three layers concerned with the duplication of variables, the \OR layer, the sorting network, and the tensioning layer.
    }

    \label{fig:worked-eg:Pdsz}
\end{figure*} %

\begin{definition}[$\mathcal P_{d\varphi k}$ polycules]
	\label{def:gadgets-assemble}
	The decision polycule $\mathcal P_{d\varphi k}=(P,R,f)$ is constructed in layers as follows:
	\begin{thmenumerate}{def:gadgets-assemble}
	\item \textbf{Variable layer:}\\
		The variable layer consists of a True Clock and a variable gadget with outputs $3_{iR}$ and $3_{iB}$ for each variable $x_i\in X$.
	\item \textbf{Duplication layer:}\\
		The duplication layer duplicates outputs of the variable layer: 
		$D_3$ duplicators create one $3_{iR}$ edge for each $x_i\in C$, and one $3_{iB}$ edge for each $\overline{x_i}\in C$.
		They also create as many $\color{red}3_R$ and $\color{blue}3_B$ edges as are needed, while $D_6$ duplicators do the same for $\color{applegreen}6_G$ and $\color{blue-violet}6_P$ edges.
		Unused edges are connected to pendent nodes.
	\item \textbf{Clause layer:}\\
		The clause layer consists of one \OR gadget for each clause $c_j\in C$. \OR gadgets have three inputs, each corresponding to a literal in $c_j$: $3_{iR}$ for $x_i$, $3_{iB}$ for $\overline{x_i}$. 
		For clauses with less than 3 literals, $\color{blue}3_B$ edges fill the \OR gadget's unused inputs.
	\item \textbf{Sorting network:}\\
		The next layer is a single gadget -- a sorting network.
		The $12_O$ output edges of the clause layer will each be the input to one of $m$ channels, each of which terminates with a $12_O$ output edge. This gadget also consumes $\color{blue}12_B$ and $\color{applegreen}12_G$ edges created by $S_{B12}$ and $S_{G12}$ gadgets.
	\item \textbf{Tension layer:}\\
		The tensioning layer attaches tension gadgets to the leftmost $k$ outputs of the sorting layer -- any other outputs from this layer and any spare inputs to the tensioning layer are connected to pendant nodes.
	\end{thmenumerate}
\end{definition}

\subsection{The True Clock \& Colour Slots}
\label{sec:TrueClock}

\begin{figure}[hbt]
    \centering
    \begin{tikzpicture}[scale=2]
            
        \node[graph node] (T) at (0.25, 0) {T};
        \node[graph node] (x1) at (1.25, 0) {$x_1$};
        \node[graph node] (x2) at (2.25, 0) {$x_2$};
        \node[invisible node] (right) at (3, 0) {};
        \node[invisible node] (low left 1) at (-0.25, -0.75) {};    
        \node[invisible node] (low left 2) at (0.25, -0.75) {};
        \node[invisible node] (low left 3) at (0.75, -0.75) {};
        \node[invisible node] (low mid 1) at (0.875, -0.75) {};
        \node[invisible node] (low mid 2) at (1.625, -0.75) {};
        \node[invisible node] (low right 1) at (1.875, -0.75) {};
        \node[invisible node] (low right 2) at (2.625, -0.75) {};

        \path[-] (T) edge[std, ourGreen] node[edge descriptor] {\color{applegreen}$6_G$} (x1);
        \path[-] (x1) edge[std, ourPurple] node[edge descriptor] {\color{blue-violet}$6_P$} (x2);    
        \path[->] (x2) edge[std, ourGreen] node[edge descriptor] {\color{applegreen}$6_G$} (right);
        \path[->] (T) edge[std, ourPurple] node[edge descriptor] {\color{blue-violet}$6_P$} (low left 1);
        \path[->] (T) edge[std, ourRed] node[edge descriptor] {\color{red}$3_R$} (low left 2);
        \path[->] (T) edge[std, ourBlue] node[edge descriptor] {\color{blue}$3_B$} (low left 3);
        \path[->] (x1) edge[std, ourBlack] node[edge descriptor] {$3_{1R}$} (low mid 1);
        \path[->] (x1) edge[std, ourBlack] node[edge descriptor] {$3_{1B}$} (low mid 2);
        \path[->] (x2) edge[std, ourBlack] node[edge descriptor] {$3_{2r}$} (low right 1);
        \path[->] (x2) edge[std, ourBlack] node[edge descriptor] {$3_{2B}$} (low right 2);

        \node[shorthand] (T short) at (0.25, -1.25) {T};
        \node[invisible node] (short left 1) at (-0.25, -2) {};
        \node[invisible node] (short left 3) at (0.25, -2) {}; 
        \node[invisible node] (short left 4) at (0.75, -2) {};
        
        \path[->] (T short) edge[std, ourPurple] node[edge descriptor] {\color{blue-violet}$6_P$} (short left 1);
        \path[->] (T short) edge[std, ourRed] node[edge descriptor] {\color{red}$3_R$} (short left 3);
        \path[->] (T short) edge[std, ourBlue] node[edge descriptor] {\color{blue}$3_B$} (short left 4);
        
        \node[shorthand] (xi short) at (1.25,-1.25) {$x_1$};
        \node[invisible node] (short right low 1) at (0.875, -2) {};
        \node[invisible node] (short right low 2) at (1.625, -2) {};

        \path[-]  (T short) edge[std, ourGreen] node[edge descriptor] {\color{applegreen}$6_G$} (xi short);

        \node[shorthand] (xi2 short) at (2.25,-1.25) {$x_2$};
        \node[invisible node] (short2 right right) at (3, -1.25) {};
        \node[invisible node] (short2 right low 1) at (1.875, -2) {};
        \node[invisible node] (short2 right low 2) at (2.625, -2) {};

        \path[-] (xi short) edge[std, ourPurple] node[edge descriptor] {\color{blue-violet}$6_P$} (xi2 short);
        \path[->] (xi short) edge[std, ourBlack] node[edge descriptor] {$3_{1R}$} (short right low 1);
        \path[->] (xi short) edge[std, ourBlack] node[edge descriptor] {$3_{1B}$} (short right low 2);
        
        \path[->] (xi2 short) edge[std, ourGreen] node[edge descriptor] {\color{applegreen}$6_G$} (short2 right right);
        \path[->] (xi2 short) edge[std, ourBlack] node[edge descriptor] {$3_{2r}$} (short2 right low 1);
        \path[->] (xi2 short) edge[std, ourBlack] node[edge descriptor] {$3_{2B}$} (short2 right low 2);

    \end{tikzpicture}
    \caption{
    Gadgets for the True Clock and for sample variables $x_1$ and $x_2$ (top, left to right), with shorthand versions shown below. Gadgets are shown connected as they would be in a sample variable layer, and their colours and schedules are discussed in \wref{sec:TrueClock}. Further variables can be added to the right, and must be added in pairs to conserve the $\color{applegreen}6_G$ output edge (though the final variable may be connected to pendant nodes if not otherwise needed).
    }
    \label{fig:variables}
\end{figure}
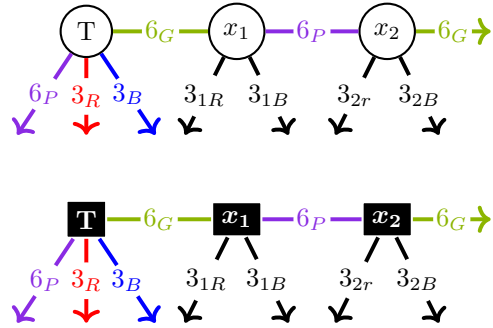 %

\wref{fig:variables} introduces gadgets representing sample variables $x_1$ and $x_2$, as well as a special variable: the True Clock $T$, which acts as a drumbeat for the polycule as a whole. 
Variables, including $T$, have four relationships: $[3, 3, 6, 6]$, so their local schedules must all have the following form: $[\splitaftercomma{3_a, 3_b, 6_a, 3_a, 3_b, 6_b}]$. 
As schedules are cyclic, we can choose to start this schedule with the $3_a$ edge of $T$ without loss of generality. We then assign names to these edges as dictated by the True Clock.

\begin{definition}[Slots]
    \label{def:slots}
	We call days $t\in\N_0$ with $t\equiv 0 \pmod 3$ the \textcolor{red}{red} slots, days $t\equiv 1\pmod 3$  \textcolor{blue}{blue} slots, days $t\equiv 2\pmod 6$ \textcolor{applegreen}{green} slots, and days $t\equiv 5 \pmod 6$ \textcolor{blue-violet}{purple} slots.
\end{definition}

The local schedule of the True Clock is therefore given by
[\color{red}$3_R$\color{black},
\color{blue}$3_B$\color{black},
\color{applegreen}$6_G$\color{black},
\color{red}$3_R$\color{black},
\color{blue}$3_B$\color{black},
\color{blue-violet}$6_P$\color{black}].
As $\color{red}3_R$ is scheduled on day 0, it will always be assigned in red slots, 
with $\color{blue}3_B$, $\color{applegreen}6_G$, and $\color{blue-violet}6_P$ edges also restricted to slots of their respective colours.
We will sometimes represent this by underlining elements or gaps in a schedule, \eg, 
[\color{red}$\SLOT$\color{black},
\color{blue}$\SLOT$\color{black},
\color{applegreen}$\SLOT$\color{black},
\color{red}$\SLOT$\color{black},
\color{blue}$\SLOT$\color{black},
\color{blue-violet}$\SLOT$\color{black}], 
or by referring to edges as being red, blue, green, or purple.

All gadgets introduced below will be constructed such that the lengths of their schedules are integer multiples of 6. 
In the final polycule $\mathcal P_{d\varphi k}$ they will be connected (usually through intermediaries) to the True Clock, such that their edges must stick to certain slots.
\newcommand\slotgood{slot-respecting\xspace}%
To keep correctness proofs of individual gadgets readable, we call a schedule $S$
that schedules all coloured edges in slots of the given colour \emph{\slotgood}.

\paragraph{Drawing Conventions}
\label{para:in/out}
Gadgets are connected by input and output edges, represented by incoming and outgoing arrows respectively -- each one being half of a relationship between two people from different gadgets. 
In addition to their frequencies, input and output edges share restrictions on their permissible schedules, carrying them from one gadget to another as discussed in the proofs associated with each gadget.

In shorthand gadgets, vertical incoming edges are the primary \emph{input} to a gadget, encoding the value of some variable or logical function; horizontal inputs contain edges of fixed colour, which we will refer to as \emph{constants}. 
Similarly, vertical outgoing edges represent primary \emph{outputs} that encode the result of the gadget, while horizontal outgoing edges represent incidentally created constants which may either be used by other gadgets or connected to pendent vertices.

Some incoming and outgoing edges will have end labels of the form ``$\cdot i$'', indicating $i$ connections of the given type, each between different people.

\subsection{Variables}
\label{sec:variables}

\wref{fig:variables} also introduces the gadget for a sample variable, $x_1$, which again has four relationships:
[$3_{1R}, 3_{1B},$
$\color{applegreen}6_G$,
$\color{blue-violet}6_P$].
The key property of variable gadgets is summarized in the following lemma.

\begin{lemma}[Variable gadget schedules]
\label{lem:variable-6-colour}
    Any valid global schedule must yield a local schedule for the variable gadget $x_i$ of the form  
    $[3_a, 3_b, {\color{applegreen}6_G}, 
    3_a, 3_b, {\color{blue-violet}6_P}]$.
\end{lemma}

\begin{proof}
    Note that $x_i$ has local density $D=1$, so $3_{iR}$ and $3_{iB}$ must be scheduled exactly once in each 3-day period, forcing every 3 days to be of the form $[3_a, 3_b, \SLOT]$, and every 6-day schedule to be of the form $[3_a, 3_b, \SLOT, 3_a, 3_b, \SLOT]$; this leaves two remaining slots, which must contain $\color{applegreen}6_G$ and $\color{blue-violet}6_P$. 
    
    The $\color{applegreen}6_G$ edge of the first variable, $x_1$, is shared with $T$ such that it must be green, so schedules for $x_1$
    are of the form
    $[3_a, 3_b, \color{applegreen}6_G\color{black}, 3_a, 3_b, \color{blue-violet}\SLOT\color{black}]$, which must be completed as $[3_a, 3_b, \color{applegreen}6_G\color{black}, 3_a, 3_b, \color{blue-violet}6_P\color{black}]$.  
    This proceeds for $x_2$, whose $\color{blue-violet}6_P$ edge is shared with $x_1$ such that it must be purple, 
    forcing the partial schedule $[3_a, 3_b, \color{applegreen}\SLOT\color{black}, 3_a, $
    $3_b,$
    $\color{blue-violet}6_P\color{black}]$ which likewise must be completed $[3_a, 3_b, \color{applegreen}6_G\color{black}, 3_a, 3_b, \color{blue-violet}6_P\color{black}]$. 
    
    Further pairs of variables are each connected to the $\color{applegreen}6_G$ edge returned by the previous pair, so their schedules must be of the same form.
    Thus, by induction, schedules for each variable gadget $x_i$ must be of the form $[3_a, 3_b, \color{applegreen}6_G\color{black}, 3_a,$
    $ 3_b, \color{blue-violet}6_P\color{black}]$.
\end{proof}

According to \wref{lem:variable-6-colour}, the incident $\color{applegreen}6_G$ edge is used to restrict the valid schedules for $x_1$, leaving two possibilities: 
$[
    \mathunderline{red}{3_{1R}},
    \mathunderline{blue}{3_{1B}},
    \color{applegreen}6_G\color{black},
    \mathunderline{red}{3_{1R}},
    \mathunderline{blue}{3_{1B}},
    \color{blue-violet}6_P\color{black}
]$ 
and
$[
    \mathunderline{red}{3_{1B}},
    \mathunderline{blue}{3_{1R}},
    \color{applegreen}6_G\color{black},
    \mathunderline{red}{3_{1B}},
    \mathunderline{blue}{3_{1R}},
    \color{blue-violet}6_P\color{black}
]$. 
The former schedule, where $3_{1R}$ is scheduled in red slots, corresponds to a variable assignment where $x_1$ is True, whereas the the second schedule corresponds to $x_1$ being assigned False.

This technique of connecting $\color{red}3_R$, $ \color{blue}3_B$, $\color{applegreen}6_G$, or $\color{blue-violet}6_P$ edges to people in a gadget to limit their valid local schedules and force relationships between edges, slots, and particular meanings will be used extensively in what follows. 

Note that $3_{iB}$ has the opposite value to $3_{iR}$, so using $3_{iB}$ edges in the polycule corresponds to the negated literal $\overline{x_i}$, just as $3_{iR}$ edges correspond to the literal $x_i$.

\subsection{Duplication of Variables and Constants}
\label{sec:duplication}

Variables may appear in multiple clauses, while the constant $\color{red}3_R$, $\color{blue}3_B$, $\color{applegreen}6_G$, and $\color{blue-violet}6_P$ edges are used in multiple gadgets, engendering a need for the duplication of variables and constants. 

\subsubsection{3-Duplicators}
\label{sec:D3}

A gadget for duplicating edges with period 3 is shown in \wref{fig:3-splitter} and proven to accurately reproduce 
its input by \wref{lem:D3}.

\begin{figure}[htb]
    \centering
    \begin{tikzpicture}[scale=2]
        \node[invisible node] (top) at (0, 0.75) {};
        
        \node[invisible node] (top in) at (-0.75, 0) {};
        \node[invisible node] (top out) at (0.75, 0) {};
        \node[graph node] (D1) at (0, 0) {a};

        \path[->] (top) edge[std, ourBlack] node[edge descriptor] {$3_a$} (D1);
        \path[->] (top in) edge[std, ourGreen] node[edge descriptor] {\color{applegreen}$6_G$} (D1);
        \path[->] (D1) edge[std, ourPurple] node[edge descriptor] {\color{blue-violet}$6_P$} (top out);

        \node[invisible node] (low in) at (-1.5, -0.75) {};
        \node[graph node] (D2) at (-0.75, -0.75) {b};
        \node[graph node] (D3) at (0, -0.75) {c};
        \node[graph node] (D4) at (0.75, -0.75) {d};
        \node[invisible node] (low out) at (1.5, -0.75) {};

        \path[-] (D1) edge[std, ourBlack] node[edge descriptor] {$9_{b1}$} (D2);
        \path[-] (D1) edge[std, ourBlack] node[edge descriptor] {$9_{b2}$} (D3);
        \path[-] (D1) edge[std, ourBlack] node[edge descriptor] {$9_{b3}$} (D4);

        \path[->] (low in) edge[std, ourPurple] node[edge descriptor] {\color{blue-violet}$6_P$} (D2);
        \path[-] (D2) edge[std, ourGreen] node[edge descriptor] {\color{applegreen}$6_G$} (D3);
        \path[-] (D3) edge[std, ourPurple] node[edge descriptor] {\color{blue-violet}$6_P$} (D4);
        \path[->] (D4) edge[std, ourGreen] node[edge descriptor] {\color{applegreen}$6_G$} (low out);

        \node[invisible node] (output 1) at (-1, -1.5) {};
        \node[invisible node] (repeat 1) at (-0.5, -1.5) {$\cdot2$};
        \node[invisible node] (output 2) at (-0.25, -1.5) {};
        \node[invisible node] (repeat 2) at (0.25, -1.5) {$\cdot2$};        
        \node[invisible node] (output 3) at (0.5, -1.5) {};
        \node[invisible node] (repeat 3) at (1, -1.5) {$\cdot2$};

        \path[->] (D2) edge[std, ourBlack] node[edge descriptor] {$3_a$} (output 1);
        \path[->] (D2) edge[std, ourBlack] node[edge descriptor] {$9_b$} (repeat 1);
        \path[->] (D3) edge[std, ourBlack] node[edge descriptor] {$3_a$} (output 2);
        \path[->] (D3) edge[std, ourBlack] node[edge descriptor] {$9_b$} (repeat 2);
        \path[->] (D4) edge[std, ourBlack] node[edge descriptor] {$3_a$} (output 3);
        \path[->] (D4) edge[std, ourBlack] node[edge descriptor] {$9_b$} (repeat 3);

        \def\yadj{2.125}

        \node[shorthand] (D3 short)             at (2.5, -2.5+\yadj) {$D_{3}$};
        \node[invisible node] (short input a)   at (2.5, -1.75+\yadj) {};
        \node[invisible node] (short input G)   at (1.75, -2.25+\yadj) {};    
        \node[invisible node] (short input P)   at (1.75, -2.75+\yadj) {};
        \node[invisible node] (short output a)  at (2.5, -3.25+\yadj) {$\cdot 3$};
        \node[invisible node] (short output G)  at (3.25, -2.25+\yadj) {};    
        \node[invisible node] (short output P)  at (3.25, -2.75+\yadj) {};
        
        \path[->] (short input a) edge[std, ourBlack] node[edge descriptor] {$3_a$} (D3 short);
        \path[->] (short input G) edge[std, ourGreen] node[edge descriptor] {\color{applegreen}$6_G$} (D3 short);
        \path[->] (short input P) edge[std, ourPurple] node[edge descriptor] {\color{blue-violet}$6_P$} (D3 short);
        \path[->] (D3 short) edge[std, ourBlack] node[edge descriptor] {$3_a$} (short output a);
        \path[->] (D3 short) edge[std, ourGreen] node[edge descriptor] {\color{applegreen}$6_G$} (short output G);
        \path[->] (D3 short) edge[std, ourPurple] node[edge descriptor] {\color{blue-violet}$6_P$}(short output P);
        
    \end{tikzpicture}
\caption{A gadget for duplicating input edges with frequency 3, with a shorthand version below. Note that the input can be duplicated indefinitely many times by repeating the second layer using the $9_{b}$ edges from the previous layer and a $\color{blue-violet}6_P$ or $\color{applegreen}6_P$ edge from two layers above.}
\label{fig:3-splitter}
\end{figure}
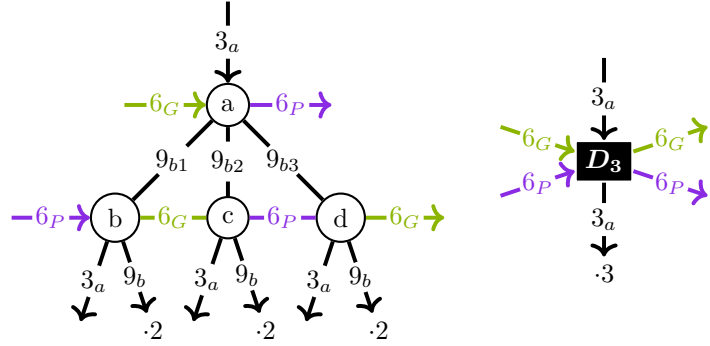 %

\begin{lemma}[3-Duplicator gadget schedules]
    \label{lem:D3}
    Any \slotgood schedule must yield a local schedule for each node in any 3-duplicator gadget $D_3$ of the same form: either
    \[[
        \mathunderline{red}{3_a},
        \mathunderline{blue}{9_{b}},
        \color{applegreen}6_G\color{black},
        \mathunderline{red}{3_a},
        \mathunderline{blue}{9_{b'}},
        \color{blue-violet}6_P\color{black},
        \mathunderline{red}{3_a},
        \mathunderline{blue}{9_{b''}},
        \color{applegreen}6_G\color{black},
        \mathunderline{red}{3_a},
        \mathunderline{blue}{9_{b}},
        \color{blue-violet}6_P\color{black},
        \mathunderline{red}{3_a},
        \mathunderline{blue}{9_{b'}},
        \color{applegreen}6_G\color{black},
        \mathunderline{red}{3_a},
        \mathunderline{blue}{9_{b''}},
        \color{blue-violet}6_P\color{black}
    ]\quad\text{or}
    \] 
    \[[
        \mathunderline{red}{9_{b}},
        \mathunderline{blue}{3_a},
        \color{applegreen}6_G\color{black},
        \mathunderline{red}{9_{b'}},
        \mathunderline{blue}{3_a},
        \color{blue-violet}6_P\color{black},
        \mathunderline{red}{9_{b''}},
        \mathunderline{blue}{3_a},
        \color{applegreen}6_G\color{black},
        \mathunderline{red}{9_{b}},
        \mathunderline{blue}{3_a},
        \color{blue-violet}6_P\color{black},
        \mathunderline{red}{9_{b'}},
        \mathunderline{blue}{3_a},
        \color{applegreen}6_G\color{black},
        \mathunderline{red}{9_{b''}},
        \mathunderline{blue}{3_a},
        \color{blue-violet}6_P\color{black}
    ].\]
\end{lemma}
\begin{proof}
    Each node in $D_3$ has tasks 
    [$3_a,$
    $\color{blue-violet}6_P$, $\color{applegreen}6_G$,
    $9_{b1}, 9_{b2}, 9_{b3}$]
    and has local
    density $D=1$, so each task with frequency $f$ must appear exactly once every $f$ days. 
    Further, in any \slotgood schedule, tasks $\color{blue-violet}6_P$ and $\color{applegreen}6_G$ are scheduled in purple and green slots respectively, forcing partial schedules of the form
        [$\color{red}\SLOT$,
        $\color{blue}\SLOT$,
        $\color{applegreen}6_G$,
        $\color{red}\SLOT$,
        $\color{blue}\SLOT$,
        $\color{blue-violet}6_P$].

    Considering node $a$, note that if incident edge $3_a$ is scheduled on a combination of red and blue days then either it will appear more than once in some 3-day period or its constraint will be violated.
    This demonstrates that schedules must be of the form
    \[[\mathunderline{red}{3_a},
        \color{blue}\SLOT\color{black},
        \color{applegreen}6_G\color{black},
        \mathunderline{red}{3_a},
        \color{blue}\SLOT\color{black},
        \color{blue-violet}6_P\color{black},
        \mathunderline{red}{3_a},
        \color{blue}\SLOT\color{black},
        \color{applegreen}6_G\color{black},
        \mathunderline{red}{3_a},
        \color{blue}\SLOT\color{black},
        \color{blue-violet}6_P\color{black},
        \mathunderline{red}{3_a},
        \color{blue}\SLOT\color{black},
        \color{applegreen}6_G\color{black},
        \mathunderline{red}{3_a},
        \color{blue}\SLOT\color{black},
        \color{blue-violet}6_P\color{black}] 
    \quad\text{or}
    \]
    \[
       [\color{red}\SLOT\color{black},
        \mathunderline{blue}{3_a},
        \color{applegreen}6_G\color{black},
        \color{red}\SLOT\color{black},
        \mathunderline{blue}{3_a},
        \color{blue-violet}6_P\color{black},
        \color{red}\SLOT\color{black},
        \mathunderline{blue}{3_a},
        \color{applegreen}6_G\color{black},
        \color{red}\SLOT\color{black},
        \mathunderline{blue}{3_a},
        \color{blue-violet}6_P\color{black},
        \color{red}\SLOT\color{black},
        \mathunderline{blue}{3_a},
        \color{applegreen}6_G\color{black},
        \color{red}\SLOT\color{black},
        \mathunderline{blue}{3_a},
        \color{blue-violet}6_P\color{black}],
    \] 
    depending on the colour of the incident $3_a$ edge.
    In either case, $9_{b1}$, $9_{b2}$ and $9_{b3}$ must occupy the remaining slots, which match the schedules shown in the lemma for some mapping of $9_{b1}$, $9_{b2}$ and $9_{b3}$ to $9_{b}$, $9_{b'}$ and $9_{b''}$.

    Now consider an arbitrary node $p\neq a$, with an incident $9_{b}$ node. If the $3_a$ edge incident to $a$ is red, the schedule for $n$ must be of the form
    \[[\color{red}\SLOT\color{black},
        \mathunderline{blue}{9_{b}},
        \color{applegreen}6_G\color{black},
        \color{red}\SLOT\color{black},
        \color{blue}\SLOT\color{black},
        \color{blue-violet}6_P\color{black},
        \color{red}\SLOT\color{black},
        \color{blue}\SLOT\color{black},
        \color{applegreen}6_G\color{black},
        \color{red}\SLOT\color{black},
        \mathunderline{blue}{9_{b}},
        \color{blue-violet}6_P\color{black},
        \color{red}\SLOT\color{black},
        \color{blue}\SLOT\color{black},
        \color{applegreen}6_G\color{black},
        \color{red}\SLOT\color{black},
        \color{blue}\SLOT\color{black},
        \color{blue-violet}6_P\color{black}],
    \]
    with $9_{b'}$, $9_{b''}$, and $3_a$ filling the remaining slots. If either $9_{b'}$ or $9_{b''}$ are ever scheduled in a red slot, the constraint of $3_a$ will be violated, therefore either partial schedule leads to the full schedule 
    \[[
        \mathunderline{red}{3_a},
        \mathunderline{blue}{9_{b}},
        \color{applegreen}6_G\color{black},
        \mathunderline{red}{3_a},
        \mathunderline{blue}{9_{b'}},
        \color{blue-violet}6_P\color{black},
        \mathunderline{red}{3_a},
        \mathunderline{blue}{9_{b''}},
        \color{applegreen}6_G\color{black},
        \mathunderline{red}{3_a},
        \mathunderline{blue}{9_{b}},
        \color{blue-violet}6_P\color{black}
        \mathunderline{red}{3_a},
        \mathunderline{blue}{9_{b'}},
        \color{applegreen}6_G\color{black},
        \mathunderline{red}{3_a},
        \mathunderline{blue}{9_{b''}},
        \color{blue-violet}6_P\color{black}
    ]\] 
    matching the schedule of $a$, and of the lemma.
    If the $3_a$ edge incident to $a$ is blue, the same logic applies, with all $3_a$ edges also being blue and the $9_b$ nodes being red.
\end{proof}

\subsubsection{6-Duplicators}
\label{sec:D6}
\wref{fig:6-splitter} and \wref{lem:D6} introduce a gadget which duplicates incident $\color{applegreen}6_G$ or $\color{blue-violet}6_P$ edges.

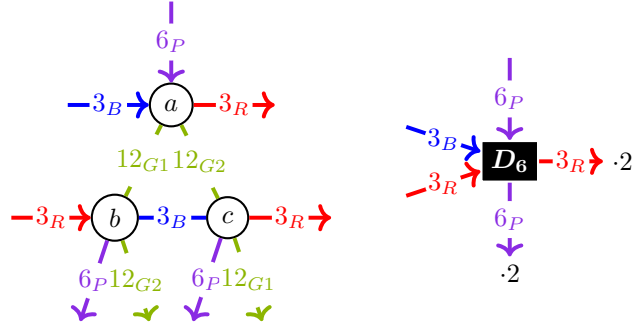
\begin{figure}[hbt]
    \centering
    \begin{tikzpicture}[scale=2]

        \node[invisible node] (top) at (0, 0.75) {};
        \node[invisible node] (left1) at (-0.75, 0) {};
        \node[graph node] (D1) at (0, 0) {$a$};
        \node[invisible node] (right1) at (0.75, 0) {};
        \node[invisible node] (left2) at (-1.125, -0.75) {};
        \node[graph node] (D2) at (-0.375, -0.75) {$b$};
        \node[graph node] (D3) at (0.375, -0.75) {$c$};
        \node[invisible node] (right2) at (1.125, -0.75) {};
        \node[invisible node] (low left 1) at (-0.625, -1.5) {};
        \node[invisible node] (low left 2) at (-0.125, -1.5) {};
        \node[invisible node] (low mid 1) at (0.125, -1.5) {};
        \node[invisible node] (low mid 2) at (0.625, -1.5) {};
        
        \path[->] (top) edge[std, ourPurple] node[edge descriptor] {\color{blue-violet}$6_P$} (D1);
        \path[->] (left1) edge[std, ourBlue] node[edge descriptor] {\color{blue}$3_B$} (D1);
        \path[->] (D1) edge[std, ourRed] node[edge descriptor] {\color{red}$3_R$} (right1);
        \path[-] (D1) edge[std, ourGreen] node[edge descriptor] {\color{applegreen}$12_{G1}$} (D2);
        \path[-] (D1) edge[std, ourGreen] node[edge descriptor] {\color{applegreen}$12_{G2}$} (D3);
        \path[->] (left2) edge[std, ourRed] node[edge descriptor] {\color{red}$3_R$} (D2);
        \path[-] (D2) edge[std, ourBlue] node[edge descriptor] {\color{blue}$3_B$} (D3);
        \path[->] (D3) edge[std, ourRed] node[edge descriptor] {\color{red}$3_R$} (right2);
        \path[->] (D2) edge[std, ourPurple] node[edge descriptor] {\color{blue-violet}$6_P$} (low left 1);
        \path[->] (D2) edge[std, ourGreen] node[edge descriptor] {\color{applegreen}$12_{G2}$} (low left 2);
        \path[->] (D3) edge[std, ourPurple] node[edge descriptor] {\color{blue-violet}$6_P$} (low mid 1);
        \path[->] (D3) edge[std, ourGreen] node[edge descriptor] {\color{applegreen}$12_{G1}$} (low mid 2);        
        \node[invisible node](short top)        at (2.25, 0.75-0.375){};
        \node[shorthand] (short D)              at (2.25, 0-0.375) {$D_6$};
        \node[invisible node](input top)        at (1.5, 0.25-0.375){};
        \node[invisible node](input bottom)     at (1.5, -0.25-0.375){};
        \node[invisible node](unused constant)  at (3, 0-0.375){$\cdot 2$};
        \node[invisible node](output)           at (2.25, -0.75-0.375){$\cdot 2$};

        \path[->] (short top) edge[std, ourPurple] node[edge descriptor] {\color{blue-violet}$6_P$} (short D);
        \path[->] (input top) edge[std, ourBlue] node[edge descriptor] {\color{blue}$3_B$} (short D);
        \path[->] (input bottom) edge[std, ourRed] node[edge descriptor] {\color{red}$3_R$} (short D);
        \path[->] (short D) edge[std, ourRed] node[edge descriptor] {\color{red}$3_R$} (unused constant);
        \path[->] (short D) edge[std, ourPurple] node[edge descriptor] {\color{blue-violet}$6_P$} (output);
        
    \end{tikzpicture}
\caption{A gadget for duplicating \color{blue-violet}$6_P$ \color{black} input edges, with a shorthand version below. Note that if additional \color{blue-violet}$6_P$ \color{black}edges are needed, additional layers can be added, re-using constants from above. Also note that \color{applegreen}$6_G$ \color{black} edges can be duplicated with the same gadget simply by replacing the topmost \color{blue-violet}$6_P$ \color{black} edge with a \color{applegreen}$6_G$ \color{black} edge, forcing the $12_1$ and $12_2$ edges to be purple.}
\label{fig:6-splitter}
\end{figure} %

\begin{lemma}[6-Duplicator gadget schedules]
    \label{lem:D6}
    Any \slotgood schedule must yield a local schedule for each node in any 6-duplicator gadget $D_6$ of the form
    \[[\color{red}3_R\color{black}, 
        \color{blue}3_B\color{black}, 
        \color{applegreen}12_{G1}\color{black}, 
        \color{red}3_R\color{black}, 
        \color{blue}3_B\color{black}, 
        \color{blue-violet}6_P\color{black},  
        \color{red}3_R\color{black}, 
        \color{blue}3_B\color{black}, 
        \color{applegreen}12_{G2}\color{black}, 
        \color{red}3_R\color{black},
        \color{blue}3_B\color{black}, \color{blue-violet}6_P\color{black}].
    \]
\end{lemma}

\begin{proof}
    Each node in $D_6$ has tasks 
    [$\color{red}3_R$, $\color{blue}3_B$, $\color{blue-violet}6_P$, 
    $\color{applegreen}12_{G1}$, $\color{applegreen}12_{G2}$]. It also has
    density $D=1$, so each task with frequency $f$ must appear exactly once every $f$ days. 
    
    Consider node $a$, which has inputs $\color{blue}3_B$ and $\color{blue-violet}6_P$. In any \slotgood schedule these are scheduled in slots of their respective colours, which
    forces partial schedules for $a$ to be of the form 
        [$\color{red}\SLOT$, $\color{blue}3_B$, $\color{applegreen}\SLOT$,
        $\color{red}\SLOT$, $\color{blue}3_B$, $\color{blue-violet}6_P$]. 
    Exactly two of these slots must be filled by $\color{red}3_R$, and if these are not both red slots, the constraint of $\color{red}3_R$ will be violated. Thus, the schedule for $a$ must be of the form
    \[
        [{\color{red}3_R},
         {\color{blue}3_B,} 
         {\color{applegreen}\SLOT},
         {\color{red}3_R}, 
         {\color{blue}3_B},
         {\color{blue-violet}6_P},
         {\color{red}3_R}, 
         {\color{blue}3_B},
         {\color{applegreen}\SLOT},
         {\color{red}3_R},
         {\color{blue}3_B}, 
         {\color{blue-violet}6_P}].
    \]
    Two spaces remain, which must then contain $\color{applegreen}12_{G1}$ and $\color{applegreen}12_{G2}$, as shown in the Lemma.

    Now consider an arbitrary node other than $a$. All such nodes will have inputs
    $\color{applegreen}12_{Ga}$ and $\color{red}3_R$, forcing their partial schedules to be of the form
    \[
       [
        {\color{red}3_R},
        {\color{blue}\SLOT},
        {\color{applegreen}12_{Ga}},
        {\color{red}3_R},
        {\color{blue}\SLOT},
        {\color{blue-violet}\SLOT},
        {\color{red}3_R},
        {\color{blue}\SLOT},
        {\color{applegreen}\SLOT},
        {\color{red}3_R},
        {\color{blue}\SLOT},
        {\color{blue-violet}\SLOT}]
    .\]
    As with $a$, exactly four of these slots must schedule $\color{blue}3_B$ edges, and these must be the blue spaces for the $\color{blue}3_B$ constraint not to be violated, forcing schedules of the form
    \[
        [
        {\color{red}3_R},
        {\color{blue}3_B},
        {\color{applegreen}12_{Ga}},
        {\color{red}3_R},
        {\color{blue}3_B},
        {\color{blue-violet}\SLOT},
        {\color{red}3_R},
        {\color{blue}3_B},
        {\color{applegreen}\SLOT},
        {\color{red}3_R},
        {\color{blue}3_B},
        {\color{blue-violet}\SLOT]}.
    \]
    One of these slots must contain the $\color{applegreen}12_{Gb}$ edge, with the other two scheduling the $\color{blue-violet}6_P$ edge. If the $\color{applegreen}12_{Gb}$ edge is not scheduled in the remaining green slot, the constraint on the $\color{blue-violet}6_P$ edge will be violated, so the schedule must be as shown in the Lemma.
\end{proof}

\subsection{Clauses}
\label{sec:OR}
Clauses in $C$ are disjunctions of at most 3 literals, \ie, the logical OR of at most 3 variables, any of which may be negated. 
A gadget which determines the truth value of a clause given the values of its variables is shown in \wref{fig:OR}. 
\wref{lem:OR} and \wref{rem:OR-no-purple} show that the output of this gadget ($12_O$) can be scheduled in blue slots iff the corresponding clause is evaluated to be True, though it can always be scheduled in green slots. 

\begin{figure}[hbt]
    \centering
    \begin{tikzpicture}[scale=2]
        \node[invisible node] (x1) at (-1, 0) {};
        \node[invisible node] (x2) at (0, 0) {};
        \node[invisible node] (x3) at (1, 0) {};

        \node[graph node] (I1) at (-1, -1) {$I_1$};
        \node[graph node] (I2) at (0, -1) {$I_2$};
        \node[graph node] (I3) at (1, -1) {$I_3$};

        \path[->] (x1) edge[std, ourBlack] node[edge descriptor] {$3_{1R}$} (I1);
        \path[->] (x2) edge[std, ourBlack] node[edge descriptor] {$3_{2R}$} (I2);
        \path[->] (x3) edge[std, ourBlack] node[edge descriptor] {$3_{3R}$} (I3);

        \node[graph node] (OR) at (0, -2) {OR};

        \path[-] (I1) edge[std, ourBlack] node[edge descriptor] {$12_{1}$} (OR);
        \path[-] (I2) edge[std, ourBlack] node[edge descriptor] {$12_{2}$} (OR);
        \path[-] (I3) edge[std, ourBlack] node[edge descriptor] {$12_{3}$} (OR);

        \node[invisible node] (red) at (-1, -2){};

        \path[->] (red) edge[std, ourRed] node[edge descriptor] {\color{red}$3_R$} (OR);

        \node[graph node] (fill1) at (1, -1.75){$f_1$};
        \node[graph node] (fill2) at (1, -2.25){$f_2$};

        \path[-] (OR) edge[std, ourBlack] node[edge descriptor] {$6_{1}$} (fill1);
        \path[-] (OR) edge[std, ourBlack] node[edge descriptor] {$6_{2}$} (fill2);

        \node[invisible node] (out) at (0, -3){};

        \path[->] (OR) edge[std, ourBlack] node[edge descriptor] {$12_O$} (out);

        \node[invisible node] (short x1)  at (2, -0.5) {};
        \node[invisible node] (short x2)  at (3, -0.5) {};
        \node[invisible node] (short x3)  at (4, -0.5) {};
        \node[shorthand]      (short OR)  at (3, -1.5) {OR};
        \node[invisible node] (short r)   at (2, -1.5){};
        \node[invisible node] (short out) at (3, -2.25){};
        
        \path[->] (short x1) edge[std, ourBlack] node[edge descriptor] {$3_{1R}$} (short OR);
        \path[->] (short x2) edge[std, ourBlack] node[edge descriptor] {$3_{2R}$} (short OR);
        \path[->] (short x3) edge[std, ourBlack] node[edge descriptor] {$3_{3R}$} (short OR);
        \path[->] (short r) edge[std, ourRed] node[edge descriptor] {\color{red}$3_R$} (short OR);
        \path[->] (short OR) edge[std, ourBlack] node[edge descriptor] {$12_O$} (short out);
    \end{tikzpicture}
    \caption{A gadget which computes $x_1 \lor x_2 \lor x_3$ (left), along with a shorthand version (right). To compute $x_1 \lor x_2 \lor \overline{x_3}$, replace the incoming $3_{3R}$ edge with a $3_{3B}$ edge. To instead compute $x_1 \lor x_2$, replace the incoming $3_{3R}$ edge with a $\color{blue}3_B$ edge.  According to \wref{lem:OR} and \wref{rem:OR-no-purple}, $12_O$ will be scheduled in green slots if all inputs are assigned False and may be scheduled in green or blue slots if any input is assigned True.}
    \label{fig:OR}
\end{figure}
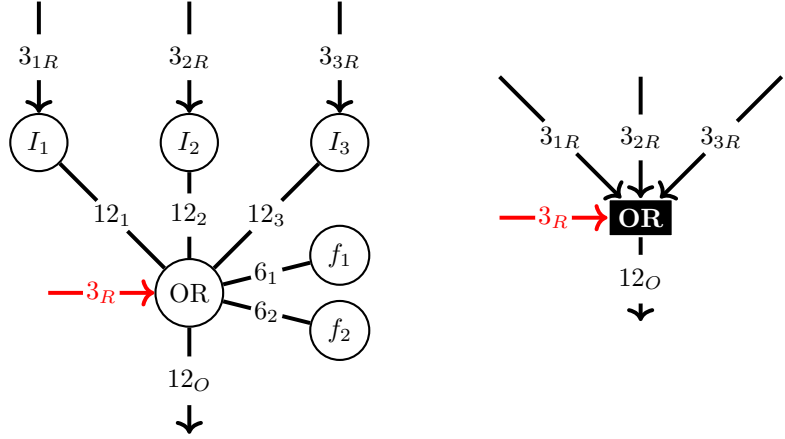 %

Recall the intuitive meaning of slots assigned to the outputs from variable gadgets: a schedule that schedules $3_{iR}$ in red slots corresponds to a variable assignment where $x_i$ is True.
Unfortunately, the output of the OR gadget has to be encoded with a different (and weaker) invariant.

\begin{lemma}[OR gadget schedules]
    \label{lem:OR}
	In any \slotgood schedule where $3_{1R}$, $3_{2R}$, and $3_{3R}$ are scheduled in red or blue slots, the local schedule of the \OR gadget has the following form:
	\begin{itemize}
	\item If all of $3_{1R}$, $3_{2R}$, and $3_{3R}$ are scheduled in blue slots,\\ then $12_O$ will be scheduled in green or purple slots.
	\item If at least one of $3_{1R}$, $3_{2R}$, and $3_{3R}$ are scheduled in red slots,\\ then $12_O$ will be scheduled in green, purple, or blue slots.
	\end{itemize}
\end{lemma} %

The lemma only covers the clause $x_1 \vee x_2 \vee x_3$, but all other clauses can be handled similarly: For a literal $x_j$, we use $3_{jR}$ as input, and for negated literals $\overline {x_j}$, we use the $3_{jB}$ edge instead of $3_{jR}$.

\begin{proof}
    The node labelled \OR has tasks 
    [$\color{red}3_R$,
    $6_1, 6_2, 12_1, 12_2, 12_3, 12_O$] 
    and local density $D=1$, so each task with frequency $f$ must appear exactly once every $f$ days. 
    Any \slotgood schedule must schedule the $\color{red}3_R$ edge in red slots  so schedules for the \OR node must be of the form 
    [\color{red}$3_R$\color{black},
    \color{blue}$\SLOT$\color{black},
    \color{applegreen}$\SLOT$\color{black},
    \color{red}$3_R$\color{black},
    \color{blue}$\SLOT$\color{black},
    \color{blue-violet}$\SLOT$\color{black}].

    Consider an inverter node $I_i$, $i=1,2,3$, which has tasks
    [$3_{iR}$, $12_i$] 
    where $3_{iR}$ is red or blue by assumption.
    Given these constraints, if the input edge $3_{iR}$ is scheduled in red slots, 
    then partial schedules for $I_i$ must be of the form
    [\coloredunderline{red}{$3_{iR}$},
    \color{blue}$\SLOT$\color{black},
    \color{applegreen}$\SLOT$\color{black},
    \coloredunderline{red}{$3_{iR}$},
    \color{blue}$\SLOT$\color{black},
    \color{blue-violet}$\SLOT$\color{black}] 
    and the $12_i$ edge must be scheduled in either green, purple, or blue slots.
    Similarly, if the input edge $3_{iR}$ is scheduled in blue slots 
    then partial schedules for $I_i$ must be of the form
    [\color{red}$\SLOT$\color{black},
    \coloredunderline{blue}{$3_{iR}$},
    \color{applegreen}$\SLOT$\color{black},
    \color{red}$\SLOT$\color{black},
    \coloredunderline{blue}{$3_{iR}$},
    \color{blue-violet}$\SLOT$\color{black}], 
    restricting the $12_i$ edge to green, purple, or red slots.
    However, the $12_i$ edge is also connected to the \OR node which has no empty red slots, further restricting it to green or purple slots.

    Suppose that all inputs $3_{1R}$, $3_{2R}$, $3_{3R}$ are scheduled in blue slots. 
    By the reasoning above, $12_1$, $12_2$, $12_3$ are then scheduled in green or purple slots.
    This will force schedules for the \OR node to be of the form
    \[
    [\color{red}3_R\color{black},
    \color{blue}\SLOT\color{black},
    \mathunderline{applegreen}{12_a},
    \color{red}3_R\color{black},
    \color{blue}\SLOT\color{black},
    \mathunderline{blue-violet}{12_b},
    \color{red}3_R\color{black},
    \color{blue}\SLOT\color{black},
    \mathunderline{applegreen}{12_c},
    \color{red}3_R\color{black},
    \color{blue}\SLOT\color{black},
    \color{blue-violet}\SLOT\color{black}]
    \quad\text{or}
    \]
    \[
    [\color{red}3_R\color{black},
    \color{blue}\SLOT\color{black},
    \mathunderline{applegreen}{12_a},
    \color{red}3_R\color{black},
    \color{blue}\SLOT\color{black},
    \mathunderline{blue-violet}{12_b},
    \color{red}3_R\color{black},
    \color{blue}\SLOT\color{black},
    \color{applegreen}\SLOT\color{black},
    \color{red}3_R\color{black},
    \color{blue}\SLOT\color{black},
    \mathunderline{blue-violet}{12_c}]
   	\]
    (for some mapping of $12_1$, $12_2$, and $12_3$ onto $12_a$, $12_b$, and $12_c$).
    Assume towards a contradiction that the $12_O$ edge is blue. This immediately causes a constraint violation when scheduling either $6_1$ or $6_2$, which demonstrates that the $12_O$ edge cannot be blue, and the schedule for the \OR node must be of the form
    \[[\color{red}3_R\color{black},
    \mathunderline{blue}{6_a},
    \mathunderline{applegreen}{12_a},
    \color{red}3_R\color{black},
    \mathunderline{blue}{6_b},
    \mathunderline{blue-violet}{12_b},
    \color{red}3_R\color{black},
    \mathunderline{blue}{6_a},
    \mathunderline{applegreen}{12_c},
    \color{red}3_R\color{black},    
    \mathunderline{blue}{6_b},
    \mathunderline{blue-violet}{12_d}]\]
    (likewise, for some mapping of $12_1$, $12_2$, and $12_3$, and $12_O$ onto $12_a$, $12_b$, $12_c$, and $12_d$).
    This ensures that the $12_O$ edge must be scheduled in either green or purple slots if all input edges are blue.

    Now suppose that one input edge, $3_{tR}$, is scheduled in red slots, while the other two, $3_{iR}$ and $3_{i'R}$ may be either red or blue.
    Consider the schedule
    \[[\color{red}3_R\color{black},
    \mathunderline{blue}{12_t},
    \mathunderline{applegreen}{6_a},
    \color{red}3_R\color{black},
    \mathunderline{blue}{6_b},
    \mathunderline{blue-violet}{12_i},
    \color{red}3_R\color{black},
    \mathunderline{blue}{12_O},
    \mathunderline{applegreen}{6_a},
    \color{red}3_R\color{black},    
    \mathunderline{blue}{6_b},
    \mathunderline{blue-violet}{12_{i'}}],\]
    where $12_O$ is scheduled in blue slots and no constraints are violated.
    However, even if $3_{tR}$ is assigned red, the corresponding edge $12_t$ may be also be scheduled in green, or purple slots, permitting schedules of the form
    \[[\color{red}3_R\color{black},
    \mathunderline{blue}{6_a},
    \mathunderline{applegreen}{12_a},
    \color{red}3_R\color{black},
    \mathunderline{blue}{6_b},
    \mathunderline{blue-violet}{12_b},
    \color{red}3_R\color{black},
    \mathunderline{blue}{6_a},
    \mathunderline{applegreen}{12_c},
    \color{red}3_R\color{black},    
    \mathunderline{blue}{6_b},
    \mathunderline{blue-violet}{12_d}],\]
    in which $12_O$ is scheduled in green or purple slots. 
    This shows that if one or more inputs are red, then $12_O$ may be scheduled in blue, green, or purple slots.
\end{proof} %

\begin{remark}[No purple output for OR]
\label{rem:OR-no-purple}
	Note that in $\mathcal P_{d\varphi k}$, the $12_O$ edge of an \OR gadget will always be connected to the $I$ node of a $D_{12}$ gadget (discussed below, see \wref{fig:12-duplicator}). 
	Due to this node's $\color{blue-violet} 6_P$ edge, $I$ has no empty purple slots, so the purple slots in \wref{lem:OR} for $12_O$ are not actually possible once the gadgets is part of the overall polycule.
\end{remark}

\subsection{Sorting Networks}
\label{sec:sorting network}
The previous subsection introduced gadgets whose output edges ($12_{O1}, 12_{O2}, \ldots, 12_{Om}$) each capture the truth value of one clause from $\varphi$ by being 
blue \emph{or} green if the clause can be satisfied, but restricted to green slots if it cannot.
The green ambiguity will be resolved by a gadget that applies \emph{Tension} to the system by forcing a subset of $k$ of the $12_O$ edges to be scheduled in blue slots in any valid schedule (\wref{sec:Tension}). 
However, this Tension gadget requires us to pick a \emph{fixed} $k$-subset; to be able to capture the difference between \emph{any} $k$ clauses being potentially blue and at most $k-1$ being blue, we build a \emph{sorting network} gadget.
This moves edges scheduled in green slots to the right, which in turn moves edges scheduled in blue slots to the left. We can then apply tension to the leftmost $k$ outputs to obtain a polycule which is schedulable iff the corresponding $\varphi$ has at least $k$ simultaneously satisfiable clauses.

Sorting networks have two components: \emph{wires}, which carry data, and \emph{comparators}\footnote{sometimes called modules or comparator modules}, which compare two inputs $a$ and $b$, re-ordering them if necessary. More specifically, the left output always contains $\max\{a,b\}$ and the right output $\min\{a,b\}$.
Note that since we are sorting Boolean values, we can compute these as $a\lor b$ and $a\land b$, respectively.

\wref{fig:worked-eg:Pdsz} includes the sorting network for 4 inputs where the wires are represented by \emph{channels} of $12_O$ edges, and comparators by \SWAP gadgets (introduced below).
We use the simple $\Theta(m^2)$-size insertion/bubble sort network~\cite[\S5.3.4]{Knuth1998}.
While asymptotically better sorting networks are obviously available, this simple method is sufficient for our reduction as we are only aiming for polynomial time overall.
We give the general construction here for reference.

\pagebreak[2]

\begin{definition}[DPS sorting networks]
	\label{def:sorting-networks}
	Our DPS sorting network takes $m$ input $12_O$ edges arranged in vertical channels which are connected by layers of \SWAP nodes and terminate with output $12_O$ edges. $12_{O1}$ and $12_{O2}$ are connected by $1$ \SWAP node in layer $0$,  $12_{O2}$ and $12_{O3}$ are connected by $2$ \SWAP nodes in layers $1$ and $-1$, and $12_{Ok}$ and $12_{O(k+1)}$ are connected by $k$ \SWAP nodes in layers $k-1, k-3, \ldots, 1-k$. 
\end{definition}

The goal of the rest of this subsection is to establish the following lemma for the sorting network:

\begin{lemma}[Sorting Network gadget schedules]
\label{lem:sort}
	Consider a DPS sorting network with $m$ each of input $12_O$ edges and output $12_O$ edges. 
	\begin{thmenumerate}{lem:sort}
	\item \label{lem:sort-schedule-properties}
		(sufficient input) for all $\ell$, every \slotgood schedule in which the leftmost $\ell$ output edges are scheduled in blue slots must schedule at least $\ell$ input edges in blue slots.
	\item \label{lem:sort-exists-schedule}
		(sorted output) for any assignment $C:[m]\to \{\text{green}, \text{blue}\}$ of colours to input edges such that at least $\ell$ are blue, there exists a \slotgood schedule $S$ for the sorting Network that assigns in which the leftmost $\ell$ output edges are scheduled in blue slots.
	\end{thmenumerate}
\end{lemma}%

The proof relies on the \SWAP gadgets, which in turn require a few auxiliary gadgets,
so we introduce the latter first before we return to \wref{lem:sort} in \wref{sec:sort}.

\subsubsection{12-Edge Duplicator}
\label{sec:D12}

We begin by introducing our first auxiliary gadget: a frequency-12 edge duplicator, $D_{12}$, shown in \wref{fig:12-duplicator}.

\begin{figure}[hbt]
    \centering
    \begin{tikzpicture}[scale=2]
        \node[invisible node] (input) at (0, 0.25) {};
        \node[graph node] (invert) at (0, -0.75){$I$};
        \node[invisible node] (red) at (-1, -0.25) {};
        \node[invisible node] (blue) at (-1, -0.75) {};
        \node[invisible node] (purple) at (-1, -1.25) {};
        \node[graph node] (spare) at (0.75, -0.75) {};
        \node[graph node] (D12) at (0, -2) {$D_{12}$};
        \node[invisible node] (red2) at (-1, -1.5) {};
        \node[invisible node] (blue2) at (-1, -2) {};
        \node[invisible node] (purple2) at (-1, -2.5) {};
        \node[invisible node] (out1) at (-0.3333, -3) {};
        \node[invisible node] (out2) at (0.3333, -3) {};
        
        \path[->] (input) edge[std, ourBlack] node[edge descriptor] {$12_{O}$} (invert);
        \path[->] (red) edge[std, ourRed] node[edge descriptor] {$\color{red}3_R$} (invert);
        \path[->] (blue) edge[std, ourBlue] node[edge descriptor] {$\color{blue}6_B$} (invert);
        \path[->] (purple) edge[std, ourPurple] node[edge descriptor] {$\color{blue-violet}6_P$} (invert);
        \path[-] (invert) edge[std, ourBlack] node[edge descriptor] {$12_{O'}$} (spare);
        \path[-] (invert) edge[std, ourBlack] node[edge descriptor] {$6_{\overline{O}}$} (D12);
        \path[->] (red2) edge[std, ourRed] node[edge descriptor] {$\color{red}3_R$} (D12);
        \path[->] (blue2) edge[std, ourBlue] node[edge descriptor] {$\color{blue}6_B$} (D12);
        \path[->] (purple2) edge[std, ourPurple] node[edge descriptor] {$\color{blue-violet}6_P$} (D12);
        \path[->] (D12) edge[std, ourBlack] node[edge descriptor] {$12_{O}$} (out1);
        \path[->] (D12) edge[std, ourBlack] node[edge descriptor] {$12_{O'}$} (out2);

        \node[invisible node] (short in)        at (3, -0.5-0.125) {};
        \node[shorthand] (short D12)            at (3, -1.5-0.125) {$D_{12}$};
        \node[invisible node] (short red)       at (2, -1-0.125) {$\cdot2$};
        \node[invisible node] (short blue)      at (2, -1.5-0.125) {$\cdot2$};
        \node[invisible node] (short purple)    at (2, -2-0.125) {$\cdot2$};
        \node[invisible node] (short out)       at (3, -2.5-0.125) {$\cdot2$};

        \path[->] (short in) edge[std, ourBlack] node[edge descriptor] {$12_{O}$} (short D12);
        \path[->] (short red) edge[std, ourRed] node[edge descriptor] {$\color{red}3_R$} (short D12);
        \path[->] (short blue) edge[std, ourBlue] node[edge descriptor] {$\color{blue}6_B$} (short D12);
        \path[->] (short purple) edge[std, ourPurple] node[edge descriptor] {$\color{blue-violet}6_P$} (short D12);
        \path[->] (short D12) edge[std, ourBlack] node[edge descriptor] {$12_{O}$} (short out);

    \end{tikzpicture}
    \caption{A gadget for duplicating input edges with frequency 12 (left), and a shorthand version (right). Note that while it can be easily modified to have three $12_O$ outputs, two are sufficient for our purposes. Also note that the $12_O$ outputs will not both be scheduled concurrently with the $12_O$ input, merely in a slot of the same colour (green or blue).}
    \label{fig:12-duplicator}
\end{figure}
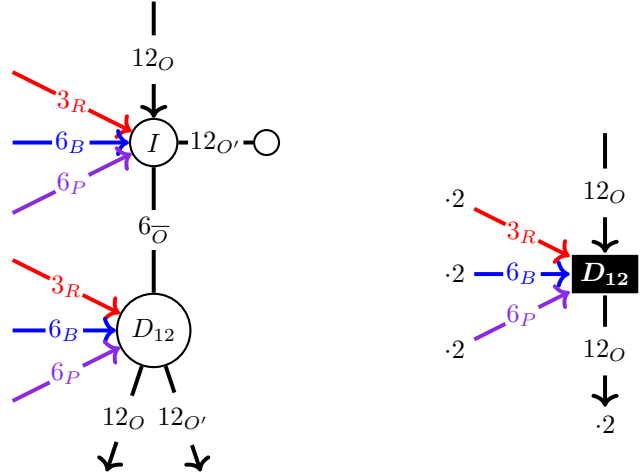%

\begin{lemma}[12-duplicator gadget schedules]
    Any \slotgood schedule must yield a local schedule for the $D_{12}$ gadget where all edges labelled $12_O$ are scheduled in slots of the same colour.
    \label{lem:D12}
\end{lemma}%

\begin{proof}%
    Nodes $I$ and $D_{12}$ each have tasks    
    [$\color{red}3_R$,
    $\color{blue}6_B$,
    $\color{blue-violet}6_P$,
    $6_{\overline{O}}, 12_O, 12_{O'}$] and density $D=1$, so each
    task with frequency $f$ must appear exactly once in every $f$ day period in the schedules for either node.
    Further, tasks $\color{red}3_R$, $\color{blue}6_B$, and $\color{blue-violet}6_P$ are indirectly connected to the True Clock such that their partial schedules must either be of the form
        \[[
        {\color{red}3_R},
        {\color{blue}{6_B}},
        {\color{applegreen}{\SLOT}},
        {\color{red}3_R},
        {\color{blue}{\SLOT}},
        {\color{blue-violet}{6_P}},
        {\color{red}3_R},
        {\color{blue}{6_B}},
        {\color{applegreen}{\SLOT}},
        {\color{red}3_R},
        {\color{blue}{\SLOT}},
        {\color{blue-violet}{6_P}}
        ],
    \quad\text{or}
        \]
        \[[
        {\color{red}3_R},
        {\color{blue}{\SLOT}},
        {\color{applegreen}{\SLOT}},
        {\color{red}3_R},
        {\color{blue}{6_b}},
        {\color{blue-violet}{6_P}},
        {\color{red}3_R},
        {\color{blue}{\SLOT}},
        {\color{applegreen}{\SLOT}},
        {\color{red}3_R},
        {\color{blue}{6_B}},
        {\color{blue-violet}{6_P}} 
        ].\]
    Note that all empty slots in both schedules are either blue or green. 
    
    Assume towards a contradiction that the $6_{\overline{O}}$ edge is scheduled in a mixture of blue and green slots and note that, as it is also scheduled twice in each 12-day period, its constraint must be violated.
    This means that either the $6_{\overline{O}}$ edge is consistently green and all $12_O$ and $12_{O'}$ edges are blue, or the $6_{\overline{O}}$ edge is consistently blue and all $12_O$ and $12_{O'}$ edges are green. Either way, all $12_O$ and $12_{O'}$ edges are scheduled in slots of the same colour. 
\end{proof}%

\subsubsection{OR2 Gadget}

Next, we introduce the $\lor$ gadget, shown in \wref{fig:OR2}. 
Note that the surrounding construction in the sorting layer 
provides weaker guarantees than for the OR gadget of the clause layer,
requiring a slightly different approach.

\begin{figure}[htb]
    \centering
    \begin{tikzpicture}[scale=2]
        \node[invisible node]   (input 1)   at (-0.375, 1) {};
        \node[invisible node]   (input 2)   at (0.375, 1) {};
        \node[graph node]       (OR)        at (0, 0){$\lor$};
        \node[invisible node]   (red)       at (-1, 0.25) {};
        \node[invisible node]   (purple)    at (-1, -0.25) {};
        \node[graph node]       (spare 1)   at (0.75, 0.25) {};
        \node[graph node]       (spare 2)   at (0.75, -0.25) {};
        \node[graph node]       (invert)    at (0, -1.5){$I$};{};
        \node[invisible node]   (red2)      at (-1, -0.5) {};
        \node[invisible node]   (purple2)   at (-1, -1) {};
        \node[invisible node]   (blue2 6)   at (-1, -1.5) {};
        \node[invisible node]   (blue2 12)  at (-1, -2) {};
        \node[invisible node]   (green2 12) at (-1, -2.5) {};
        \node[invisible node]   (out)       at (0, -2.5) {};
        
        \path[->] (input 1) edge[std, ourBlack] node[edge descriptor] {$12_{O1}$} (OR);
        \path[->] (input 2) edge[std, ourBlack] node[edge descriptor] {$12_{O2}$} (OR);
        \path[->] (red) edge[std, ourRed] node[edge descriptor] {$\color{red}3_R$} (OR);
        \path[->] (purple) edge[std, ourPurple] node[edge descriptor] {$\color{blue-violet}6_P$} (OR);
        \path[-] (OR) edge[std, ourBlack] node[edge descriptor] {$12_S$} (spare 1);
        \path[-] (OR) edge[std, ourBlack] node[edge descriptor] {$6_S$} (spare 2);
        \path[-] (OR) edge[std, ourBlack] node[edge descriptor] {$12_{\lor '}$} (invert);
        \path[->] (red2) edge[std, ourRed] node[edge descriptor] {$\color{red}3_R$} (invert);
        \path[->] (purple2) edge[std, ourPurple] node[edge descriptor] {$\color{blue-violet}6_P$} (invert);
        \path[->] (blue2 6) edge[std, ourBlue] node[edge descriptor] {$\color{blue}6_B$} (invert);
        \path[->] (blue2 12) edge[std, ourBlue] node[edge descriptor] {$\color{blue}12_B$} (invert);
        \path[->] (green2 12) edge[std, ourGreen] node[edge descriptor] {$\color{applegreen}12_G$} (invert);
        \path[->] (invert) edge[std, ourBlack] node[edge descriptor] {$12_{O\lor}$} (out);
    
    \begin{scope}[shift={(.5,0)}]
        \node[invisible node]   (short input 1)   at (2-0.375, 1) {};
        \node[invisible node]   (short input 2)   at (2+0.375, 1) {};
        \node[shorthand]        (short OR)        at (2, -0.75){$\lor$};
        \node[invisible node]   (short red)       at (1.25, 0.25) {$\cdot2$};
        \node[invisible node]   (short purple)    at (1.25, -0.25) {$\cdot2$};
        \node[invisible node]   (short blue 6)    at (1.25, -0.75) {};
        \node[invisible node]   (short blue 12)   at (1.25, -1.25) {};
        \node[invisible node]   (short green)     at (1.25, -2) {};
        \node[invisible node]   (short out)       at (2, -2.5) {};

        \path[->] (short input 1) edge[std, ourBlack] node[edge descriptor] {$12_{O1}$} (short OR);
        \path[->] (short input 2) edge[std, ourBlack] node[edge descriptor] {$12_{O2}$} (short OR);
        \path[->] (short red) edge[std, ourRed] node[edge descriptor] {$\color{red}3_R$} (short OR);
        \path[->] (short purple) edge[std, ourPurple] node[edge descriptor] {$\color{blue-violet}6_P$} (short OR);
        \path[->] (short blue 6) edge[std, ourBlue] node[edge descriptor] {$\color{blue}6_B$} (short OR);
        \path[->] (short blue 12) edge[std, ourBlue] node[edge descriptor] {$\color{blue}12_B$} (short OR);
        \path[->] (short green) edge[std, ourGreen] node[edge descriptor] {$\color{applegreen}12_G$} (short OR);
        \path[->] (short OR) edge[std, ourBlack] node[edge descriptor] {$12_{O\lor}$} (short out);
    \end{scope}
        \end{tikzpicture}
    \caption{A gadget for computing $y_1\lor y_2$ (left), along with a shorthand version (right). $12_{O1}$, $12_{O2}$, and $12_{O\lor}$ can always be scheduled in green slots, but may be scheduled in blue slots only if the corresponding logical term is assigned True. Again, ambiguity introduced by True logical terms being scheduled in green slots will be addressed in \wref{sec:Tension}. Note that while \wref{fig:OR} also calculates a logical OR, its incoming edges both have different frequencies and a different mapping from truth values to edge colours.}
    \label{fig:OR2}
\end{figure}
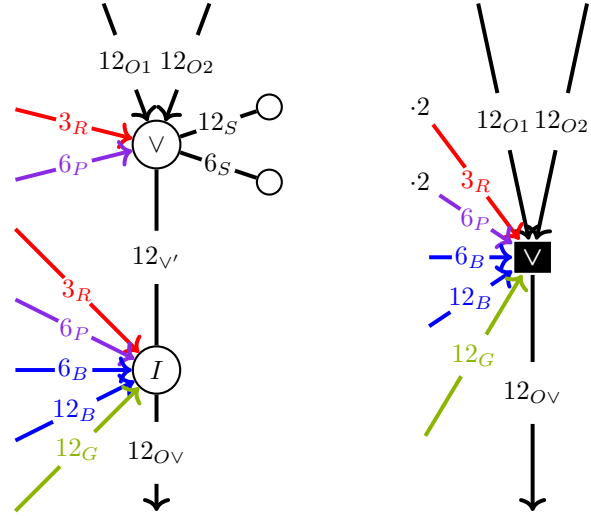%

\begin{lemma}[$\lor$ gadget schedules]
	\label{lem:OR2}
	Any \slotgood schedule must yield a local schedule for the $\lor$ gadget where
	the $12_{O\lor}$ edge is either blue or green. 
	Further, if the $12_{O\lor}$ edge is blue then either $12_{O1}$, $12_{O2}$, or both are also blue.
\end{lemma}%

\begin{proof}%
    Consider the node labelled ``$\lor$'' which, by \wref{fig:OR2}, has tasks
    [$\color{red}3_R$,
    $\color{blue-violet}6_P$,
    $6_S, 12_{S},$
    $12_{\lor'},$
    $12_{O1}, 12_{O2}$] 
    and density $D=1$, so each task with frequency $f$ must appear exactly once in every $f$ day period.
    Further, $\color{red}3_R$ and $\color{blue-violet}6_P$ are indirectly connected to the True Clock such that partial schedules must be of the form
    \[[
        {\color{red}3_R},
        {\color{blue}{\SLOT}},
        {\color{applegreen}{\SLOT}},
        {\color{red}3_R},
        {\color{blue}{\SLOT}},
        {\color{blue-violet}{6_P}},
        {\color{red}3_R},
        {\color{blue}{\SLOT}},
        {\color{applegreen}{\SLOT}},
        {\color{red}3_R},
        {\color{blue}{\SLOT}},
        {\color{blue-violet}{6_P}} 
       ].
	\]
    Suppose that neither $12_{O1}$ nor $12_{O2}$ is blue~-- instead, schedule both in the remaining green slots. This forces partial schedules to be of the form
        \[[
        {\color{red}3_R},
        {\color{blue}{\SLOT}},
        {\mathunderline{applegreen}{12_O}},
        {\color{red}3_R},
        {\color{blue}{\SLOT}},
        {\color{blue-violet}{6_P}},
        {\color{red}3_R},
        {\color{blue}{\SLOT}},
        {\mathunderline{applegreen}{12_{O'}}},
        {\color{red}3_R},
        {\color{blue}{\SLOT}},
        {\color{blue-violet}{6_P}} 
       ],\]
    that is, schedules where all remaining slots are blue, including the slot assigned to $12_{\lor'}$.

    Now consider the alternative, scheduling some $12_O$ input edges ($12_{O1}$ or $12_{O2}$ or both) in blue slots. Under this assumption, partial schedules must be of the form
        \[[
        {\color{red}3_R},
        {\mathunderline{blue}{12_O}},
        {\color{applegreen}{\SLOT}},
        {\color{red}3_R},
        {\color{blue}{\SLOT}},
        {\color{blue-violet}{6_P}},
        {\color{red}3_R},
        {\color{blue}{\SLOT}},
        {\color{applegreen}{\SLOT}},
        {\color{red}3_R},
        {\color{blue}{\SLOT}},
        {\color{blue-violet}{6_P}} 
        ]
    \quad\text{or}\]
        \[[
        {\color{red}3_R},
        {\color{blue}{\SLOT}},
        {\color{applegreen}{\SLOT}},
        {\color{red}3_R},
        {\mathunderline{blue}{12_O}},
        {\color{blue-violet}{6_P}},
        {\color{red}3_R},
        {\color{blue}{\SLOT}},
        {\color{applegreen}{\SLOT}},
        {\color{red}3_R},
        {\color{blue}{\SLOT}},
        {\color{blue-violet}{6_P}} 
       ].\]
    Note that either case allows schedules where $12_{\lor'}$ is green, and also schedules where $12_{\lor'}$ is blue. Thus it is always possible for $12_{\lor'}$ to be scheduled in blue slots, but if $12_{O1}$ or $12_{O2}$ are scheduled in blue slots then $12_{\lor'}$ must be scheduled in green slots.

    Now consider schedules for node $I$, which has tasks
    [$\color{red}3_R$,
    $\color{blue-violet}6_P$,
    $\color{blue}6_B$,
    $\color{blue}12_B$,
    $\color{applegreen}12_G$,
    $12_{\lor'} 12_{O\lor}$], 
    and has density $D=1$, with the same implication.
    Partial schedules for $I$ have two free slots in each 12-day period, with all other slots taken by $\color{red}3_R$,
    $\color{blue-violet}6_P$,
    $\color{blue}6_B$,
    $\color{blue}12_B$, and
    $\color{applegreen}12_G$ -- all of which must be scheduled in slots of their corresponding colour due to their indirect connections to the True Clock.
    These free slots are blue and green, so either $12_{\lor'}$ is blue and $12_{O\lor}$ must be green, or $12_{\lor'}$ is green and $12_{O\lor}$ must be blue. 

    It, therefore, follows that if both $12_{O1}$ and $12_{O2}$ are green then $12_{\lor'}$ will be blue and $12_{O\lor}$ must be green, while if either $12_{O1}$ or $12_{O2}$ are scheduled in blue slots then $12_{\lor'}$ and $12_{O\lor}$ may both be scheduled in either green or blue slots.
\end{proof} %

\subsubsection{AND2 Gadget}

We will now introduce the $\land$ gadget, shown in \wref{fig:AND}. 

\begin{figure}[hbt]
    \centering
    \begin{tikzpicture}[scale=2]
        \node[invisible node]   (input 1)   at (-0.375, 1) {};
        \node[invisible node]   (input 2)   at (0.375, 1) {};
        \node[graph node]       (AND)       at (0, 0){$\land$};
        \node[invisible node]   (red)       at (-0.75, 0.25) {};
        \node[invisible node]   (purple)    at (-0.75, -0.25) {};
        \node[graph node]       (spare 1)   at (0.75, 0.25) {};
        \node[graph node]       (spare 2)   at (0.75, -0.25) {};
        \node[graph node]       (invert)    at (0, -1){$I$};
        \node[graph node]       (spare 3)   at (0.75, -1) {};
        \node[invisible node]   (red2)      at (-0.75, -0.5) {};
        \node[invisible node]   (blue2)     at (-0.75, -1) {};
        \node[invisible node]   (purple2)   at (-0.75, -1.5) {};
        \node[invisible node]   (out)       at (0, -2) {};
        
        \path[->] (input 1) edge[std, ourBlack] node[edge descriptor] {$12_{O1}$} (AND);
        \path[->] (input 2) edge[std, ourBlack] node[edge descriptor] {$12_{O2}$} (AND);
        \path[->] (red) edge[std, ourRed] node[edge descriptor] {$\color{red}3_R$} (AND);
        \path[->] (purple) edge[std, ourPurple] node[edge descriptor] {$\color{blue-violet}6_P$} (AND);
        \path[-] (AND) edge[std, ourBlack] node[edge descriptor] {$12_{S}$} (spare 1);
        \path[-] (AND) edge[std, ourBlack] node[edge descriptor] {$12_{S'}$} (spare 2);
        \path[-] (AND) edge[std, ourBlack] node[edge descriptor] {$6_\land$} (invert);
        \path[->] (red2) edge[std, ourRed] node[edge descriptor] {$\color{red}3_R$} (invert);
        \path[->] (blue2) edge[std, ourBlue] node[edge descriptor] {$\color{blue}6_B$} (invert);
        \path[->] (purple2) edge[std, ourPurple] node[edge descriptor] {$\color{blue-violet}6_P$} (invert);
        \path[-] (invert) edge[std, ourBlack] node[edge descriptor] {$12_{\land'}$} (spare 3);
        \path[->] (invert) edge[std, ourBlack] node[edge descriptor] {$12_{O\land}$} (out);
	\begin{scope}[shift={(0.5,0)}]
        \node[invisible node]   (short input 1)   at (2-0.375, 0.5) {};
        \node[invisible node]   (short input 2)   at (2+0.375, 0.5) {};
        \node[shorthand]        (short AND)       at (2, -0.5){$\land$};
        \node[invisible node]   (short red)       at (1.25, 0) {$\cdot2$};
        \node[invisible node]   (short blue)      at (1.25, -0.5) {};
        \node[invisible node]   (short purple)    at (1.25, -1) {$\cdot2$};
        \node[invisible node]   (short out)       at (2, -1.5) {};

        \path[->] (short input 1) edge[std, ourBlack] node[edge descriptor] {$12_{O1}$} (short AND);
        \path[->] (short input 2) edge[std, ourBlack] node[edge descriptor] {$12_{O2}$} (short AND);
        \path[->] (short red) edge[std, ourRed] node[edge descriptor] {$\color{red}3_R$} (short AND);
        \path[->] (short blue) edge[std, ourBlue] node[edge descriptor] {$\color{blue}6_B$} (short AND);
        \path[->] (short purple) edge[std, ourPurple] node[edge descriptor] {$\color{blue-violet}6_P$} (short AND);
        \path[->] (short AND) edge[std, ourBlack] node[edge descriptor] {$12_{O\land}$} (short out);
    \end{scope}
        \end{tikzpicture}
    \caption{A gadget for computing $y_1\land y_2$ (left), along with a shorthand version (right). Here, $12_{Oi}$ will be scheduled in blue or green slots if $y_i$ is True and restricted to green slots if $y_i$ is False. Similarly, $12_{O\land}$ will be scheduled in blue or green slots if $y_1\land y_2$ is True and restricted to green slots otherwise. As with other such gadgets, ambiguities resulting from edges corresponding to True variables being scheduled in green slots will be addressed in \wref{sec:Tension}.}
    \label{fig:AND}
\end{figure}
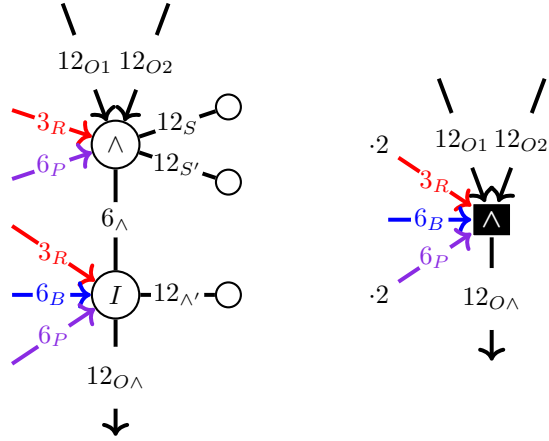 %

\begin{lemma}[$\land$ gadget schedules]
	\label{lem:AND}
	Any \slotgood schedule must yield a local schedule for the $\land$ gadget where
	the $12_{O\land}$ edge is either blue or green.
	Further, if the $12_{O\land}$ edge is blue then both $12_{O1}$ and $12_{O2}$ must also be blue.
\end{lemma} %

\begin{proof}%
    Consider the node labelled ``$\land$'' in \wref{fig:AND}, which has tasks
    $[\splitaftercomma{{\color{red}3_R},
    {\color{blue-violet}6_P},
    6_{\land}, 12_{O1}, 12_{O2}, 12_{S},12_{S'}}]$
    and density $D=1$, forcing each task with frequency $f$ to appear exactly once in every $f$ day period.
    Further, $\color{red}3_R$ and $\color{blue-violet}6_P$ are indirectly connected to the True Clock such that partial schedules must be of the form
        [$\color{red}3_R$,
        $\color{blue}{\SLOT}$,
        $\color{applegreen}{\SLOT}$,
        $\color{red}3_R$,
        $\color{blue}{\SLOT}$,
        $\color{blue-violet}{6_P}$,
        $\color{red}3_R$,
        $\color{blue}{\SLOT}$,
        $\color{applegreen}{\SLOT}$,
        $\color{red}3_R$,
        $\color{blue}{\SLOT}$,
        $\color{blue-violet}{6_P}$].

    Assume towards a contradiction that the $6_{\land}$ edge is scheduled in a mixture of blue and green slots. Under this assumption, either $6_{\land}$ occurs more than once in some 6-day periods, or its constraint must violated. Thus, it must be scheduled consistently in either blue slots or green slots.

    Suppose that some $12_O$ edge, either $12_{O1}$ or $12_{O2}$, is scheduled in a green slot. In this case, partial schedules for $\land$ will either be of the form
        \[[
        {\color{red}3_R},
        {\color{blue}{\SLOT}},
        {\mathunderline{applegreen}{12_O}},
        {\color{red}3_R},
        {\color{blue}{\SLOT}},
        {\color{blue-violet}{6_P}},
        {\color{red}3_R},
        {\color{blue}{\SLOT}},
        {\color{applegreen}{\SLOT}},
        {\color{red}3_R},
        {\color{blue}{\SLOT}},
        {\color{blue-violet}{6_P}} 
        ]
    \quad\text{or}\]
        \[[
        {\color{red}3_R},
        {\color{blue}{\SLOT}},
        {\color{applegreen}{\SLOT}},
        {\color{red}3_R},
        {\color{blue}{\SLOT}},
        {\color{blue-violet}{6_P}},
        {\color{red}3_R},
        {\color{blue}{\SLOT}},
        {\mathunderline{applegreen}{12_O}},
        {\color{red}3_R},
        {\color{blue}{\SLOT}},
        {\color{blue-violet}{6_P}} 
       ].\]
    Either way, $6_\land$ must then be scheduled in a blue slot to avoid constraint violations.    
    If we instead suppose that both $12_{O1}$ and $12_{O2}$ are scheduled in blue slots, then two valid schedules for $\land$ are 
        \[[
        {\color{red}3_R},
        {\mathunderline{blue}{6_{\land}}},
        {\mathunderline{applegreen}{12_S}},
        {\color{red}3_R},
        {\mathunderline{blue}{12_{O1}}},
        {\color{blue-violet}{6_P}},
        {\color{red}3_R},
        {\mathunderline{blue}{6_{\land}}},
        {\mathunderline{applegreen}{12_{S'}}},
        {\color{red}3_R},
        {\mathunderline{blue}{12_{O2}}},
        {\color{blue-violet}{6_P}} 
       ] 
    \quad\text{and}\]        
        \[[
        {\color{red}3_R},
        {\mathunderline{blue}{12_{O1}}},
        {\mathunderline{applegreen}{6_{\land}}},
        {\color{red}3_R},
        {\mathunderline{blue}{12_{O2}}},
        {\color{blue-violet}{6_P}},
        {\color{red}3_R},
        {\mathunderline{blue}{12_{S}}},
        {\mathunderline{applegreen}{6_{\land}}},
        {\color{red}3_R},
        {\mathunderline{blue}{12_{S'}}},
        {\color{blue-violet}{6_P}} 
        ],
       \]
    demonstrating that in this case $6_P$ may be scheduled in either blue or green slots.

    Now consider the node labelled $I$, which has tasks     
    [$\color{red}3_R$,
    $\color{blue}6_B$,
    $\color{blue-violet}6_P$,
    $6_{\land}, 12_{O\land}, 12_{O\land'}$], 
    and density $D=1$ with the same implication as above for the $\land$ node.
    Again, partial schedules of the coloured nodes ($\color{red}3_R$, $\color{blue}6_B$, and $\color{blue-violet}6_P$)
    are restricted by their indirect connections to the True Clock, and must either be of the form
        \[[
        {\color{red}3_R},
        {\color{blue}{6_B}},
        {\color{applegreen}{\SLOT}},
        {\color{red}3_R},
        {\color{blue}{\SLOT}},
        {\color{blue-violet}{6_P}},
        {\color{red}3_R},
        {\color{blue}{6_B}},
        {\color{applegreen}{\SLOT}},
        {\color{red}3_R},
        {\color{blue}{\SLOT}},
        {\color{blue-violet}{6_P}} 
        ]
    \quad\text{or}\]
        \[[
        {\color{red}3_R},
        {\color{blue}{\SLOT}},
        {\color{applegreen}{\SLOT}},
        {\color{red}3_R},
        {\color{blue}{6_B}},
        {\color{blue-violet}{6_P}},
        {\color{red}3_R},
        {\color{blue}{\SLOT}},
        {\color{applegreen}{\SLOT}},
        {\color{red}3_R},
        {\color{blue}{6_B}},
        {\color{blue-violet}{6_P}} 
       ].\]
    In either case, if $6_{\land}$ is scheduled in blue slots then only green slots remain for $12_{O\land}$ and $ 12_{O\land'}$. Likewise, if $6_{\land}$ is scheduled in green slots then $12_{O\land}$ and $ 12_{O\land'}$ must be scheduled in the remaining blue slots. Thus $12_{O\land}$ must be either green or blue, and further can only be blue if $6_{\land}$ is green, which is only possible if both $12_{O1}$ and $12_{O2}$ are blue, as claimed.
\end{proof} %

\subsubsection{Slot Splitting Gadgets}
\label{sec:making-12s}

Several of the gadgets introduced in this section have constant input edges which are scheduled in slots with a specific colour, but with a larger frequency than those produced by the gadgets introduced in \wref{sec:duplication}.
These can be produced by the simple slot-splitter gadgets shown in \wref{fig:B6,B12,G12}.

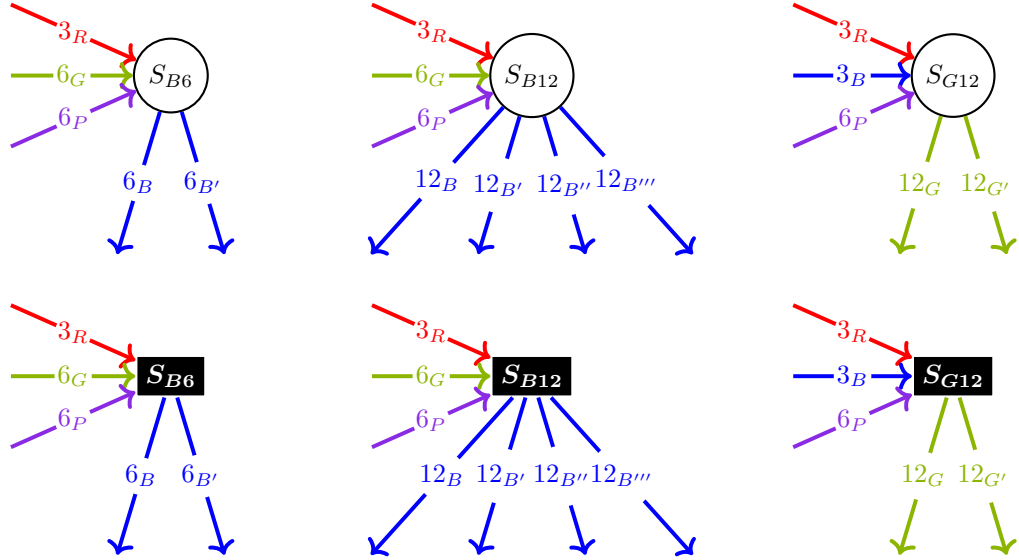
\begin{figure}[hbt]
    \centering
    \begin{tikzpicture}[scale=2]
        \def\xadj{-.9}
    
        \node[graph node] (split 6)      at (\xadj+0, -0.75){$S_{B6}$};
        \node[invisible node] (red 6)    at (\xadj+-1.125, -0.25) {};
        \node[invisible node] (green 6)  at (\xadj+-1.125, -0.75) {};
        \node[invisible node] (purple 6) at (\xadj+-1.125, -1.25) {};
        \node[invisible node] (out 6 1)  at (\xadj+-0.375, -2) {};
        \node[invisible node] (out 6 2)  at (\xadj+0.375, -2) {};

        \path[->] (red 6) edge[std, ourRed] node[edge descriptor] {$\color{red}3_R$} (split 6);
        \path[->] (green 6) edge[std, ourGreen] node[edge descriptor] {$\color{applegreen}6_G$} (split 6);
        \path[->] (purple 6) edge[std, ourPurple] node[edge descriptor] {$\color{blue-violet}6_P$} (split 6);
        \path[->] (split  6) edge[std, ourBlue] node[edge descriptor] {$\color{blue}6_B$} (out 6 1);
        \path[->] (split 6) edge[std, ourBlue] node[edge descriptor] {$\color{blue}6_{B'}$} (out 6 2);

        \node[shorthand]      (short split 6)  at (\xadj+0, -2.75){$S_{B6}$};
        \node[invisible node] (short red 6)    at (\xadj+-1.125, -2.25) {};
        \node[invisible node] (short green 6)  at (\xadj+-1.125, -2.75) {};
        \node[invisible node] (short purple 6) at (\xadj+-1.125, -3.25) {};
        \node[invisible node] (short out 6 1)  at (\xadj+-0.375, -4) {};
        \node[invisible node] (short out 6 2)  at (\xadj+0.375, -4) {};

        \path[->] (short red 6) edge[std, ourRed] node[edge descriptor] {$\color{red}3_R$} (short split 6);
        \path[->] (short green 6) edge[std, ourGreen] node[edge descriptor] {$\color{applegreen}6_G$} (short split 6);
        \path[->] (short purple 6) edge[std, ourPurple] node[edge descriptor] {$\color{blue-violet}6_P$} (short split 6);
        \path[->] (short split 6) edge[std, ourBlue] node[edge descriptor] {$\color{blue}6_B$} (short out 6 1);
        \path[->] (short split 6) edge[std, ourBlue] node[edge descriptor] {$\color{blue}6_{B'}$} (short out 6 2);

        \node[graph node] (split 12)        at (1.5, -0.75){$S_{B12}$};
        \node[invisible node] (red 12)      at (0.375, -0.25) {};
        \node[invisible node] (green 12)    at (0.375, -0.75) {};
        \node[invisible node] (purple 12)   at (0.375, -1.25) {};
        \node[invisible node] (out 12 1)    at (0.375, -2) {};
        \node[invisible node] (out 12 2)    at (1.125, -2) {};
        \node[invisible node] (out 12 3)    at (1.875, -2) {};
        \node[invisible node] (out 12 4)    at (2.625, -2) {};

        \path[->] (red 12) edge[std, ourRed] node[edge descriptor] {$\color{red}3_R$} (split 12);
        \path[->] (green 12) edge[std, ourGreen] node[edge descriptor] {$\color{applegreen}6_G$} (split 12);
        \path[->] (purple 12) edge[std, ourPurple] node[edge descriptor] {$\color{blue-violet}6_P$} (split 12);
        \path[->] (split 12) edge[std, ourBlue] node[edge descriptor] {$\color{blue}12_B$} (out 12 1);
        \path[->] (split 12) edge[std, ourBlue] node[edge descriptor] {$\color{blue}12_{B'}$} (out 12 2);
        \path[->] (split 12) edge[std, ourBlue] node[edge descriptor] {$\color{blue}12_{B''}$} (out 12 3);
        \path[->] (split 12) edge[std, ourBlue] node[edge descriptor] {$\color{blue}12_{B'''}$} (out 12 4);

        \node[shorthand]      (short split 12)  at (1.5, -2.75){$S_{B12}$};
        \node[invisible node] (short red 12)    at (0.375, -2.25) {};
        \node[invisible node] (short green 12)  at (0.375, -2.75) {};
        \node[invisible node] (short purple 12) at (0.375, -3.25) {};
        \node[invisible node] (short out 12 1)  at (0.375, -4) {};
        \node[invisible node] (short out 12 2)  at (1.125, -4) {};
        \node[invisible node] (short out 12 3)  at (1.875, -4) {};
        \node[invisible node] (short out 12 4)  at (2.625, -4) {};
        
        \path[->] (short red 12) edge[std, ourRed] node[edge descriptor] {$\color{red}3_R$} (short split 12);
        \path[->] (short green 12) edge[std, ourGreen] node[edge descriptor] {$\color{applegreen}6_G$} (short split 12);
        \path[->] (short purple 12) edge[std, ourPurple] node[edge descriptor] {$\color{blue-violet}6_P$} (short split 12);
        \path[->] (short split 12) edge[std, ourBlue] node[edge descriptor] {$\color{blue}12_B$} (short out 12 1);
        \path[->] (short split 12) edge[std, ourBlue] node[edge descriptor] {$\color{blue}12_{B'}$} (short out 12 2);
        \path[->] (short split 12) edge[std, ourBlue] node[edge descriptor] {$\color{blue}12_{B''}$} (short out 12 3);
        \path[->] (short split 12) edge[std, ourBlue] node[edge descriptor] {$\color{blue}12_{B'''}$} (short out 12 4);

        \def\xadj{4.3}
        
        \node[graph node] (split 12)        at (\xadj+0, -0.75){$S_{G12}$};
        \node[invisible node] (red 12)      at (\xadj+-1.125, -0.25) {};
        \node[invisible node] (blue 12)     at (\xadj+-1.125, -0.75) {};
        \node[invisible node] (purple 12)   at (\xadj+-1.125, -1.25) {};
        \node[invisible node] (out 12 1)    at (\xadj+-0.375, -2) {};
        \node[invisible node] (out 12 2)    at (\xadj+0.375, -2) {};

        \path[->] (red 12) edge[std, ourRed] node[edge descriptor] {$\color{red}3_R$} (split 12);
        \path[->] (blue 12) edge[std, ourBlue] node[edge descriptor] {$\color{blue}3_B$} (split 12);
        \path[->] (purple 12) edge[std, ourPurple] node[edge descriptor] {$\color{blue-violet}6_P$} (split 12);
        \path[->] (split  12) edge[std, ourGreen] node[edge descriptor] {$\color{applegreen}12_G$} (out 12 1);
        \path[->] (split 12) edge[std, ourGreen] node[edge descriptor] {$\color{applegreen}12_{G'}$} (out 12 2);

        \def\xadj{2.8}
        \def\yadj{-2}

        \node[shorthand]      (short split 12)  at (\xadj+1.5, \yadj-0.75){$S_{G12}$};
        \node[invisible node] (short red 12)    at (\xadj+0.375, \yadj-0.25) {};
        \node[invisible node] (short green 12)  at (\xadj+0.375, \yadj-0.75) {};
        \node[invisible node] (short purple 12) at (\xadj+0.375, \yadj-1.25) {};
        \node[invisible node] (short out 12 1)  at (\xadj+1.125, \yadj-2) {};
        \node[invisible node] (short out 12 2)  at (\xadj+1.875, \yadj-2) {};

        \path[->] (short red 12) edge[std, ourRed] node[edge descriptor] {$\color{red}3_R$} (short split 12);
        \path[->] (short green 12) edge[std, ourBlue] node[edge descriptor] {$\color{blue}3_B$} (short split 12);
        \path[->] (short purple 12) edge[std, ourPurple] node[edge descriptor] {$\color{blue-violet}6_P$} (short split 12);
        \path[->] (short split 12) edge[std, ourGreen] node[edge descriptor] {$\color{applegreen}12_G$} (short out 12 1);
        \path[->] (short split 12) edge[std, ourGreen] node[edge descriptor] {$\color{applegreen}12_{G'}$} (short out 12 2);
        
    \end{tikzpicture}
    \caption{Gadgets for generating $\color{blue}6_B$ edges (top left) and $\color{blue}12_B$ edges (top centre), and $\color{applegreen}12_G$ edges (top right), with shorthand versions below. 
    Excess $\color{blue}6_B$, $\color{blue}12_B$, and $\color{applegreen}12_G$ edges should be connected to pendant nodes.
    }
    \label{fig:B6,B12,G12}
\end{figure} %

\begin{lemma}[Slot splitting gadget schedules]
    \label{lem:split}
    Any \slotgood schedule must yield a local schedule for the slot splitting gadgets
    where $\color{blue}6_{B}$, $\color{blue}12_{B}$, and $\color{applegreen}12_{G}$ edges produced by $S_{B6}$, $S_{B12}$, and $S_{G12}$ nodes must be scheduled in blue, blue, and green slots, respectively.
\end{lemma}

\begin{proof}
    In addition to their respective output edges, the $S_{B6}$ and $S_{B12}$ nodes each have 
        $\color{red}3_R$,
        $\color{applegreen}{6_G}$, and
        $\color{blue-violet}{6_P}$
    input edges, which are connected indirectly to the True Clock such that they must be scheduled in slots of their respective colour, with partial schedules of the form:
        [$\color{red}3_R$,
        $\color{blue}{\SLOT}$,
        $\color{applegreen}{6_G}$,
        $\color{red}3_R$,
        $\color{blue}{\SLOT}$,
        $\color{blue-violet}{6_P}$,
        $\color{red}3_R$,
        $\color{blue}{\SLOT}$,
        $\color{applegreen}{6_G}$,
        $\color{red}3_R$,
        $\color{blue}{\SLOT}$,
        $\color{blue-violet}{6_P}$].
    In these schedules, all remaining slots are blue, so the output edges ($\color{blue}6_B$ and $\color{blue}6_{B'}$ for the $S_{B6}$ node, and $\color{blue}12_{B}$, $\color{blue}12_{B'}$, $\color{blue}12_{B''}$, $\color{blue}12_{B'''}$ for the $S_{B12}$ node) must be scheduled in blue slots.

    Similarly, the $S_{G12}$ node has 
        $\color{red}3_R$,
        $\color{blue}{3_B}$, and
        $\color{blue-violet}{6_P}$
    input edges, and partial schedules of the form
        [$\color{red}3_R$,
        $\color{blue}{3_B}$,
        $\color{applegreen}{\SLOT}$,
        $\color{red}3_R$,
        $\color{blue}{3_B}$,
        $\color{blue-violet}{6_P}$,
        $\color{red}3_R$,
        $\color{blue}{3_B}$,
        $\color{applegreen}{\SLOT}$,
        $\color{red}3_R$,
        $\color{blue}{3_B}$,
        $\color{blue-violet}{6_P}$], 
    where, as above, all remaining slots are green, so the output edges ($\color{applegreen}12_{G}$ and $\color{applegreen}12_{G'}$) must be scheduled in green slots. 
\end{proof}

\subsubsection{SWAP Gadgets}
\label{sec:SWAP}

With these preparations, we are finally able to build the \SWAP gadgets which act as comparators in the sorting network.
The internal structure of \SWAP gadgets is shown in \wref{fig:swapper}. 

\begin{figure}[htp]
    \centering

    \def\SBsixyadj{1.5}

    \resizebox{\textwidth}{!}{
    
    \begin{tikzpicture}[scale=2]
        \node[invisible node] (OR1) at (-1, 0.25) {};
        \node[invisible node] (OR2) at (1, 0.25){};
        \node[shorthand] (Split1) at (-1, -0.75) {$D_{12}$};
        \node[shorthand] (Split2) at (1, -0.75) {$D_{12'}$};
        \node[shorthand] (OR) at (-1, -2.75) {$\lor$};
        \node[shorthand] (AND) at (1, -2.75) {$\land$};
        \node[invisible node] (out1) at (-1, -3.75) {};
        \node[invisible node] (out2) at (1, -3.75){};
        \node[shorthand] (SB6 right) at (2.75, \SBsixyadj-0.75){$S_{B6}$};
        \node[shorthand] (SB6 mid) at (0, \SBsixyadj-0.75){$S_{B6}$};
        \node[shorthand] (SB6 left) at (-2.75, \SBsixyadj-0.75){$S_{B6}$};
        \node[invisible node] (SB12) at (-1.75, -3.25){};
        \node[invisible node] (SG12) at (-1.75, -3.75){};
        \node[invisible node] (red SB6 left)    at (-3.5, \SBsixyadj-0.25) {};
        \node[invisible node] (green SB6 left)  at (-3.5, \SBsixyadj-0.75) {};
        \node[invisible node] (purple SB6 left) at (-3.5, \SBsixyadj-1.25) {};
        \node[invisible node] (red SB6 mid)    at (-0.75, \SBsixyadj-0.25) {};
        \node[invisible node] (green SB6 mid)  at (-0.75, \SBsixyadj-0.75) {};
        \node[invisible node] (purple SB6 mid) at (-0.75, \SBsixyadj-1.25) {};
        \node[invisible node] (red SB6 right)    at (2, \SBsixyadj-0.25) {};
        \node[invisible node] (green SB6 right)  at (2, \SBsixyadj-0.75) {};
        \node[invisible node] (purple SB6 right) at (2, \SBsixyadj-1.25) {};
        \node[invisible node] (red Split1)    at (-1.75, -0.25) {$\cdot 2$};
        \node[invisible node] (purple Split1) at (-1.75, -1.25) {$\cdot 2$};
        \node[invisible node] (red Split2)    at (1.75, -0.25) {$\cdot 2$};
        \node[invisible node] (purple Split2) at (1.75, -1.25) {$\cdot 2$};
        \node[invisible node] (red OR) at (-1.75, -1.75) {$\cdot 2$};
        \node[invisible node] (purple OR) at (-1.75, -2.25) {$\cdot 2$};
        \node[invisible node] (red AND) at (1.75, -1.75) {$\cdot 2$};
        \node[invisible node] (purple AND) at (1.75, -2.25) {$\cdot 2$};            

        \path[->] (OR1) edge[std, ourBlack] node[edge descriptor] {$12_{O1}$} (Split1);
        \path[->] (OR2) edge[std, ourBlack] node[edge descriptor] {$12_{O2}$} (Split2);
        \path[-] (Split1) edge[std, ourBlack] node[edge descriptor] {$12_{O11}$} (OR);
        \path[-] (Split1) edge[bend left, std, ourBlack] node[edge descriptor] {$12_{O12}$} (AND);
        \path[-] (Split2) edge[bend right, std, ourBlack] node[edge descriptor] {$12_{O21}$} (OR);
        \path[-] (Split2) edge[std, ourBlack] node[edge descriptor] {$12_{O22}$} (AND);
        \path[->] (OR) edge[std, ourBlack] node[edge descriptor] {$12_{O\lor}$} (out1);
        \path[->] (AND) edge[std, ourBlack] node[edge descriptor] {$12_{O\land}$} (out2);
        \path[->] (red SB6 left) edge[std, ourRed] node[edge descriptor] {$\color{red}3_R$} (SB6 left);
        \path[->] (green SB6 left) edge[std, ourGreen] node[edge descriptor] {$\color{applegreen}6_G$} (SB6 left);
        \path[->] (purple SB6 left) edge[std, ourPurple] node[edge descriptor] {$\color{blue-violet}6_P$} (SB6 left);
        \path[->] (red SB6 mid) edge[std, ourRed] node[edge descriptor] {$\color{red}3_R$} (SB6 mid);
        \path[->] (green SB6 mid) edge[std, ourGreen] node[edge descriptor] {$\color{applegreen}6_G$} (SB6 mid);
        \path[->] (purple SB6 mid) edge[std, ourPurple] node[edge descriptor] {$\color{blue-violet}6_P$} (SB6 mid);
        \path[->] (red SB6 right) edge[std, ourRed] node[edge descriptor] {$\color{red}3_R$} (SB6 right);
        \path[->] (green SB6 right) edge[std, ourGreen] node[edge descriptor] {$\color{applegreen}6_G$} (SB6 right);
        \path[->] (purple SB6 right) edge[std, ourPurple] node[edge descriptor] {$\color{blue-violet}6_P$} (SB6 right);
        \path[->] (red Split1) edge[std, ourRed] node[edge descriptor] {$\color{red}3_R$} (Split1);
        \path[->] (purple Split1) edge[std, ourPurple] node[edge descriptor] {$\color{blue-violet}6_P$} (Split1);
        \path[->] (red Split2) edge[std, ourRed] node[edge descriptor] {$\color{red}3_R$} (Split2);
        \path[->] (purple Split2) edge[std, ourPurple] node[edge descriptor] {$\color{blue-violet}6_P$} (Split2);
        \path[-] (SB6 left) edge[out = -80, in = 180, std, ourBlue] node[edge descriptor] {$\color{blue}6_B$} (Split1);
        \path[-] (SB6 right) edge[out = -100, in = 0, std, ourBlue] node[edge descriptor] {$\color{blue}6_{B'}$} (Split2);
        \path[-] (SB6 left) edge[out = -90, in = 180, std, ourBlue] node[edge descriptor] {$\color{blue}6_{B'}$} (OR);
        \path[-] (SB6 right) edge[out = -90, in = 0, std, ourBlue] node[edge descriptor] {$\color{blue}6_B$} (AND);
        \path[-] (SB6 mid) edge[out = -90, in = 0, std, ourBlue] node[edge descriptor] {$\color{blue}6_B$} (Split1);
        \path[-] (SB6 mid) edge[out = -90, in = 180, std, ourBlue] node[edge descriptor] {$\color{blue}6_{B'}$} (Split2);
            \path[->] (SB12) edge[std, ourBlue] node[edge descriptor] {$\color{blue}12_B$} (OR);
            \path[->] (SG12) edge[std, ourGreen] node[edge descriptor] {$\color{applegreen}12_G$} (OR);
        \path[->] (red OR) edge[std, ourRed] node[edge descriptor] {$\color{red}3_R$} (OR);
        \path[->] (purple OR) edge[std, ourPurple] node[edge descriptor] {$\color{blue-violet}6_P$} (OR);
        \path[->] (red AND) edge[std, ourRed] node[edge descriptor] {$\color{red}3_R$} (AND);
        \path[->] (purple AND) edge[std, ourPurple] node[edge descriptor] {$\color{blue-violet}6_P$} (AND);

        \node[invisible node] (short in 1) at (-0.5, -4) {};
        \node[invisible node] (short in 2) at (0.5, -4){};
        \node[shorthand] (SWAP) at (0, -5) {\SWAP};
        \node[invisible node] (short out 1) at (-0.5, -6) {};
        \node[invisible node] (short out 2) at (0.5, -6){};
        \node[invisible node] (short red3) at (-1.15, -4.5) {$\cdot 11$};
        \node[invisible node] (short blue12) at (-1.15, -5.5) {};
        \node[invisible node] (short purple6) at (1.15, -4.5){$\cdot 11$};
        \node[invisible node] (short green6) at (1.25, -5) {$\cdot 3$};
        \node[invisible node] (short green12) at (1.15, -5.5){};

        \path[->] (short in 1) edge[std, ourBlack] node[edge descriptor] {$12_{O1}$} (SWAP);
        \path[->] (short in 2) edge[std, ourBlack] node[edge descriptor] {$12_{O2}$} (SWAP);
        \path[->] (SWAP) edge[std, ourBlack] node[edge descriptor] {$12_{O\lor}$} (short out 1);
        \path[->] (SWAP) edge[std, ourBlack] node[edge descriptor] {$12_{O\land}$} (short out 2);
        \path[->] (short red3) edge[std, ourRed] node[edge descriptor] {$\color{red}3_R$} (SWAP);
        \path[->] (short blue12) edge[std, ourBlue] node[edge descriptor] {$\color{blue}12_B$} (SWAP);
        \path[->] (short purple6) edge[std, ourPurple] node[edge descriptor] {$\color{blue-violet}6_P$} (SWAP);
        \path[->] (short green6) edge[std, ourGreen] node[edge descriptor] {$\color{applegreen}6_G$} (SWAP);
        \path[->] (short green12) edge[std, ourGreen] node[edge descriptor] {$\color{applegreen}12_G$} (SWAP);

    \end{tikzpicture}

    }
    
    \caption{A gadget for comparing and re-ordering two input edges (top), and a shorthand version (bottom).
    The $\color{blue}6_{B}$, $\color{blue}12_{B}$, and $\color{applegreen}12_{G}$ edges used by \SWAP nodes 
    come from slot splitting gadgets, which are not shown because they are shared between multiple \SWAP nodes. 
    If the gadget has as many outputs as it has inputs,  $12_{O1}$ and $12_{O2}$ must be scheduled in blue or green slots and $12_{O\lor}$ and $12_{O\land}$ will match them in all cases save one: if $12_{O1}$ is green and $12_{O2}$ is blue, $12_{O\lor}$ will be blue and $12_{O\land}$ will be green.}
    \label{fig:swapper}
\end{figure}
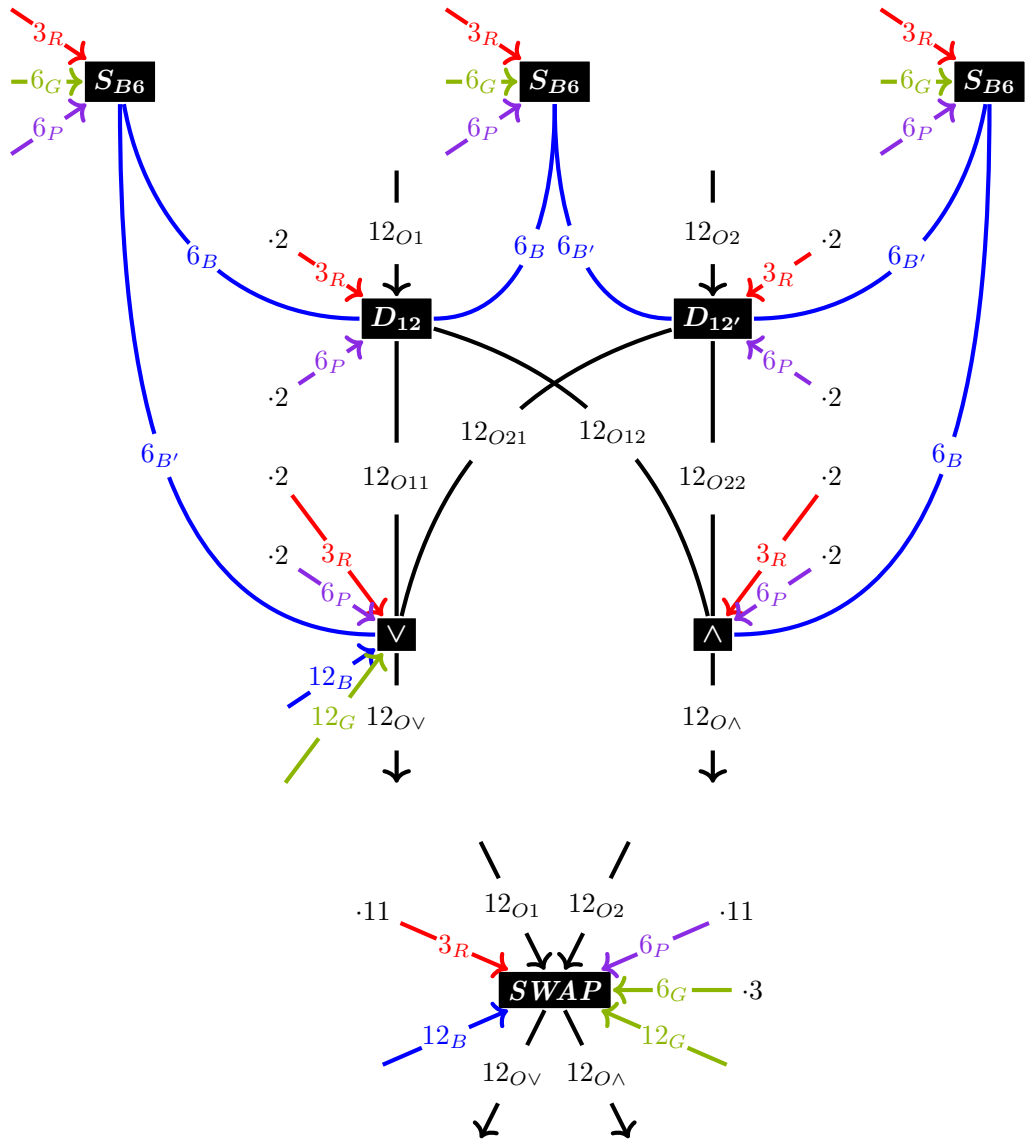%

\begin{lemma}[SWAP gadget schedules]
	\label{lem:swapper}
	Any \slotgood schedule 
	must yield a local schedule for the \SWAP gadget with the following properties: 
	\begin{itemize}
	\item Both output edges, $12_{O\lor}$ and $12_{O\land}$, are scheduled in either green or blue slots.
	\item If $12_{O\lor}$ is blue, then either $12_{O1}$,  $12_{O2}$,  or both are blue.
	\item If $12_{O\land}$ is blue, then both $12_{O1}$ and $12_{O2}$ are blue.
	\end{itemize}
\end{lemma} %

Note that this is again a weaker invariant than a true $\lor$ and $\land$;
\SWAP correctly sorts its inputs, sending green to the right and blue to the left, but it is allowed to ``swallow'' a blue value, turning it green.
Fortunately, this is good enough for our purposes, as we will go on to demonstrate.

\begin{proofof}{\wref{lem:swapper}}
	Assume a \slotgood global schedule $S$ is given, \ie, 
 all coloured edges are scheduled in slots of their respective colours.
	At $D_{12}$, the $12_{O11}$ and $12_{O12}$ edges must be scheduled in slots of the same colour as $12_{O1}$;
	similarly, at $D_{12'}$, edges $12_{O21}$ and $12_{O22}$ must be scheduled in slots of the same colour as $12_{O2}$ (both by \wref{lem:D12}). 
    By \wref{lem:OR2}, $12_{O\lor}$ must be green or blue, and can only be blue if at least one of its inputs, $12_{O11}$ and $12_{O21}$, is also blue.
	Similarly, by \wref{lem:AND}, $12_{O\land}$ must be green or blue, and further can only be blue if both of its inputs, $12_{O12}$ and $12_{O22}$, are also blue.
\end{proofof} %

\subsubsection{Sorting Network Schedules}
\label{sec:sort}

We now finally show that under the assumptions enforced on the schedule by the surrounding gadgets,
our sorting network correctly groups green edges on the right, with enough blue edges on the left to suit our purposes.

\begin{proofof}{\wpref{lem:sort}}
	(a) Assume a \slotgood schedule $S$ that schedules the leftmost $\ell$ output edges in blue slots.
	Then, starting at the outputs of the sorting network, move backwards through the sorting network, one \SWAP gadget at a time (sweeping from bottom to top in \wref{fig:worked-eg:Pdsz}).
	We will show by induction that while moving through the sorting network in this way, there are always at least $\ell$ blue edges between all channels.
    Most importantly, this holds at the inputs -- implying the claim.
	
	The inductive basis at the output level is true by assumption.
	When moving over each \SWAP gadget, we replace the two outputs of that \SWAP with its two inputs.
	By \wref{lem:swapper}, there are several different valid colour combinations for the inputs, 
	but we always have at least as many blue values among the inputs as are among the outputs:
	If both outputs are green or blue, this is obvious. If exactly one output is blue it will either be the $12_{O\lor}$ input (implying that at least one input is blue), or the $12_{O\land}$ input (implying that both inputs are blue, but the $\lor$ gadget chose a green output). 
	Hence the number of blue values cannot drop when moving over a \SWAP gadget.
	
	(b) We build the local \slotgood schedule inductively by again moving through the sorting network one \SWAP node at a time, only now we will move forwards (that is, top to bottom in \wref{fig:worked-eg:Pdsz}).
    We will also restrict $S$ to the special case $S'$, where each \SWAP node has as many blue output $12_O$ edges as it has blue input $12_O$ edges (possible due to \wref{lem:swapper}). Note that not all schedules satisfying the lemma will necessarily have this property, but it is sufficient to show that one such schedule exists.

	In this class of schedules $S'$, \SWAP gadgets behave exactly like the comparator modules in a sorting network for binary inputs: if both inputs are blue then both outputs will be blue, if both inputs are green then both outputs will be green, and if exactly one input is blue then the $12_{O\lor}$ output will be blue and the $12_{O\land}$ output will be green (all by \wref{lem:swapper}).
	By the correctness of the insertion/bubble sort network, we hence end up with an output layer with exactly $\ell$ blue edges on the left, followed by $m-\ell$ green edges on the right.
\end{proofof}%

\subsection{Tension}
\label{sec:Tension}
\wref{lem:sort} assumes that the leftmost $k$ output edges of a sorting network must be scheduled in blue slots -- we ensure this by connecting Tension gadgets (shown in \wref{fig:Tension}) to these output edges.

\begin{figure}[hbt]
    \centering
    \begin{tikzpicture}[scale=2]
        \node[graph node]     (tension)     at (1.5, -0.75){$T_e$};
        \node[invisible node] (red)         at (0.5-0.125, -0.25) {};
        \node[invisible node] (green)       at (0.5-0.125, -0.75) {};
        \node[invisible node] (purple)      at (0.5-0.125, -1.25) {};
        \node[invisible node] (in 1)        at (0.5-0.125, 0.25) {};
        \node[invisible node] (in 2)        at (1.16667-0.0625, 0.25) {};
        \node[invisible node] (in 3)        at (1.83333+0.0625, 0.25) {};
        \node[invisible node] (in 4)        at (2.5+0.125, 0.25) {};

        \path[<-] (tension) edge[std, ourBlack] node[edge descriptor] {$12_{O1}$} (in 1);
        \path[<-] (tension) edge[std, ourBlack] node[edge descriptor] {$12_{O2}$} (in 2);
        \path[<-] (tension) edge[std, ourBlack] node[edge descriptor] {$12_{O3}$} (in 3);
        \path[<-] (tension) edge[std, ourBlack] node[edge descriptor] {$12_{O4}$} (in 4);
        \path[->] (red) edge[std, ourRed] node[edge descriptor] {$\color{red}3_R$} (tension);
        \path[->] (green) edge[std, ourGreen] node[edge descriptor] {$\color{applegreen}6_G$} (tension);
        \path[->] (purple) edge[std, ourPurple] node[edge descriptor] {$\color{blue-violet}6_P$} (tension);
        
        \node[shorthand]      (tension)     at (4.5, -0.75){$T_e$};
        \node[invisible node] (red)         at (3.5-0.125, -0.25) {};
        \node[invisible node] (green)       at (3.5-0.125, -0.75) {};
        \node[invisible node] (purple)      at (3.5-0.125, -1.25) {};
        \node[invisible node] (in 1)        at (3.5-0.125, 0.25) {};
        \node[invisible node] (in 2)        at (4.16667-0.0625, 0.25) {};
        \node[invisible node] (in 3)        at (4.83333+0.0625, 0.25) {};
        \node[invisible node] (in 4)        at (5.5+0.125, 0.25) {};

        \path[<-] (tension) edge[std, ourBlack] node[edge descriptor] {$12_{O1}$} (in 1);
        \path[<-] (tension) edge[std, ourBlack] node[edge descriptor] {$12_{O2}$} (in 2);
        \path[<-] (tension) edge[std, ourBlack] node[edge descriptor] {$12_{O3}$} (in 3);
        \path[<-] (tension) edge[std, ourBlack] node[edge descriptor] {$12_{O4}$} (in 4);
        \path[->] (red) edge[std, ourRed] node[edge descriptor] {$\color{red}3_R$} (tension);
        \path[->] (green) edge[std, ourGreen] node[edge descriptor] {$\color{applegreen}6_G$} (tension);
        \path[->] (purple) edge[std, ourPurple] node[edge descriptor] {$\color{blue-violet}6_P$} (tension);
    \end{tikzpicture}
    \caption{A gadget (left) which applies tension to four inputs, ensuring that either $12_{O1}$, $12_{O2}$, $12_{O3}$, and $12_{O4}$ are scheduled in blue slots or no schedule can be found for $T_e$. To apply tension to less than 4 inputs, connect any unneeded inputs to its own pendant node.}
    \label{fig:Tension}
\end{figure}
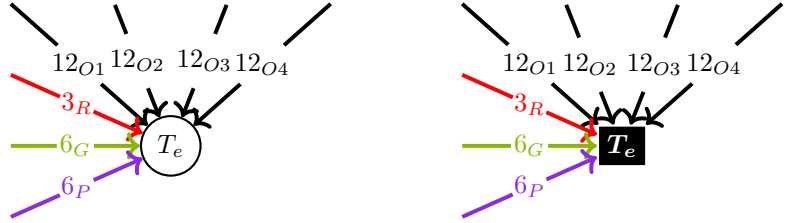 %
\begin{lemma}[Tension gadget schedules]
	\label{lem:tension}
	Any \slotgood schedule must yield a local schedule for a tension gadget $T_e$ where
	each $12_O$ edge must be scheduled in blue slots.
\end{lemma}%
\begin{proof}
    The single node of each tension gadget $T_e$, as shown in \wref{fig:Tension}, has tasks
    [$\color{red}3_R$,
    $\color{applegreen}6_G$,
    $\color{blue-violet}6_P$,
    $12_{O1}, 12_{O2}, 12_{O3}, 12_{O4}$] 
    and density $D=1$, so each task with frequency $f$ must appear exactly once in every $f$ day period. 
    Further, the $\color{red}3_R$, $\color{applegreen}6_G$, and $\color{blue-violet}6_P$ edges are indirectly connected to the True Clock such that they must be scheduled in slots of their respective colours. Partial schedules for $T_e$ nodes must therefore be of the form
    [$\color{red}3_R$,
    $\color{blue}\SLOT$,
    $\color{applegreen}6_G$,
    $\color{red}3_R$,
    $\color{blue}\SLOT$,
    $\color{blue-violet}6_P$,
    $\color{red}3_R$,
    $\color{blue}\SLOT$,
    $\color{applegreen}6_G$,
    $\color{red}3_R$,
    $\color{blue}\SLOT$,
    $\color{blue-violet}6_P$].
    All empty slots in schedules of this form are blue, so the remaining $12_{O1}, 12_{O2}, 12_{O3}, 12_{O4}$ edges must be scheduled in blue slots.
\end{proof}%

\subsection{Correctness Proof of Reduction}

With these preparations we can, at long last, prove \wref{lem:dps=3sat}, and hence complete the proof of \wref{thm:sat-inapprox}.

\begin{proofof}{\wpref{lem:dps=3sat}}
	Consider a polycule $\mathcal P_{d\varphi k}$ built from some 3-CNF formula $\varphi$ using the algorithm defined by \wref{def:gadgets-assemble}. 
	
    First, suppose that there is a variable assignment $v:X\to\{\mathrm{True},\mathrm{False}\}$ such that clauses $c_{i_1},\ldots,c_{i_k}$ evaluate to True under $v$ ($1 \le i_1 < i_2 <\cdots<i_k \le m$).
	We construct a \slotgood schedule $S$ for $\mathcal P_{d\varphi k}$ as follows:
	At the variable gadget for each $x_i\in X$, schedule $3_{iR}$ in red slots if $v(x_i) = \mathrm{True}$ and in the blue slots otherwise.
	This fixes a schedule for all edges in the variable and duplication layers (by \wref{lem:D3} and \wref{lem:D6}).
    In the clause layer, we schedule the $12_O$ output edges of the \OR gadgets for clauses $c_{i_j}$, $j=1,\ldots,k$, in blue slots, and schedule all other $12_O$ output edges in green slots 
    (Noting that, by \wref{lem:OR}, this yields a valid schedule for all \OR gadgets).
	In the sorting layer, we now have at least $k$ blue inputs and at most $m-k$ green inputs.
	By \wref{lem:sort-exists-schedule}, we can extend this schedule to the sorting layer such that
	the leftmost $k$ output edges of the sorting layer are scheduled in blue slots.
	This then also yields a valid schedule for the Tension gadgets (\wref{lem:tension}).
	Overall, this shows that a schedule $S$ for $\mathcal P_{d\varphi k}$ exists.
	
	Now assume that we are given a schedule $S$ for $\mathcal P_{d\varphi k}$.
	By applying \wref{lem:variable-6-colour} to the True Clock, we can assign coloured slots to $S$. It then follows from the construction of $\mathcal P_{d\varphi k}$ that $S$ must be \slotgood.
	Set $v(x_i) = \mathrm{True}$ if $S$ schedules $3_{iR}$ in red slots and $v(x_i) = \mathrm{False}$ otherwise.
    Next, show that $v$ satisfies at least $k$ clauses in $\varphi$:
	By \wref{lem:tension}, the $k$ leftmost output edges of the sorting layer must all be scheduled in blue slots.
	By \wref{lem:sort-schedule-properties}, this means that there are (at least) $k$ output edges of the clause layer that are scheduled in blue slots, say those for clauses $c_{i_1},\ldots,c_{i_k}$ for some $1 \le i_1 < i_2 <\cdots<i_k \le m$.
	Considering this along with \wref{lem:OR} implies that for each $c_{i_j}$, at least one input is scheduled in a red slot. By the definition of $v$, this means that clauses $c_{i_1},\ldots,c_{i_k}$ all evaluate to True under $v$, and the claim follows.

	It remains to argue that our reduction can be realised in polynomial time.
	The size of the polycule $\mathcal P_{d\varphi k}$ is clearly polynomial in the size of the formula, with the quadratic sorting network contributing the most persons. 
	All other gadgets, including the \SWAP nodes composing the sorting network, have constant size and thus are easy to implement in polynomial time.
\end{proofof}

\plaincenter{{\bfseries%
	\textcolor{red}{*}%
	\qquad%
	\textcolor{blue}{*}%
	\qquad%
	\textcolor{applegreen}{*}%
	\qquad%
	\textcolor{red}{*}%
	\qquad%
	\textcolor{blue}{*}%
	\qquad%
	\textcolor{blue-violet}{*}%
}}

\section{Computational Complexity}
\label{sec:computational-complexity}

One proof of the NP-hardness of 
the Decision Poly Scheduling (DPS) Problem is that it contains Pinwheel Scheduling as a special case, an NP-hard problem~\cite{pinwheelNPhard}.
We show in \wref{sec:unweighted-poly-scheduling} that OPS also contains the Chromatic Index problem as a special case, which gives another proof of the NP-hardness of DPS using the conversion in \ref{lem:ops-to-dps}.
Since all good things come in threes, our inapproximability result in \wref{sec:inapproximability} gives a third independent proof of NP-hardness by reducing 3SAT to DPS.

Upper bounds on the the complexity of DPS are much less clear.
Similar to other periodic scheduling problems, the characterizing the computational complexity 
of Poly Scheduling is complicated by the fact that there are feasible instances that require an exponentially large schedule.
It is therefore not clear whether Decision Poly Scheduling is in NP since no succinct Yes-certificates are known;
this is unknown even for the more restricted Pinwheel Scheduling Problem~\cite{Kawamura2023}.

The following simple algorithm shows that DPS is at least in PSPACE (see also \cite{GasieniecSmithWild2022}, \cite{KawamuraSoejima2020}):
Given the polycule $\mathcal P_d = (P,R,f)$ with $|P|=n$ and $|R|=m$, 
construct the \emph{configuration graph} $\mathcal G_c = (V,E)$, where $V$ consists of ``countdown vectors'' listing for each edge $e$ how many days remain before $e$ has to be scheduled again. $E$ has an edge for every maximal matching $M$ in $\mathcal M(P,R)$, and leads to a successor configuration where all $e\in M$ have their urgency reset to $f(e)$ and all $e\notin M$ have their countdown decremented.
Feasible schedules for $\mathcal P_d$ correspond to infinite walks in the finite $\mathcal G_c$, and hence must contain a cycle.
Conversely, any cycle forms a valid periodic schedule.
Our algorithm for DPS thus checks in time $O(|V|+|E|)$ whether $\mathcal G_c$ contains a cycle.

The configuration graph $\mathcal G_c$ has single exponential size:
$V = \{ (u_e)_{e\in R} : u_e \in [0..f(e)]\}$ and $E$ has an edge for every matching in $(P,R)$.
So $|E| \le |V| \cdot 2^m$ (since we have at most $2^{|E|}$ matchings) and
$|V| \le \prod_{e\in R} f(e)$.
To further bound this, we use that all $f(e)$ need to be encoded explicitly in binary in the input.
$\prod_{e\in R} f(e) \le \prod_{e\in R} 2^{|f_e|} = 
2^{\sum |f_e|} 
\le 2^N$ for $N$ 
the size of the encoding of the input.

To obtain a PSPACE algorithm, we use the polylog-space $s$-$t$-connectivity algorithm (using Savich's Theorem on the NL-algorithm that guesses the next vertex in the path) on $\mathcal G_c$, computing the required part of the graph on-the-fly when queried; this yields overall polynomial space.

\section{Unweighted Poly Scheduling \& Edge Coloring}
\label{sec:unweighted-poly-scheduling}

Given an OPS instance $\mathcal P_o = (P,R,g)$, one can always obtain a feasible schedule from a proper \emph{edge colouring} $c:E\to[C]$ of the graph $(P,R)$: any round-robin schedule of the $C$ colours is a valid schedule for $\mathcal P_o$, and the number of colours becomes the separation between visits. More formally,
we can define a schedule $S$ via $S(t) = \{e\in R : c(e) \equiv t \pmod C\}$.
An example is shown in \wref{fig:unweighted}.

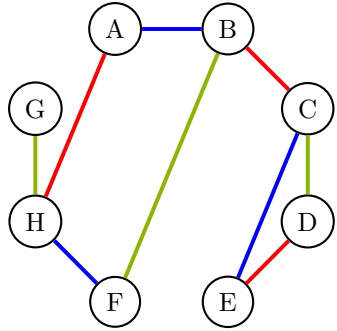
\begin{figure}[htb]
    \centering\begin{tikzpicture}[scale=1.5]
        \node[graph node] (A) at (0, 0) {A};
        \node[graph node] (B) at (1, 0) {B};
        \node[graph node] (C) at (1.707, -0.707) {C};
        \node[graph node] (D) at (1.707, -1.707) {D};
        \node[graph node] (E) at (1, -2.42) {E};
        \node[graph node] (F) at (0, -2.42) {F};
        \node[graph node] (G) at (-0.707, -0.707) {G};
        \node[graph node] (H) at (-0.707, -1.707) {H};

        \path[-] (A) edge[std, C2] (B);
        \path[-] (B) edge[std, C3] (C);
        \path[-] (H) edge[std, C1] (G);
        \path[-] (A) edge[std, C3] (H);
        \path[-] (C) edge[std, C2] (E);
        \path[-] (C) edge[std, C1] (D);
        \path[-] (F) edge[std, C2] (H);
        \path[-] (F) edge[std, C1] (B);
        \path[-] (D) edge[std, C3] (E);
        
    \end{tikzpicture}
	\caption{%
		An unweighted polyamorous scheduling instance (that is, an OPS instance where all edges have growth rate $1$). Edge colours show one optimal schedule, where every edge is visited exactly every three days: 
		[\color{red}3\color{black},
		\color{blue}3\color{black},
		\color{applegreen}3\color{black}], 
		\ie, all red edges are scheduled on days $t$ with $t\equiv 0 \pmod 3$, all blue edges when $t\equiv 1 \pmod 3$ and green edges for $t\equiv 2 \pmod 3$.
	}
	\label{fig:unweighted}
\end{figure}

Such a schedule can yield an arbitrarily bad solution to general instances of $\mathcal P_o$, but it gives optimal solutions for a special case:
The non-hierarchical polycule $\mathcal P_u$, which is an OPS polycule where all growth rates are $g_{i,j} = 1$ (\ie, an unweighted graph).

Recall that any graph with maximal degree $\Delta$ can be edge-coloured with at most $\Delta + 1$ colours and clearly needs at least $\Delta$ colours. 

\begin{proposition}[Unweighted OPS = edge coloring]
\label{pro:unweighted}
	An unweighted OPS problem admits a schedule with heat $h$ if and only if the corresponding graph
	is $h$-edge-colourable.
\end{proposition}
\begin{proof}
First note that any $k$-edge-colouring immediately corresponds to a schedule that visits every edge every $k$ days, since we can schedule all edges $e$ with $c(e) = i$ on days $t\equiv i \pmod C$.
Moreover, any schedule with height $h$ must visit every edge at least once within the first $h$ days (otherwise it would grow to desire $>h\cdot 1$). We can therefore assign $h$ colours according to these first $h$ days of the schedule; some edges might receive more than one colour, but we can use any of these and retain a valid colouring using $h$ colours.
\end{proof}

Since it is NP-complete to decide whether a graph has chromatic index~$\Delta$ even when the graph is 3-regular~\cite{Holyer1981}, unweighted Poly Scheduling is NP-hard.
This provides a second restricted special case of the problem that is NP-hard.
This observation also gives us the inapproximability result stated in \wref{thm:color-inapprox}:

\begin{proofof}{\wpref{thm:color-inapprox}}
	Assume that there is a polynomial-time algorithm $A$ that achieves an approximation ratio of $\frac43-\varepsilon$ for some $\varepsilon>0$.
	Given an input $(V,E)$ to the 3-Regular Chromatic Index Problem (\ie, given a 3-regular graph, decide whether $\chi_1(G)=3$), we can apply $A$ to $(V,E,g)$, setting $g(e) = 1$ for all $e\in E$. 
	By \wref{pro:unweighted}, $A$ finds an edge colouring with $c \le (\frac{\Delta+1}{\Delta}-\varepsilon) \cdot \chi_1(G)$ colours.
	If $\chi_1(G) = \Delta$, then $c \le \Delta+1 - \varepsilon \Delta < \Delta + 1$, so $c=\Delta$; if $\chi_1(G) = \Delta+1$, then $c\ge \Delta+1$.
	Comparing $c$ to $\Delta$ thus determines $\chi_1(G)$ exactly in polynomial time; in particular, for every 3-regular graph, this decides whether $\chi_1(G)=3$.
	Since 3-Regular Chromatic Index is NP-complete, it follows that P $=$ NP.
\end{proofof}

We close this section with the remark that there are weighted DPS instances where any feasible schedule must ``multi-colour'' some edges, including the polycule shown in  \wref{fig:pentagon}.
For the general problem, we thus cannot restrict our attention to edge colourings (though they may be a valuable tool for future work).

\begin{figure}[hbt]
    \centering
    \begin{tikzpicture}[
    	scale=.8,
    	edge descriptor/.append style={inner sep=1pt},
    ]
        \node[graph node] (A) at (0, 6.87) {A};
        \node[graph node] (B) at (2.43,5.1) {B};
        \node[graph node] (C) at (1.5,2.25) {C};
        \node[graph node] (D) at (-1.5,2.25) {D};
        \node[graph node] (E) at (-2.43,5.1) {E};       
    
        \path[-] (A) edge[std, C1] node[edge descriptor] {$3$} (B);
        \path[-] (B) edge[std, C2] node[edge descriptor] {$3$} (C);
        \path[-] (E) edge[std, C2] node[edge descriptor] {$3$} (D);
        \path[-] (A) edge[std, C3] node[edge descriptor] {$3$} (E);

        \path[-] (C) edge[2 colour 1, C1] node[edge descriptor] {$2$} (D);
        \path[-] (C) edge[2 colour 2, C3] node[edge descriptor] {$2$} (D);
    \end{tikzpicture}
    \caption{A discrete polyamorous scheduling instance which is solvable only by assigning multiple colours to the CD edge}
    \label{fig:pentagon}
\end{figure}
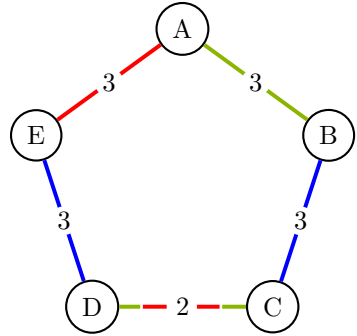

\section{Approximation Algorithms}
\label{sec:approximations}

In this section, we present two efficient polynomial-time approximation algorithms for Poly Scheduling,
thereby proving \wref[Theorems]{thm:color-approx} and~\ref{thm:layer-approx}.
Throughout this section, we assume a fixed instance $\mathcal P_o = (P,R,g)$ of Optimisation Polyamorous Scheduling (OPS) is given.

\subsection{Lower Bounds}

We first collect a few simple lower bounds used in the analysis later; 
note that \wref{sec:fractional} has further lower bounds.

\begin{lemma}[Simple lower bound]
\label{lem:lower-bound-trivial}
    Given an OPS instance $\mathcal P_o = (P,R,g)$, set $g_{\min} =\min_{e\in R} g(e)$,  $g_{\max}=\max_{e\in R} g(e)$,
    and $\Delta = \max_{p\in P} \deg(p)$.
    Any periodic schedule for $\mathcal P_o$ has heat $h\ge \max \{\Delta \cdot g_{\min},\, g_{\max}\}$.
\end{lemma}

\begin{proof}
The chromatic number $\chi_1$ of the unweighted graph $(P,R)$ is $\chi_1 \in \{\Delta, \Delta+1\}$.
This means that under any periodic schedule, some edge desires will grow to at least to $\chi_1 \cdot g_{\min} \ge \Delta \cdot g_{\min}$, since we cannot schedule any two edges incident to a degree-$\Delta$ node on the same day.
Moreover, we cannot prevent the weight-$g_{\max}$ edge from growing to heat~$g_{\max}$.
\end{proof}

A second observation is that the lower bound for any subset of the problem is also a lower bound for the problem as a whole:

\begin{lemma}[Subset bound]
\label{lem:lower-bound-subset}
	Given two OPS instances $\mathcal P_o = (P,R,g)$ and $\mathcal P'_o = (P,R',g')$ with $R' \subseteq R$ and $g(e) = g'(e)$ for all $e\in R'$, \ie,
	$\mathcal P'_o$ results from $\mathcal P_o$ by dropping some edges.
	Assume further that any schedule for $\mathcal P'_o$ has heat at least $h^*$.
	Then, any schedule for $\mathcal P_o$ also has heat at least $h^*$.
\end{lemma}
\begin{proof}
Suppose there is a schedule $S$ for $\mathcal P_o$ of heat $h'< h$.
We obtain a schedule $S'$ for $\mathcal P'_o$ by dropping all edges $e\notin R'$.
(The resulting schedule may have empty days.)
By construction, when using $S'$ to schedule $\mathcal P'_o$, all edges in $R'$ will grow to the same heat as in $\mathcal P_o$ under $S$, and hence also to heat $h'<h$.
\end{proof}

\subsection{Approximation for Almost Equal Growth Rates}
\label{sec:approx-similar-growth-rates}

We first focus on a special case of OPS instances with ``almost equal weights'', which is used as base for our main algorithm. Let the edge weights satisfy $g_{\min} \le g(e) \le g_{\max}$ for all $e\in R$.
We will show that scheduling a proper edge colouring round-robin gives a $\frac{\Delta+1}{\Delta}\cdot\frac{g_{\max}}{g_{\min}}$ approximation algorithm, establishing \wref{thm:color-approx}.

\begin{proofof}{\wpref{thm:color-approx}}
We compute a proper edge colouring for $(P,R)$ with $\Delta+1$ colours using the algorithm from~\cite{deltaPlusOneEdgeColouringAlg} and schedule these $\Delta+1$ matchings in a round-robin schedule.  
No edge desire will grow higher than $(\Delta+1)\cdot g_{\max}$ in this schedule.
\wref{lem:lower-bound-trivial} shows that $\mathrm{OPT} \ge \max\{\Delta \cdot g_{\min}, g_{\max}\}$.
The edge-colouring schedule is thus never more than a 
$\min\{ \frac{\Delta+1}{\Delta}\cdot\frac{g_{\max}}{g_{\min}} ,\; \Delta + 1 \}$ 
factor worse than OPT.
\end{proofof}

\subsection{Layering Algorithm}

The colouring-based algorithm from \wref{thm:color-approx} can be arbitrarily bad if desire growth rates are vastly different and $\Delta$ is large. 
For these cases, a more sophisticated algorithm achieves a much better guarantee (\wref{thm:layer-approx}).
The algorithm consists of 3 steps:
\begin{enumerate}
\item breaking the graph into layers (by edge growth rates), 
\item solving each layer using \wref{thm:color-approx}, and 
\item interleaving the layer schedules into an overall schedule.
\end{enumerate}

Let $L$ be a parameter to be chosen later.
We define layers of $\mathcal P_o = (P,R,g)$ as follows.
For $i = 0,\ldots, L-1$, set $\mathcal P_i = (P,R_i,g)$ where 
\[
    R_i \wrel= \biggl\{  e\in R \rel: \frac{g_{\max}}{2^{i+1}} < g(e) \le \frac{g_{\max}}{2^i} \biggr\}.
\]
Moreover, $\mathcal P_L = (P,R_L,g)$ with 
$R_L = \bigl\{  e\in R : g(e) \le \frac{g_{\max}}{2^L} \bigr\}$.

Denote by $\Delta_i$, for $i=0,\ldots,L$, the maximal degree in $(P,R_i)$. 
Let $S_i$ be the round-robin-$(\Delta_i+1)$-colouring schedule from \wref{thm:color-approx} applied on the OPS instance $\mathcal P_i$.
If run in isolation on $\mathcal P_i$, schedule $S_i$ has heat $h_i \le (\Delta_i+1)g_{\max}/2^{i} \le (\Delta+1)g_{\max}/2^i$ by the same argument as in \wref{sec:approx-similar-growth-rates}.
Moreover, for $i<L$, $S_i$ is a $2 \frac{\Delta_i+1}{\Delta_i}$-approximation (on $\mathcal P_i$ in isolation);
for $i=L$, we can only guarantee a $(\Delta_L+1)$-approximation.

To obtain an overall schedule $S$ for $\mathcal P$, we schedule the $L+1$ layers in round-robin fashion, and within each layer's allocated days, we advance through its schedule as before, \ie,
$S(t) = S_{(t \bmod (L+1))} \bigl(\lfloor t/(L+1) \rfloor\bigr)$.
Any advance in layer $i$ is now delayed by a factor $(L+1)$. 
Hence $S$ achieves heat at most
\[
		\overline h 
	\wwrel= 
		\max_{i\in[0..L]} (L+1)\cdot h_i
	\wwrel\le
		\max_{i\in[0..L]} (L+1) (\Delta_i+1) \cdot \frac{g_{\max}}{2^{i}}
\]
Using \wref{lem:lower-bound-subset} on the layers and \wref{lem:lower-bound-trivial}, 
we obtain a lower bound for OPT of
\[
    \underline h \wwrel= 
    	\max \left\{ 
    		\max_{i\in[0..L-1]} \Delta_i \cdot \frac{g_{\max}}{2^{i+1}},\;
    		g_{\max}
    		\right\}
\]

We now distinguish two cases for whether the maximum in $\overline h$ is attained for an $i < L$ or for $i=L$.
First suppose $\overline h = (L+1)(\Delta_i+1) g_{\max}/2^i$ for some $i < L$.
Since we also have $\underline h \ge \Delta_i \cdot g_{\max}/2^{i+1}$,
we obtain an approximation ratio of $2(L+1)\frac{\Delta_i+1}{\Delta_i}\le 3(L+1)$ overall in this case.

For the other case, namely $\overline h = (L+1)(\Delta_L+1)g_{\max}/2^L > (L+1)\cdot (\Delta_i+1) g_{\max}/2^i$ for all $i<L$,
we do not have lower bounds on the edge growth rates.
But we still know $\underline h \ge g_{\max}$, so we obtain a $(L+1)(\Delta_L+1)/2^L$-approximation
overall in this case.

Equating the two approximation ratios (and using $\Delta_L \le \Delta$) suggests to choose
$L = \lg(\Delta_L+1) - \lg 3$, which gives an overall approximation ratio
of 
$3 \lg(\Delta_L+1) - 3\lg(3/2) \le 3 \lg(\Delta+1) \le  3\lg n$.
This concludes the proof of \wref{thm:layer-approx}.

\section{Fractional Poly Scheduling}
\label{sec:fractional}

In this section, we generalize the notion of density from Pinwheel Scheduling for the Polyamorous Scheduling Problem. 
For that, we consider the dual of the linear program corresponding to a fractional variant of Poly Scheduling.

\subsection{Linear Programs for Poly Scheduling}

In the fractional Poly Scheduling problem, instead of committing to a single matching $M$ in $(P,R)$ each day, we are allowed to devote an arbitrary \emph{fraction} $y_M \in [0,1]$ of our day to $M$, but then switch to other matchings without cost or delay for the rest of the day (a simple form of scheduling with preemption).
The heat of a fractional schedule is again defined as $\max_{e\in R} r(e) g(e)$, but the recurrence time $r(e)$ now is the maximal time in $S$ before the fraction of days devoted to matchings containing $e$ sum to at least~1. (For a non-preemptive schedule with one matching per day, this coincides with the definition from \wref{sec:preliminaries}.)

Schedules for the fractional problem are substantially easier because there is no need to have different fractions $y_M$ for different days: the schedule obtained by always using the average fraction of time spent on each matching yields the same recurrence times.
We can therefore assume without loss of generality that our schedule is given by $S = S(\{y_M\}_{M\in \mathcal M})$, with $y_M \in [0,1]$ and $\sum_{M\in \mathcal M} y_M = 1$. $S$ schedules the matchings in some arbitrary fixed order, each day devoting the same $y_M$ fraction of the day to $M$.
Then, recurrence times are simply given by $r_{S}(e) = 1\big/ \sum_{M\in\mathcal M : e\in M} y_M$.

With these simplifications, we can state the fractional relaxation of Optimisation Poly Scheduling instance $\mathcal P_o = (P,R,g)$
as an optimisation problem as follows:
\begin{alignat}{4}
	&\min \quad \mathrlap{\bar h} \\
	&\text{s.\,t. } 
	&	\sum_{{M \in \mathcal M}} y_M &\wrel\le 1 \\
	&&	\frac1{\sum_{M \in \mathcal M: e \in M} y_M} \cdot g_e &\wrel\le \bar h  &\qquad& \forall e\in R\\
	&&	y_M & \wrel\in [0,1] &\qquad& \forall M \in \mathcal M
\end{alignat}
Substituting $\bar h=1/\ell$, this is equivalent to the following linear program (LP):
\begin{alignat}{4}
\label{eq:fractional-LP}
	&\max \quad \mathrlap{\ell} \\
	&\text{s.\,t. } 
	&	\sum_{{M \in \mathcal M}} y_M &\wrel\le 1 \\
	&&	\frac1{g_e} \sum_{M \in \mathcal M: e \in M} y_M &\wrel\ge \ell  &\qquad& \forall e\in R\\
	&&	y_M & \wrel\ge 0 &\qquad& \forall M \in \mathcal M
\end{alignat}
The optimal objective value $\ell^*$ of this LP gives $\bar h^* = 1/\ell^*$, the optimal fractional heat.

\begin{lemma}[Fractional lower bound]
\label{lem:fractional-lower-bound}
	Consider an OPS instance $\mathcal P_o = (P,R,g)$ with optimal heat $h^*$ and let $\bar h^* = 1/\ell^*$ where $\ell^*$ is the optimal objective value of the fractional-problem LP from \weqref{eq:fractional-LP}. 
	Then $\bar h^* \le h^*$.
\end{lemma}
\begin{proof}
We use the same approach as in~\cite[\S3]{CiceroneDiStefanoGasieniecJurdzinskiNavarraRadzikStachowiak2019}:
For any schedule $S$, $h(S)$ is at least the heat $h_T(S)$ obtained during the first $T$ days only, which in turn is at least $\max_e g(e)\cdot \bar r(e)$ for $\bar r(e)$ the \emph{average} recurrence time of edge $e$ during the first $T$ days.
A basic calculation shows that 
for the fractions $y_M$ of time spent on matching $M$ during the first $T$ days 
there exists a value $1/\ell = h(S)(1-o(T))$, so that we obtain a feasible solution of the LP~\eqref{eq:fractional-LP}. Hence $1/\ell^* \le 1/\ell = h(S)(1-o(T))$.
Since these inequalities hold simultaneously for all $T$,
taking the limit as $T\to\infty$, we obtain $1/\ell^* = \bar h^* \le h(S)$.
\end{proof}

The immediate usefulness of \wref{lem:fractional-lower-bound} is limited since the number of matchings can be exponential in $n$.

\begin{remark}[Randomized-rounding approximation?]
One could try to use this LP as the basis of a randomized-rounding approximation algorithm,
but since it is not clear how to obtain an efficient algorithm from that, we do not pursue this route here.
The simple route taken in~\cite{CiceroneDiStefanoGasieniecJurdzinskiNavarraRadzikStachowiak2019} 
cannot achieve an approximation ratio better than $O(\log n)$, so \wref{thm:layer-approx} already provides an equally good deterministic algorithm.
\end{remark}

We therefore proceed to the dual LP of \weqref{eq:fractional-LP}:
\begin{alignat}{4}
\label{eq:dual-LP}
	&\min \mathrlap{\quad x} \\
	&\text{s.\,t. } 
	&	\sum_{e\in R} z_e &\wrel\ge 1 \\
	&&	\sum_{e\in M} \frac{z_e}{g_e}  &\wrel\le x &\quad& \forall M \in \mathcal M\\
	&&	z_e & \wrel\ge 0 &\quad& \forall e\in R
\end{alignat}
While still exponentially large and thus not easy to solve exactly, the dual LP yields the versatile
result from \wref{thm:dual-lp-bound}.

\begin{proofof}{\wpref{thm:dual-lp-bound}}
	Using the given $z_e$ and $x = \max_{M\in \mathcal M} \sum_{e\in M} \frac{z_e}{g(e)}$,
	we fulfil all constraints of \weqref{eq:dual-LP}.
	The optimal objective value $x^*$ is hence $x^* \le x$. By the duality of LPs, we have $x^* \ge \ell^*$
	for $\ell^*$ the optimal objective value of \weqref{eq:fractional-LP}.
	Together with \wref{lem:fractional-lower-bound}, this means
	$h^* \ge \bar h^* = 1/\ell^* \ge 1/x^* \ge 1/x$.
\end{proofof}

\subsection{Poly Density}

\wref{thm:dual-lp-bound} gives a more explicit way to compute the \emph{poly density} $\bar h^*$
than the primal LP, but it is unclear whether it can be computed exactly in polynomial time.
Given the more intricate global structure of Poly Scheduling, $\bar h^*$ is necessarily more
complicated than the density of Pinwheel Scheduling.
The most interest open problem for Poly Scheduling hence is whether a sufficiently small poly density
implies the existence of a valid (integral) schedule.

Specific choices for $z_e$ in \wref{thm:dual-lp-bound} yield several known bounds:
\begin{itemize}
\item 
	Setting $z_e = g_e / G$ for $G = \sum_{e\in R} g_e$ 
	yields \wref{cor:total-growth-bound}.
\item
	Fix any subset $R' \subseteq R$.
	Now set $z_e = g_e/C$ if $e\in R'$ and $0$ otherwise, where $C = \sum_{e\in R'} g_e$.
	The maximum from \wref{thm:dual-lp-bound} then simplifies to $\frac1C \max_{M\in \mathcal M} |M\cap R'|$,
	so \[h^* \ge \frac{\sum_{e\in R'} g_e}{\max_{M\in \mathcal M} |M\cap R'|}.\]

\item An immediate application of that observation with $R'$ all edges incident at a person $p\in P$ yields
	the BGT bound:
\begin{corollary}[Bamboo lower bound]
\label{cor:lower-bound-bamboo}
	Given an OPS instance $(P,R,g)$ and $p\in P$ with $g_1\ge \cdots \ge g_d$ the desire growth rates
	for edges incident at $p$. Set $G_p = g_1 + \cdots + g_d$.
	Any periodic schedule for $(P,R,g)$ has heat at least $G_p$.
\end{corollary}
\end{itemize}

\begin{remark}[Better general bounds?]
For the general case, it seems challenging to obtain other such simple bounds. 
The bound of $G/m$ is easy to justify without the linear programs by a ``preservation-of-mass argument'':
Assume a schedule $S$ could achieve a heat $h < G/m$. Every day, the overall polycule's desire grows by $G$, and $S$ can schedule at most $m$ pairs to meet, whose desire is reset to $0$ from some value $\le h$.
Every day, $S$ thus removes only a total of $\le mh < G$ desire units from the polycule, whereas the overall growth is $G$, a contradiction to the heat remaining bounded.

Note that the bound of $G/m$ is tight for some instances, so we cannot hope for a strictly lower bound. On the other hand, the example from \wref{fig:very-bad-mass-preservation} shows demonstrates that it can also be arbitrarily far from $h^*$.
\end{remark}

\begin{figure}
	\begin{tikzpicture}[
			person/.style={circle,draw,inner sep=2.7pt},
			relation/.style={thick},
			growth label/.style={inner sep=1pt,fill=white,text=red!80!black,font=\itshape},
			freq label/.style={inner sep=1pt,fill=white,text=blue!80!black,font=\bfseries\boldmath},
			scale=1.9,
		]
		\node[person] (A) at (-1.2,0) {$A$} ;
		\node[person] (B) at (-.6,1) {$B$} ;
		\node[person] (C) at (0,0) {$C$} ;
		\begin{scope}[shift={(0.5,0)}]
		\foreach \i in {1,3,...,9} {
			\node[person] (T\i) at (\i/2.4,1) {$T_{\i}$} ;
		}
		\foreach \i in {2,4,...,9} {
			\node[person] (T\i) at (\i/2.4,0) {$T_{\i}$} ;
		}
		\end{scope}
		\foreach \a/\b/\f in {%
				A/B/2,%
				B/C/3,%
				A/C/3,%
				C/T1/$F$,%
				T1/T2/$F$,%
				T2/T3/$F$,%
				T3/T4/$F$,%
				T4/T5/$F$,%
				T5/T6/$F$,%
				T6/T7/$F$,%
				T7/T8/$F$,%
				T8/T9/$F$%
		} {
			\draw[relation] (\a) to node[freq label] {\f} (\b) ;
		}
		\node[scale=2] at (4.75,.45) {\dots};
		\node[person] (Tk) at (5.3,0)  {$T_k$};
	\end{tikzpicture}
	\caption{
		The tadpole family of DPS instances, defined for parameters $k\ge 0$ (tail length) and $F\ge 3$ (tail frequency).
		The total growth rate is $G = \frac12+\frac23+k\cdot 1/F = \frac{7}6 + \frac kF$ and the size of a maximum matching is$m = 1+\lfloor (k+1)/2\rfloor$.
	}
	\label{fig:very-bad-mass-preservation}
\end{figure}
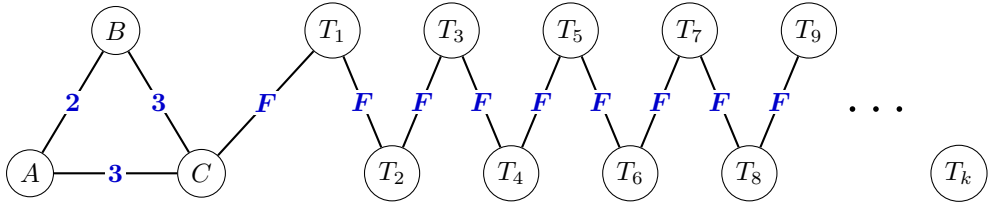

\wref{fig:very-bad-mass-preservation} shows the tadpole family of instances demonstrating the power of the dual-LP approach and \wref{thm:dual-lp-bound}.
All DPS tadpoles (as shown in the figure) are infeasible since already the triangle $A{-}B{-}C$ does not admit a schedule obeying the given frequencies.
The corresponding OPS instances (as given by \wref{lem:dps-to-ops}) with $g(e) = 1/f(e)$ thus have $h^* > 1$; indeed $h^*=4/3$ if $F\ge 2$.
However, the simple lower bounds or local arguments do not detect this:
(a)~All local Pinwheel Scheduling instances (any person plus their neighbours) are feasible.
(b)~The mass-preservation bound (\wref{cor:total-growth-bound}) is $G/m < 1$ for $k\ge 1$.
Indeed, setting $F=k$ and letting $k\to \infty$, $G/m = O(1/k)$, giving an arbitrarily large gap to $h^*$.
By contrast, consider the LP fractional lower bound. One can show that $\bar h^* = \frac76> 1$ for any $k\ge 1$ and $F\ge 2$, so \wref{thm:dual-lp-bound} correctly detects the infeasibility in this example.

\begin{remark}[Better Pinwheel density via dual LPs?]
	Since Poly Scheduling is a generalization of Pinwheel Scheduling resp.\ Bamboo Garden Trimming,
	we can apply \wref{thm:dual-lp-bound} also to these problems.
	However, for this special case, the optimal objective value of the dual LPs is \emph{always} $x^* = \ell^* = 1/G$ for $G$ the $G$ the sum of the growth rates, so we only obtain the trivial ``biomass'' lower bound of $G$ for Bamboo Garden Trimming resp.\ the density $\le 1$ necessary condition for Pinwheel Scheduling.
	The more complicated structure of matchings in non-star graphs makes fractional lower bounds in Poly Scheduling much richer and more powerful.
\end{remark}

\section{Open Problems \& Future Directions}
This paper opens up several avenues for future work. 
The most obvious open problem concerns efficient approximation algorithms: 
can we reduce the gap between our $\frac{13}{12}$
hardness of approximation lower bound and the $O(\log n)$ upper bound? 
As introduced by \wref{conj:4/3}, we expect that future works may demonstrate better inapproximability
results for OPS in the general case;
the true lower bound may even be super-constant.
However, in light of our \wref{thm:layer-approx}, 
a super-constant hardness of approximation result would have to use Poly Scheduling instances with super-constant degrees.
\wref{open:poly-density} will also have clear implications for OPS, as well as being interesting in its own right.

There is also interesting work to be done looking at specific classes of polycules. Bipartite
polycules are particularly interesting, both for the likelihood that they will permit better approximations than are possible in the general case and for their applications (\eg, modelling the users and providers of some service).

Polyamorous scheduling has several interesting generalizations including Fungible Poly\-amorous Scheduling, whose decision version we define as:
\begin{definition}[Fungible Decision Polyamorous Scheduling (FDPS)]
\label{def:FDPS}
	An  FDPS instance $\mathcal P_{fd} = (P, R, s, f)$ 
    (a ``(fungible decision) \emph{polycule}'') consists of
	an undirected graph $(P,R)$ where the 
	vertices $P = \{p_1, \ldots,p_n\}$ are $n$ \emph{classes} of fungible persons and the 
	edges $R$ are pairwise relationships between those classes. Classes have integer \emph{sizes} $s : P \to \N$  and relationships have integer \emph{frequencies} $f : R \to \N$.
	
	The goal is find an infinite schedule $S : \N_0 \to 2^R$, such that 
	\begin{thmenumerate}{def:FDPS}
	\item \label{def:DPS-no-conflicts-fungible} 
		(no overflows) for all days $t\in \N_0$, $S(t)$ is a multiset with elements from $P$ such that each node $p_i\in P$ appears at most $s(p)$ times, and 
	\item \label{def:DPS-frequencies-fungible} 
		(frequencies) for all $e \in R$ and $t\in \N_0$, we have 
		$e \in S(t)\cup S(t+1)\cup \cdots \cup S(t+f(e)-1)$;
	\end{thmenumerate}
	or to report that no such schedule exists.
\end{definition}

FDPS also has an optimisation version, which again allows each person $p\in P$ to have at most $s(p)$ meetings each day. These problems have clear applications to the scheduling of staff and resources.

Another natural generalisation is Secure Polyamorous scheduling. Suppose that Adam is dating both Brady and Charlie, who are also dating each other. In a DPS or OPS polycule, on any day, Adam must choose to meet with either Brady or Charlie, who each face similar choices; but why can't he meet both?\footnote{A key part of polyamory~\cite{why_not_both}!}
The Secure Decision Polyamorous scheduling problem allows this:

\begin{definition}[Secure Decision Polyamorous Scheduling (SDPS)]
\label{def:SDPS}
	An SDPS instance $\mathcal P_{sd} = (P, R, f)$ (a ``(secure decision) \emph{polycule}'') consists of
	an undirected graph $(P,R)$ where the 
	vertices $P = \{p_1, \ldots,p_n\}$ are $n$ \emph{persons}, and the 
	edges $R$ are pairwise relationships, with integer \emph{frequencies} $f : R \to \N$ for each relationship.
 
	The goal is find an infinite schedule $S : \N_0 \to 2^R$, such that 
	\begin{thmenumerate}{def:DPS}
	\item \label{def:SDPS-no-third-wheels} 
		(no third-wheels) for all days $t\in \N_0$, $S(t)$ is a set of meetings between cliques of people in  $P$ in which each person appears at most once, and
	\item \label{def:SDPS-frequencies} 
		(frequencies) for all $e \in R$ and $t\in \N_0$, we have 
		$e \in S(t)\cup S(t+1)\cup \cdots \cup S(t+f(e)-1)$;
	\end{thmenumerate}
	or to report that no such schedule exists.
\end{definition}
Again, this has a natural optimisation version.

Polyamorous Scheduling also motivates the study of several restricted versions of Pinwheel Scheduling and Bamboo Garden Trimming, including partial scheduling (wherein some portion of the schedule is fixed as part of the problem and the challenge is to find the remainder of the schedule), and fixed holidays (where the fixed part of the schedule consists of periodic gaps).

	\myacknowledgements
%


\bibliography{references}

\ifdraft{
	\clearpage
	\appendix
	\ifkoma{\addpart{Appendix}}{}
	\clearpage
	\part*{Notes-to-self}
	\printnotestoself
}{}

\end{document}